\documentclass[letterpaper]{article}
\usepackage{uai2018}

\usepackage[margin=1in]{geometry}

\usepackage{times}

\title{Differential Analysis of Directed Networks}

\author{ {\bf Min Ren} \\
Department of Statistics \\
Purdue University \\
West Lafayette, IN\\
ren80@purdue.edu
\And
{\bf Dabao Zhang}  \\
Department of Statistics\\
Purdue University\\
West Lafayette, IN\\
zhangdb@purdue.edu
}

\usepackage{microtype}
\usepackage{graphicx}
\usepackage[FIGTOPCAP]{subfigure}
\usepackage{algorithm}
\usepackage{algorithmic}
\usepackage{booktabs} 

\usepackage{amsmath}
\usepackage{amsthm}
\usepackage{amsfonts}
\usepackage{amssymb}
\usepackage{scalerel}
\usepackage[authoryear]{natbib}
\usepackage{float}
\usepackage{mathrsfs}
\usepackage{mathtools}
\usepackage[english]{babel}
\usepackage{mathrsfs}
\usepackage{placeins}
\newtheorem{theorem}{Theorem}
\newtheorem{lemma}[theorem]{Lemma}

\newtheorem{definition}{Definition}[section]

\newcommand{\RDnet}{\textbf{ReDNet}}

\newcommand{\frobn}[1]{||#1||_{F}}
\newcommand{\ltwon}[1]{||#1||_{2}}
\newcommand{\lonen}[1]{||#1||_{1}}

\newcommand{\infn}[1]{||#1||_{\infty}}
\newcommand{\opn}[1]{||#1||_{op}}

\newcommand{\bX}{\mbox{\bf X}}

\newcommand{\bY}{\mbox{\bf Y}}

\newcommand{\bigO}[1]{\mbox{$\mathcal{O}(#1)$}}

\newcommand{\Ex}[1]{\mbox{$\mathbb{E}(#1)$}}

\newcommand{\bweight}{\mbox{\boldmath $\omega$}}

\newcommand{\phivarmin}[3][]{\mbox{$\phi^{#1}_{\text{re}}(#2,#3)$}}

\newcommand{\sto}[1][]{\mbox{$\mathcal{I}^{#1}_{i,21}$}}
\newcommand{\soo}[1][]{\mbox{$\mathcal{I}^{#1}_{i,11}$}}
\newcommand{\hsto}[1][]{\mbox{$\hat{\mathcal{I}}^{#1}_{i,21}$}}
\newcommand{\hsoo}[1][]{\mbox{$\hat{\mathcal{I}}^{#1}_{i,11}$}}

\newcommand{\var}{\mathrm{var}}

\newcommand{\tr}{\mathrm{tr}}
\newcommand{\diag}{\mathrm{diag}}

\usepackage{enumitem}
\usepackage{chngcntr,etoolbox}
\newcounter{resetdummycounter}
\newcommand{\resetcounterlist}[1]{%
	\renewcommand*{\do}[1]{\counterwithin*{##1}{resetdummycounter}}%
	\docsvlist{#1}}
\newcommand{\resetcounters}{\stepcounter{resetdummycounter}}
\resetcounterlist{theorem,section,equation}

\begin{document}
\maketitle

\begin{abstract}
We developed a novel statistical method to identify structural differences between networks characterized by structural equation models. We propose to reparameterize the model to separate the differential structures from common structures, and then design an algorithm with calibration and construction stages to identify these differential structures. The calibration stage serves to obtain consistent prediction by building the $\ell_2$ regularized regression of each endogenous variables against pre-screened exogenous variables, correcting for potential endogeneity issue. The construction stage consistently selects and estimates both common and differential effects by undertaking $\ell_1$ regularized regression of each endogenous variable against the predicts of other endogenous variables as well as its anchoring exogenous variables. Our method allows easy parallel computation at each stage. Theoretical results are obtained to establish non-asymptotic error bounds of predictions and estimates at both stages, as well as the consistency of identified common and differential effects. Our studies on synthetic data demonstrated that our proposed method performed much better than independently constructing the networks. A real data set is analyzed to illustrate the applicability of our method.
\end{abstract}

\section{ \uppercase{INTRODUCTION}}

It is of great importance and interest to detect sparse structural differences or differential structures between two cognate networks. For instance, the gene regulatory networks of diseased and healthy individuals may differ slightly from each other \citep{west2012differential}, and identifying the subtle difference between them helps design specific drugs. Social networks evolve over times, and monitoring their abrupt changes may serve as surveillance to economic stability or disease epidemics \citep{pianese2013discovering, berkman1979social}. However, addressing such practical problems demands differential analysis of large networks, calling for development of efficient statistical method to infer and compare complex structures from high dimensional data. In this paper, we focus on differential analysis of directed acyclic or even cyclic networks which can be described by structural equation models (SEMs).

Many efforts have been made towards construction of a single network via SEM. For example, both \citet{xiong2004identification} and \citet{liu2008gene} employed genetic algorithms to search for the best SEM. Most recently, \citet{ni2017reciprocal,ni2017heterogeneous} employed a hierarchical Bayes approach to construct SEM-based networks. However, these approaches were designed for small or medium scale networks. For large-scale networks whose number of endogenous variables $p$ exceeds the sample size $n$, \citet{cai2013inference} proposed a regularization approach to fit a sparse model. Because this method suffers from incapability of parallel computation, it may not be feasible for large networks. \citet{logsdon2010gene} proposed another penalization approach to fit the model in a node-wise fashion which alleviates the computational burden. Most recently, \citet{lin2015regularization}, \citet{zhu2017sparse}, and \citet{chen2015two} each proposed a two-stage approach to construct SEMs, with different algorithms designed at different stages. As shown by \citet{chen2015two}, such a two-stage approach can have superior performance compared to other methods.

To the best of our knowledge, no algorithm has been proposed to conduct differential analysis of directed networks characterized by SEM. While a naive approach would separately construct each individual network and identify common and differential structures, this approach fails to take advantage of the commonality as well as sparse differential structures of the paired networks, leading to higher false positive rate or lower power. In this light, we introduce a novel statistical method, specially in the directed network regime, to conduct differential analysis of two networks via appropriate reparameterization of the corresponding models. There are two major features of our method. Firstly, we jointly model the commonality and difference between two networks explicitly. This helps us to gain dramatic performance improvements over the naive construction method. Secondly, benefiting from the flexible framework of SEMs, we are able to conduct differential analysis of directed networks. Most importantly, our method allow for both acyclic and cyclic networks. Compared to the other methods, directionality and allowing for cyclicity are crucial for many network studies, especially in constructing gene regulatory networks. As far as we know, our method is the first work on differential analysis of directed networks that enjoys the two promising features.

The rest of this paper is organized as follows. We first introduce the model and its identifiability condition in Section 2.1 and Section 2.2, respectively. Then, we present our proposed method of \textbf{Re}parameterization-based \textbf{D}ifferential analysis of directed \textbf{Net}works, termed as \textbf{ReDNet}, in Section~\ref{sec:moderepar}. The theoretical justification of the proposed method is described in Section~\ref{sec:theoractical}. Section~\ref{sec:simulation} includes our studies on synthetic data showing the superior performance of our method, as well as an analysis of the Genotype-Tissue Expression (GTEx) data sets. We conclude our paper with brief discussion in Section~\ref{sec:discussion}.

\section{ \uppercase{Methods}}

Here we first introduce the model and its identification condition, and then describe our proposed $\mathbf{ReDNet}$ method for identifying common and differential structures between two directed networks, followed with its theoretical justification.

\subsection{\uppercase{The Model}}

We consider two networks, each describing the dependencies among a common set of variables or nodes in a unique population. For each  node $i\in \{1, 2, \ldots,p\}$ in network $k\in\{1, 2\}$, its regulation structure can be represented by the following equation,
\begin{equation}
\label{model:onenode}
\underbrace{\mathbf{Y}^{(k)}_{i}}_{\text{node } i } = \underbrace{\mathbf{Y}^{(k)}_{-i} \boldsymbol{\boldsymbol{\gamma}}^{(k)}_{i} }_{\text{regulation by others}} + \underbrace{\mathbf{X}^{(k)} \boldsymbol{\boldsymbol{\phi}}^{(k)}_{i}}_{\text{anchoring regulation}}+ \underbrace{\boldsymbol{\epsilon}^{(k)}_{i}}_{\text{error}},
\end{equation}
where $\mathbf{Y}^{(k)}_{i}$ is the $i$-th column of $\mathbf{Y}^{(k)}$ and $\mathbf{Y}^{(k)}_{-i}$ is the submatrix of $\mathbf{Y}^{(k)}$ by excluding $\mathbf{Y}^{(k)}_{i}$, with $\mathbf{Y}^{(k)}$ a $n^{(k)} \times p$ matrix. $\mathbf{X}^{(k)}$ is a $n^{(k)} \times q$ matrix with each column standardized to have $\ell_2$ norm $\sqrt{n^{(k)}}$. The vectors $\boldsymbol{\gamma}^{(k)}_{i}$ and $\boldsymbol{\phi}^{(k)}_{i}$ encode the inter-nodes and anchoring regulatory effects, respectively. The index set of non-zeros of $\boldsymbol{\phi}^{(k)}_{i}$ is known and denoted by $\mathcal{A}^{(k)}_i$, in other words, $\mathcal{A}^{(k)}_i = \text{supp}(\boldsymbol{\phi}^{(k)}_{i})$. The support set $\mathcal{A}^{(k)}_i$ indexes the direct causal effects for the $i$-th node, and can be prespecified based on the domain knowledge. However, the size of nonzero effect $\boldsymbol{\phi}^{(k)}_{i}$ is unknown and can be estimated. Further property of $\mathcal{A}^{(k)}_i$ will be discussed in Section~\ref{subsec:modelIden}. All elements of the error term are independently distributed following a normal distribution with mean zero and standard deviation $\sigma_i^{(k)}$. We assume that the matrix $\mathbf{X}^{(k)}$ and the error term $\boldsymbol{\epsilon}^{(k)}_{i}$ are independent of each other. However $\mathbf{Y}^{(k)}_{-i}$ and $\boldsymbol{\epsilon}^{(k)}_{i}$ may correlate with each other. $\mathbf{Y}^{(k)}$ and $\mathbf{X}^{(k)}$ include observed endogenous variables and exogenous variables, respectively.

By combining the $p$ linear equations in (\ref{model:onenode}), we can rewrite the two sets of linear equations in a systematic fashion as two structural equation models below,
\begin{equation}
\label{model:fullsystem}
\begin{cases}
\mathbf{Y}^{(1)} = \mathbf{Y}^{(1)} \mathbf{\Gamma}^{(1)} + \mathbf{X}^{(1)}\mathbf{\Phi}^{(1)}+ \mathcal{E}^{(1)}, \\
\mathbf{Y}^{(2)} = \mathbf{Y}^{(2)} \mathbf{\Gamma}^{(2)} + \mathbf{X}^{(2)}\mathbf{\Phi}^{(2)}+ \mathcal{E}^{(2)},
\end{cases}
\end{equation}
where each matrix $\mathbf{\Gamma}^{(k)}$ is $p \times p $ with zero diagonal elements and represents the inter-nodes regulatory effects in the corresponding network. Specifically, excluding the $i$-th element (which is zero) from the $i$-th column of $\mathbf{\Gamma}^{(k)}$ leads to  $\boldsymbol{\gamma}^{(k)}_{i}$. The $q \times p$ matrix $\mathbf{\Phi}^{(k)}$ contains the anchoring regulatory effects and its $i$-th column is $\boldsymbol{\phi}_{i}^{(k)}$. Each error term $\mathcal{E}^{(k)}$ is $n^{(k)} \times p $ and has the error term $\boldsymbol{\epsilon}^{(k)}_{i}$ as its $i$-th column.

Figure~\ref{fig:illustrate} gives an illustrative example of networks with three nodes and one anchoring regulation per node for the structural equations in (\ref{model:fullsystem}). For example, with anchoring regulation on node $Y_1$, $X_1$ has a direct effect on node $Y_1$ but indirect effects on node $Y_2$ and $Y_3$ via $Y_1$.

\begin{figure}[H]
	\centering
 \subfigure[Network I]{\label{fig:sub1}\includegraphics[width=1.15in,height=1.25in]{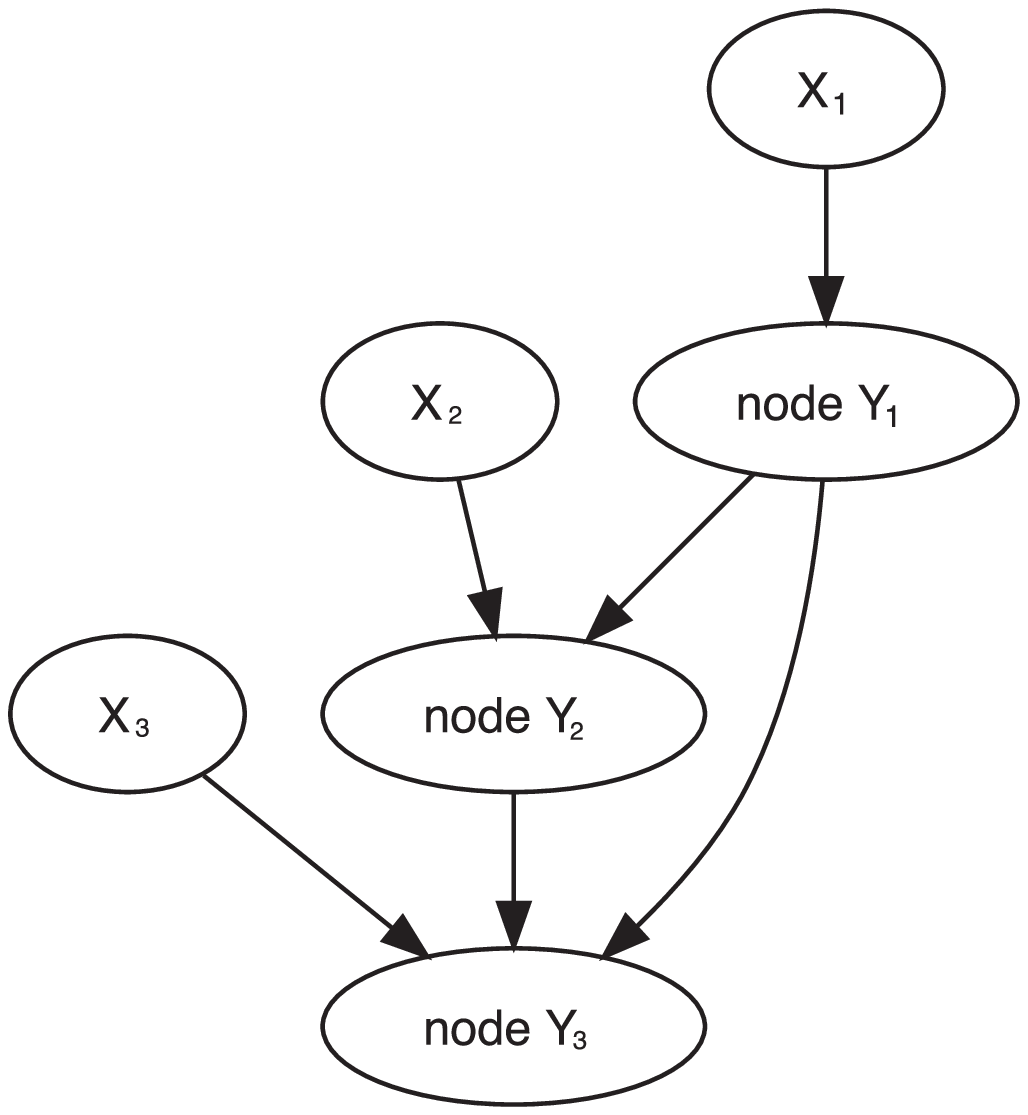}}
 \subfigure[Network II]{\label{fig:sub2}\includegraphics[width=1in,height=1.25in]{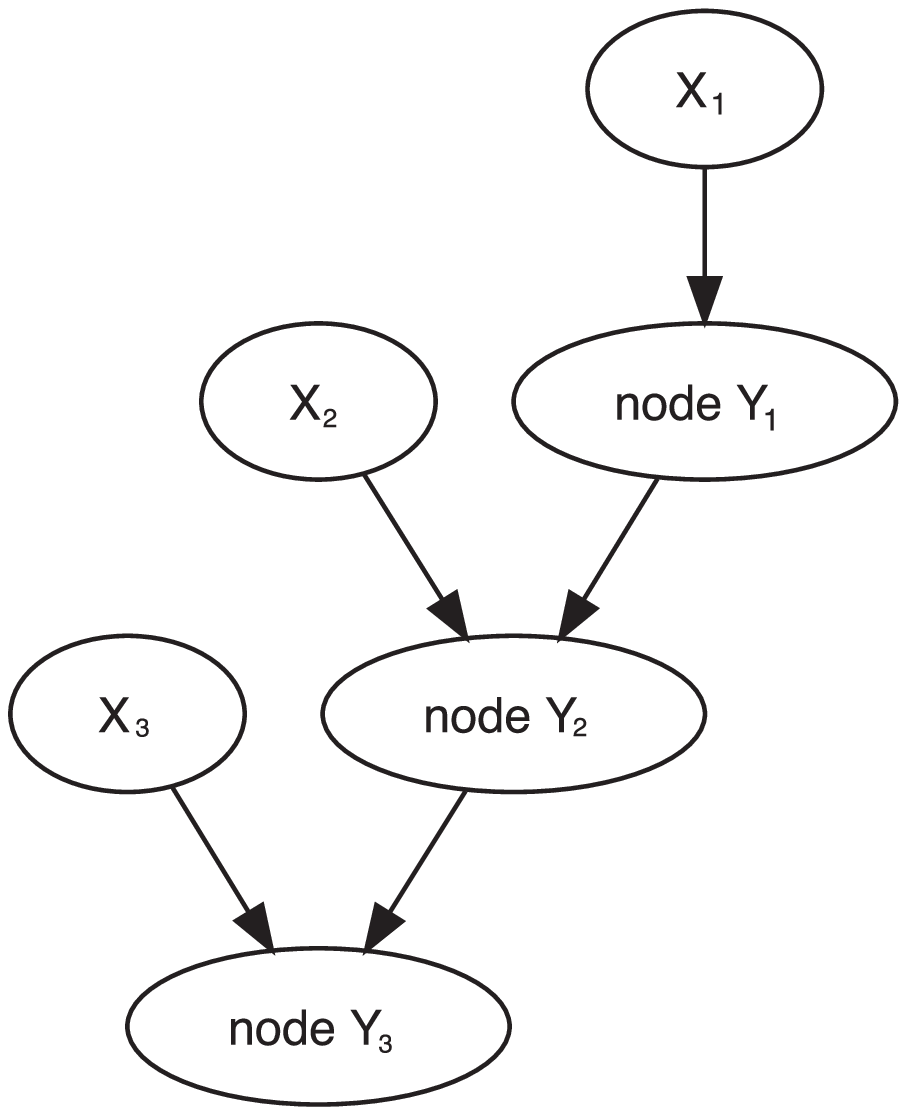}}
 \subfigure[Differential]{\label{fig:sub3}\includegraphics[width=0.85in,height=1.25in]{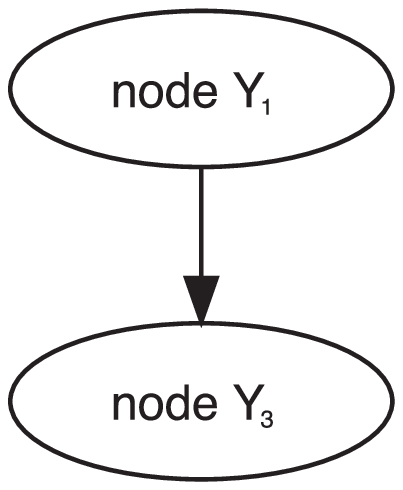}}
	\caption{An Illustrative Example of Differential Network Between Two Directed Networks. The error term for each node is not shown for simplicity.}
	\label{fig:illustrate}
\end{figure}

For each network $k$, its full model in (\ref{model:fullsystem}) can be further transformed into the reduced form as follows,
\begin{equation}\label{model:reduced}
 \mathbf{Y}^{(k)} = \mathbf{X}^{(k)} \boldsymbol{\pi}^{(k)} +\boldsymbol{\xi}^{(k)},
\end{equation}
where the $q \times p$ matrix $\boldsymbol{\pi}^{(k)} = \mathbf{\Phi}^{(k)}(\mathbf{I}-\mathbf{\Gamma}^{(k)})^{-1}$ and the transformed error term $\boldsymbol{\xi}^{(k)} = \mathcal{E}^{(k)}(\mathbf{I}-\mathbf{\Gamma}^{(k)})^{-1}$. The reduced model (\ref{model:reduced}) reveals variables observed in $\bX^{(k)}$ as instrumental variables which will be used later to correct for the endogeneity issue. Otherwise, directly applying any regularization based regression to equation (\ref{model:onenode}) will result in  non-consistent or suboptimal estimation of model parameters \citep{fan2014, chen2015two, lin2015regularization, zhu2017sparse}.

\subsection{\uppercase{The Model Identifiability} } \label{subsec:modelIden}

Here we introduce an identifiability assumption which helps to infer an identifiable system (\ref{model:fullsystem}) from available data. We assume that each endogenous variable is directly regulated by a unique set of exogenous variables as long as it regulates other endogenous variables. That is, any regulatory node needs at least one anchoring exogenous variable to distinguish the corresponding regulatory effects from association. Explicitly let $\mathcal{M}^{(k)}_{i0}$ denote the index set of endogenous variables which either directly or indirectly regulate the $i$-th endogenous variable in the $k$-th network. Thus, $\mathcal{A}_{i}^{(k)} \subseteq \mathcal{M}^{(k)}_{i0}$. The model identification condition can be stated in the below.

\textbf{Assumption 1.} For any $i=1,\cdots, p$, $\mathcal{A}_{i}^{(k)}\ne\emptyset$ if there exists $j$ such that $i\in\mathcal{M}^{(k)}_{j0}$. Furthermore, $\mathcal{A}_{i}^{(k)}\cap \mathcal{A}_{j}^{(k)}=\emptyset$ as long as $i\ne j$.

This assumption is slightly less restrictive than the one employed by \citet{chen2015two}, and is a sufficient condition for model identifiability as it satisfies the rank condition in \citet{Schmidt1976}. It can be further relaxed to allow nonempty $\mathcal{A}_{i}^{(k)}\cap \mathcal{A}_{j}^{(k)}$ as long as each regulatory node has its own unique anchoring exogenous variables.

The above identifiability assumption not only identifies $\boldsymbol{\gamma}^{(k)}_i$ in model~(\ref{model:onenode}) from $\boldsymbol{\pi}^{(k)}$ in model~(\ref{model:reduced}) but also helps reveal regulatory directionality of the networks. As illustrated in Figure~\ref{fig:illustrateMarkovEqu}, we can not recover the directionality between nodes $Y_1$ and $Y_2$ without the extra information provided by the direct causal factors $X_1$ and $X_2$ because all four sub-networks consisting of $Y_1$ and $Y_2$ (without $X_1$ and $X_2$) will be Markov equivalent. The known set $\mathcal{A}_{j}^{(k)}$ serves as external prior knowledge which helps recover the directionality. In our two-stage construction of the differential network, the additional anchors $X_1$ and $X_2$ serve as instrumental variables in the calibration stage, since both $X_1$ and $X_2$ are independent of the error terms. The present direct causal effects from $\bX^{(k)}$ together with Assumption~1 differentiates our approach from the classical graphical models \citep{meinshausen2006high, yuan2007model} or the PC algorithm approaches \citep{spirtes2000causation, kalisch2007estimating}, since those methods either cannot recover edge directions or do not allow for cyclic structures due to lack of additional direct causal effects from $\bX^{(k)}$.

\begin{figure}[ht]
	\centering
\subfigure[]{\label{fig:subnet1}\includegraphics[width=.2\linewidth]{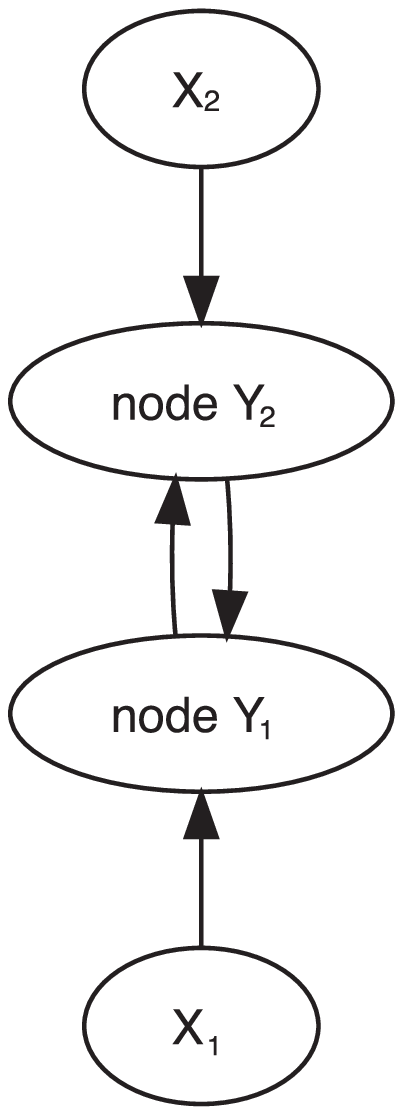}} \hspace{6pt}
\subfigure[]{\label{fig:subnet2}\includegraphics[width=.2\linewidth]{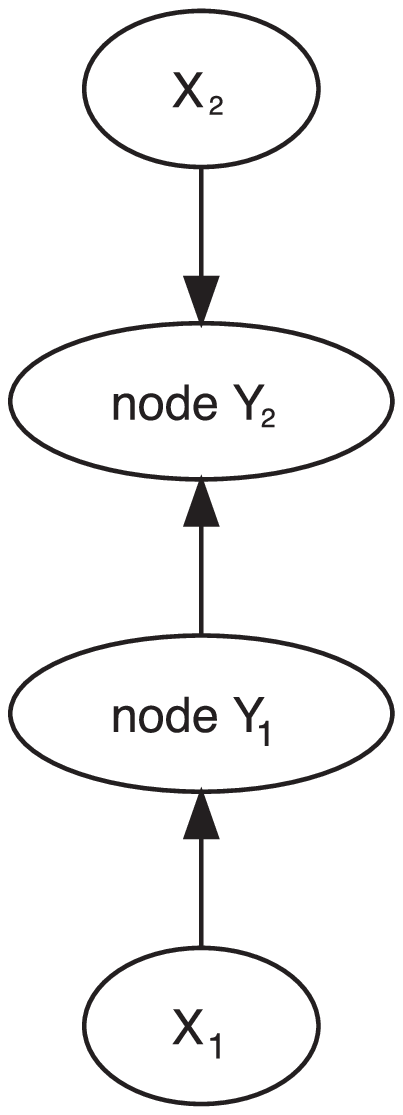}} \hspace{6pt}
\subfigure[]{\label{fig:subnet3}\includegraphics[width=.2\linewidth]{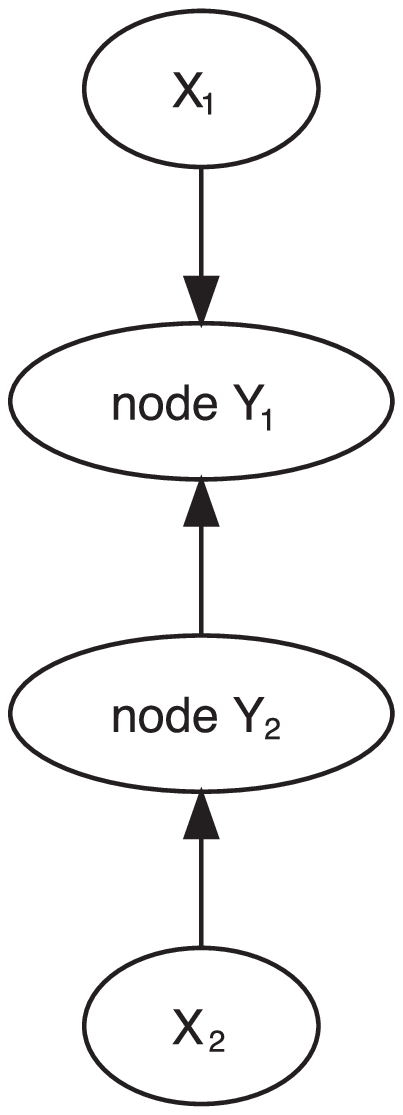}} \hspace{6pt}
\subfigure[]{\label{fig:subnet4}\includegraphics[width=.2\linewidth]{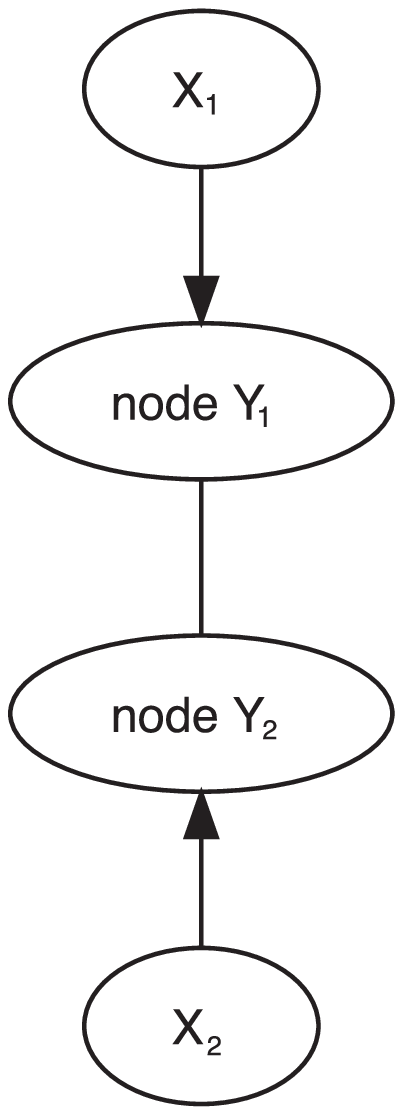}}
	\caption{An Illustrative Example of Networks Which Are Not Markov Equivalent. However, without $X_1$ and $X_2$, sub-networks consisting of only node $Y_1$ and $Y_2$ will be Markov equivalent.}
	\label{fig:illustrateMarkovEqu}
\end{figure}

\subsection{\uppercase{Two-Stage Differential Analysis of Networks} \label{sec:moderepar}}

Here we intend to develop a regularized version of the two-stage least squares. We first screen for exogenous variables and conduct $\ell_2$ regularized regression of each endogenous variable against screened exogenous variables to obtain its good prediction which helps address the endogeneity issue in the following stage. At the second stage, we reparametrize the model to explicitly model the common and differential regulatory effects and identify them via the adaptive lasso method.

\subsubsection{The Calibration Stage}

To address the endogeneity issue, we aim for good prediction of each endogenous variable following the reduced model in (\ref{model:reduced}). However, in the high-dimensional setting, the dimension $q$ of $\mathbf{X}^{(k)}$ can be much larger than the sample size $n^{(k)}$, and any direct prediction with all exogenous variables may not produce consistent prediction. Note that both \citet{lin2015regularization} and \citet{zhu2017sparse} proposed to conduct variable selection with lasso or its variants and predict with selected exogenous variables. We here instead propose to first screen for exogenous variables with ISIS \citep{fan2008sure}, and then apply ridge regression to predict the endogenous variables with screened exogenous variables. While variable screening is more robust and provides higher coverage of true variables than variable selection, its combination with ridge regression puts less computational burden. Furthermore, as shown by \citet{chen2015two}, ridge regression performs well in predicting the endogenous variables.

Let $\mathcal{M}^{(k)}_{i}$ denotes the selected index set for $i$-th node in $k$-th network from the variable screening which reduces the dimension from $q$ to $d=|\mathcal{M}^{(k)}_{i}|$. The \textit{Sure Independence Screening Property} in \citet{fan2008sure} can be directly applied in our case to guarantee that $\mathcal{M}^{(k)}_{i}$ covers the true set $\mathcal{M}^{(k)}_{i0}$ with a large probability.

\textbf{Assumption 2}. $n^{(1)}$ and $n^{(2)}$ are at the same order, i.e., $n_{\min} = \min(n^{(1)}, n^{(2)}) \asymp n^{(1)} \asymp n^{(2)}$, and $p \asymp q$.

\begin{theorem} \label{thm:sis} Assuming Conditions 1-4 in the supplemental materials which restrict positive $\tilde{\tau}$ and $\tilde{\kappa}$, under Assumption~2, there exists some $\theta \in (0, 1-2\tilde{\kappa}-\tilde{\tau})$ such that, when $d =|\mathcal{M}_{i}^{(k)}|= O((n_{\min})^{1-\theta})$, we have, for some constant $C>0$,
$$
\mathbb{P}( \mathcal{M}^{(k)}_{i0} \subseteq \mathcal{M}^{(k)}_{i} ) = 1 - \mathcal{O}\left(\exp\left\{-\frac{C (n^{(k)})^{1-2\tilde{\kappa}}}{\log(n^{(k)})}\right\}\right).
$$
\end{theorem}

Hereafter we assume that $\mathcal{M}^{(k)}_{i}$ successfully covers the true set $\mathcal{M}^{(k)}_{i0}$ for convenience of stating the following assumptions and theorems. That is, the probability of successful screening is not incorporated into our assumptions or theorems in the below.

For node $i$ in network $k$, let $\mathbf{X}^{(k)}_{\scriptscriptstyle\mathcal{M}_{i}^{(k)}}$ denotes the submatrix of $\mathbf{X}^{(k)}$ with prescreened columns which are indexed by $\mathcal{M}_{i}^{(k)}$. With $\boldsymbol{\pi}_i^{(k)}$ denoting the $i$-th column of $\boldsymbol{\pi}^{(k)}$, the subvector of $\boldsymbol{\pi}_i^{(k)}$ indexed by $\mathcal{M}_{i}^{(k)}$ will be simply denoted by $\boldsymbol{\pi}^{(k)}_{\scriptscriptstyle\mathcal{M}_{i}^{(k)}}$ without confusion. Such simplified notations will apply to other vectors and matrices in the rest of this paper.

With $d$ pre-screened exogenous variables, we can apply ridge regression to the model
\begin{equation}
\label{model:preScreenedStep1}
\mathbf{Y}^{(k)}_{i} = \mathbf{X}_{\scriptscriptstyle \mathcal{M}_{i}^{(k)}} ^{(k)}\, \boldsymbol{\pi}^{(k)}_{\scriptscriptstyle\mathcal{M}_{i}^{(k)}} +\boldsymbol{\xi}^{(k)}_i,
\end{equation}
to obtain the estimates $\hat{\boldsymbol{\pi}}^{(k)}_{\scriptscriptstyle\mathcal{M}_{i}^{(k)}}$ of $\boldsymbol{\pi}^{(k)}_{\scriptscriptstyle\mathcal{M}_{i}^{(k)}}$, and predict $\mathbf{Y}^{(k)}_{i}$ with $\hat{\mathbf{Y}}^{(k)}_{i} = \mathbf{X}_{\scriptscriptstyle \mathcal{M}_{i}^{(k)}}^{(k)}  \hat{\boldsymbol{\pi}}^{(k)}_{\scriptscriptstyle\mathcal{M}_{i}^{(k)}}$.

\subsubsection{The Construction Stage}

With known $\mathcal{A}_{i}^{(k)}$, we can rewrite model (\ref{model:onenode}) as,
\begin{equation}
\label{model:onenodewithS}
\mathbf{Y}^{(k)}_{i}= \mathbf{Y}^{(k)}_{-i} \boldsymbol{\gamma}^{(k)}_{i} +\mathbf{X}^{(k)}_{ \scriptscriptstyle \mathcal{A}_{i}^{(k)}} \boldsymbol{\phi}^{(k)}_{ \scriptscriptstyle \mathcal{A}_{i}^{(k)}}+\boldsymbol{\epsilon}^{(k)}_{i}.
\end{equation}
Before we use the predicted $\mathbf{Y}^{(k)}$ to identify both common and differential regulatory effects across the two networks, we first reparametrize the model so as to define differential regulatory effects explicitly,
\begin{equation}
\label{notation:repar}
\begin{aligned}
\boldsymbol{\beta}_{i}^{-} = \frac{\boldsymbol{\gamma}_{i}^{(1)}-\boldsymbol{\gamma}_{i}^{(2)}}{2}, \
\boldsymbol{\beta}_{i}^{+} =\frac{\boldsymbol{\gamma}_{i}^{(1)}+\boldsymbol{\gamma}_{i}^{(2)}}{2}.
\end{aligned}
\end{equation}
Here $\boldsymbol{\beta}_i^{-}$ and $\boldsymbol{\beta}_i^{+}$ represent the \textbf{differential} and \textbf{average regulatory effects} between the two networks, respectively. We need compare $\boldsymbol{\beta}_i^{+}$ with $\boldsymbol{\beta}_i^{-}$ to identify the \textbf{common regulatory effects}, that is, effects of all regulations with nonzero values in $\boldsymbol{\beta}_i^{+}$ but zero values in $\boldsymbol{\beta}_i^{-}$.

Note that other differential analysis of networks may suggest a different reparametrization to identify common and differential regulatory effects. For example, in a typical case-control study, we may expect few structures in the case network mutated from the control network. While we are interested in identifying differential structures in the case network, we may be also interested in identifying baseline network structures in the control network. Therefore we may reparametrize the model with the regulatory effects in the control network, as well as the differential regulatory effects defined as the difference of regulatory effects between case and control networks. We want to point out that the method described here still applies and we can also derive similar theoretical results as follows.

Following the reparametrization in (\ref{notation:repar}), we can rewrite model (\ref{model:onenodewithS}) as follows,
\begin{eqnarray} \label{reparmodel}
\lefteqn{\begin{pmatrix}
\mathbf{Y}_{i}^{(1)}\\
\mathbf{Y}_{i}^{(2)}
\end{pmatrix} =
\begin{pmatrix}
\mathbf{Y}_{-i}^{(1)} & \mathbf{Y}_{-i}^{(1)}\\
\mathbf{Y}_{-i}^{(2)} & -\mathbf{Y}_{-i}^{(2)}
\end{pmatrix} \begin{pmatrix}
\boldsymbol{\beta}_{i}^{+} \\  \boldsymbol{\beta}_{i}^{-}
\end{pmatrix}+} \nonumber\\
& & \begin{pmatrix}
\mathbf{X}_{\mathcal{A}_i^{(1)}}^{(1)} & 0\\
0& \mathbf{X}_{\mathcal{A}_i^{(2)}}^{(2)}
\end{pmatrix} \begin{pmatrix}
\boldsymbol{\phi}_{\mathcal{A}_i^{(1)}}^{(1)}\\  \boldsymbol{\phi}_{\mathcal{A}_i^{(2)}}^{(2)}
\end{pmatrix}+
\begin{pmatrix}
\boldsymbol{\epsilon}_{i}^{(1)}\\
\boldsymbol{\epsilon}_{i}^{(2)}
\end{pmatrix}.
\end{eqnarray}

Denote
\begin{eqnarray*}
& \mathbf{Y}_{i} =
\begin{pmatrix}
\mathbf{Y}_{i}^{(1)}\\
\mathbf{Y}_{i}^{(2)}
\end{pmatrix}, &
\mathbf{Z}_{-i} =
\begin{pmatrix}
\mathbf{Y}_{-i}^{(1)} & \mathbf{Y}_{-i}^{(1)}\\
\mathbf{Y}_{-i}^{(2)} & -\mathbf{Y}_{-i}^{(2)}
\end{pmatrix}, \\
& \boldsymbol{\beta}_i =
\begin{pmatrix}
\boldsymbol{\beta}_i^{+}\\\boldsymbol{\beta}_i^{-}
\end{pmatrix},
& \boldsymbol{\epsilon}_i =
\begin{pmatrix}
\boldsymbol{\epsilon}_i^{(1)}\\\boldsymbol{\epsilon}_i^{(2)}
\end{pmatrix}.
\end{eqnarray*}
Further define the projection matrix for each network,
\[
\mathbf{H}_i^{(k)} = I_{n^{(k)}} - \mathbf{X}_{\mathcal{A}_i^{(k)}}^{(k)}\left(\mathbf{X}_{\mathcal{A}_i^{(k)}}^{(k)T} \mathbf{X}_{\mathcal{A}_i^{(k)}}^{(k)} \right)^{-1}\mathbf{X}_{\mathcal{A}_i^{(k)}}^{(k)T}.
\]
Applying the projection matrix $\mathbf{H}_i = \diag\{\mathbf{H}_i^{(1)}, \mathbf{H}_i^{(2)}\}$ to both sides of model (\ref{reparmodel}), we can remove the exogenous variables from the model and obtain,
\begin{equation}
\mathbf{H}_i \mathbf{Y}_i = \mathbf{H}_i \mathbf{Z}_{-i} \boldsymbol{\beta}_{i} + \mathbf{H}_i \boldsymbol{\epsilon}_i.
\label{equation:projected}
\end{equation}

To address the endogeneity issue, we predict $\mathbf{Z}_{-i}$ by replacing its component $\mathbf{Y}^{(k)}_{-i}$ with the predicted value $\mathbf{\hat{Y}}^{(k)}_{-i}$ from the previous stage, and then regressing $\mathbf{H}_i \mathbf{Y}_i$ against $\mathbf{H}_i \hat{\mathbf{Z}}_{-i}$ with the adaptive lasso to consistently estimate $\boldsymbol{\beta}_{i}$. That is, an optimal $\boldsymbol{\beta}_{i}$ can be obtained as,
\begin{equation*}
\hat{\boldsymbol{\beta}}_{i}  = \underset{\boldsymbol{\beta}_{i}}{\text{arg min }} \left\{ \frac{1}{n} ||\mathbf{H}_i \mathbf{Y}_i - \mathbf{H}_i \hat{\mathbf{Z}}_{-i}\boldsymbol{\beta}_{i} ||^2_2 +  \lambda_i\boldsymbol{\omega}_i^T|\boldsymbol{\beta}_{i}|_1 \right \},
\end{equation*}
where $|\boldsymbol{\beta}_{i}|_1$ is a vector taking elementwise absolute values of $\boldsymbol{\beta}_{i}$, $\boldsymbol{\omega}_{i}$ is the adaptive weights whose components are inversely proportional to the components of an initial estimator of $\boldsymbol{\beta}_i$, and $\lambda_{i}$ is the adaptive tuning parameter.

The two-stages algorithm is summarized in Algorithm~\ref{algo:full}. With the estimator $\hat{\boldsymbol{\beta}}_{i}$ from the second stage, we can accordingly obtain estimators $\hat{\boldsymbol{\gamma}}_{i}^{(1)} = \hat{\boldsymbol{\beta}}_{i}^{+} + \hat{\boldsymbol{\beta}}_{i}^{-} $ and $\hat{\boldsymbol{\gamma}}_{i}^{(2)} = \hat{\boldsymbol{\beta}}_{i}^{+} - \hat{\boldsymbol{\beta}}_{i}^{-} $.

\begin{algorithm}[tb]
\caption{\label{algo:full}\underline{Re}parameterization-Based \underline{D}ifferential Analysis of \underline{Net}work (ReDNet)}
    \begin{algorithmic}
 \STATE	{\bfseries Input:} For $k\in\{1, 2\}$, $\bY^{(k )}$, $\bX^{(k)}$, index set $\mathcal{A}_{i}^{(k)}$ for each $i\in\{1, 2, \ldots, p\}$. Set $d = O(n_{\min}^{1-\theta})$.
\FOR{$i \rightarrow 1$  {\bfseries to} $p$ }
    \STATE Stage 1.a. Screen for a submatrix $\mathbf{X}^{(k)}_{\scriptscriptstyle \mathcal{M}_{i}^{(k)}}$ of $\bX^{(k)}$ for $\bY^{(k)}_{i}$ versus $\bX^{(k)}$ and set $\mathbf{X}^{(k)}_{\scriptscriptstyle \mathcal{M}_{i}^{(k)}} = \bX^{(k)}$ if $ q \le n^{(k)} $.\\
	  \STATE Stage 1.b. Apply ridge regression to regress $\bY^{(k)}_{i}$ against $\mathbf{X}^{(k)}_{\scriptscriptstyle \mathcal{M}_{i}^{(k)}}$ to obtain prediction $\hat{\bY}^{(k)}_i$.
\ENDFOR
\FOR{$i \rightarrow 1$  {\bfseries to} $p$ }
  \STATE Stage 2. Apply adaptive lasso to regress $\mathbf{H}_i \mathbf{Y}_i$ against $\mathbf{H}_i \hat{\mathbf{Z}}_{-i}$ to obtain coefficients estimate $\hat{\boldsymbol{\beta}}_i$.
\ENDFOR
  \STATE {\bfseries Output:} The common and differential regulatory effects in $\hat{\boldsymbol{\beta}}_{1},\ldots,\hat{\boldsymbol{\beta}}_{p}$.
    \end{algorithmic}
\end{algorithm}

\subsection{\uppercase{Theoretical Analysis}} \label{sec:theoractical}

As shown in Theorem~\ref{thm:sis}, a screening method like ISIS \citep{fan2008sure} can identify $\mathcal{M}^{(k)}_{i}$ with size $d = O(n_{\min}^{1-\theta})$ which covers the true set $\mathcal{M}^{(k)}_{i0}$ with a sufficiently large probability. For the sake of simplicity and without loss of generality, in the following we assume $\mathcal{M}^{(k)}_{i0}\subseteq \mathcal{M}^{(k)}_{i}$.

We first investigate the consistency of predictions from the first stage. The consistency properties will be characterized by prespecified sequences $f^{(k)}=o(n^{(k)})$ but $f^{(k)}\rightarrow\infty$ as $n^{(k)}\rightarrow\infty$. We also denote $f_{\max} = f^{(1)}\lor f^{(2)}$, i.e., $\max\{f^{(1)},f^{(2)}\}$.

The following assumption is required for the consistency properties.

\textbf{Assumption 3}. For each network $k$, the singular values of $\mathbf{I} - \boldsymbol{\Gamma}^{(k)}$ are positively bounded from below, and there exist some positive constants $c_1^{(k)}$ and $c_2^{(k)}$ such that, for each node $i$, $\text{max}_{\ltwon{\delta}=1} (n^{(k)})^{-1/2}\ltwon{\bX^{(k)}_{\scriptscriptstyle \mathcal{M}_{i}^{(k)} }\delta} \le c_1^{(k)}$ and $\text{min}_{\ltwon{\delta}=1} (n^{(k)})^{-1/2}\ltwon{\bX^{(k)}_{\scriptscriptstyle \mathcal{M}_{i}^{(k)} }\delta} \ge c_2^{(k)}$. Furthermore, the ridge parameter $\lambda^{(k)}_{i}=o(n_{\min})$.

For the ease of exposition, we will omit the subscript $\mathcal{M}_i^{(k)}$ from $\bX^{(k)}_{\scriptscriptstyle \mathcal{M}_{i}^{(k)}}$ henceforth, and accordingly use $\boldsymbol{\pi}_i^{(k)}$ and $\hat{\boldsymbol{\pi}}_i^{(k)}$ which include the zero components of excluded predictors.

Denote $\mathbf{X} = \diag\{\mathbf{X}^{(1)},\mathbf{X}^{(2)}\}$, and
\begin{equation*}
\begin{aligned}
\mathbf{Z} =
\begin{pmatrix}
\mathbf{Y}^{(1)} & \mathbf{Y}^{(1)}\\
\mathbf{Y}^{(2)} & -\mathbf{Y}^{(2)}
\end{pmatrix}, \ \ \ \
\boldsymbol{\Pi} =
\begin{pmatrix}
\boldsymbol{\pi}^{(1)} & \boldsymbol{\pi}^{(1)}\\
\boldsymbol{\pi}^{(2)} & -\boldsymbol{\pi}^{(2)}
\end{pmatrix}.
\end{aligned}
\end{equation*}
We use $\boldsymbol{\Pi}_{j}$ to denote the $j$-th column of the matrix $\boldsymbol{\Pi}$ and $\boldsymbol{\pi}^{(k)}_{j}$ to denote the $j$-th column of the matrix $\boldsymbol{\pi}^{(k)}$. We also use $\hat{\mathbf{Z}}$ and $\hat{\boldsymbol{\Pi}}$ to denote the prediction of $\mathbf{Z}$ and estimate of $\boldsymbol{\Pi}$, respectively. Note that, with the ridge parameter $\lambda^{(k)}_{i}$ for the ridge regression taken on node $i$ in network $k$, we have $r_{i}^{(k)} = (\lambda^{(k)}_{i})^2\ltwon{\boldsymbol{\pi}^{(k)}_{i} }^2/n^{(k)}$ and hence define $r_{\max} = \underset{1\le i \le p}{\max} [r_{i}^{(1)} \lor r_{i}^{(2)}]$. Then the estimation and prediction losses at the first stage can be summarized in the following theorem.

\begin{theorem} \label{theorem:step1Maintheorem}
Under Assumptions 1-3, for each $j\in \{1,2,\ldots,2p\}$, there will exist some constant $C_1$ and $C_2$ such that, with probability at least $1- e^{-f^{(1)}} - e^{-f^{(2)}}$, \\
1. $\ltwon{\hat{\boldsymbol{\Pi}}_{j} - \boldsymbol{\Pi}_{j}}^2 \le C_1 \left(d\lor r_{\max} \lor f_{\max}\right)\big/n_{\min}$;\\
2. $\ltwon{\mathbf{X}(\hat{\boldsymbol{\Pi}}_{j} - \boldsymbol{\Pi}_{j})}^2 \le C_2 \left(d\lor r_{\max} \lor f_{\max}\right)$.
\end{theorem}

The proof is detailed in the supplemental materials.

Note that these two sets of losses can be controlled by the same upper bounds across the two networks with probability at least $\small 1 - e^{-f^{(1)}+\text{log}\,(p)} - e^{-f^{(2)} + \text{log}\,(p)}$. Therefore, $f^{(k)}$ can be selected such that $f^{(k)} - log(p)\rightarrow \infty$, which will provide a probability approaching one to have the network-wide losses approaching zero.

Furthermore, the dimension $p$ can be divergent up to an exponential order, say $p=e^{n_{\min}^c}$ for some $c\in(0,1)$. We can set $f^{(1)} = f^{(2)} = n_{\min}^{(1+c)/2}$ and, apparently, $f^{(k)}=o(n_{\min})$ but $f^{(k)} - log(p) = n_{\min}^{(1+c)/2}-n_{\min}^c \rightarrow \infty$.

Since the ridge parameter $\lambda^{(k)}_{i}=o(n_{\min})$, $r_i^{(k)}=\ltwon{\boldsymbol{\pi}^{(k)}_{i}}^2 \times o(n_{\min})$. Therefore, when all $\ltwon{\boldsymbol{\pi}^{(k)}_{i}}$ are uniformly bounded, we have $r_{\max} = o(n_{\min})$. Otherwise, the ridge parameter $\lambda^{(k)}_{i}$ should be adjusted accordingly to control both estimation and prediction losses.

Before we characterize the consistency of estimated regulatory effects on the second stage, we first introduce the following concept of restricted eigenvalue which is used to present an assumption.
\begin{definition}
The restricted eigenvalue of a matrix $\mathbf{A}$ on an index set $\mathcal{S}$ is defined as
\begin{equation}
\phivarmin{\mathbf{A}}{\mathcal{S}}{} = \underset{\lonen{\delta_{\mathcal{S}^c}} \le 3 \lonen{\delta_{\mathcal{S} }}}{\min}\frac{\ltwon{\mathbf{A}\delta}}{\sqrt{n}\ltwon{\delta_{\mathcal{S}}}}.
\end{equation}
\end{definition}

For the $i$-th node, we use $\mathcal{S}_i$ to denote the non-zero indices of $\boldsymbol{\beta}_{i}$, i.e., $\mathcal{S}_i = \text{supp}(\boldsymbol{\beta}_{i})$.
Further denote
\begin{equation*}
\boldsymbol{\Pi}_{-i} =
\begin{pmatrix}
\boldsymbol{\pi}_{-i}^{(1)} & \boldsymbol{\pi}_{-i}^{(1)}\\
\boldsymbol{\pi}_{-i}^{(2)} & -\boldsymbol{\pi}_{-i}^{(2)}
\end{pmatrix}.
\end{equation*}

As in \citet{bickel2009simultaneous}, we impose the following restricted eigenvalue condition on the design matrix in (\ref{equation:projected}).

\textbf{Assumption 4.} There exists a constant $\boldsymbol{\phi}_0 > 0$ such that $\phivarmin{\mathbf{H}_i\mathbf{X}\boldsymbol{\Pi}_{-i}}{ \mathcal{S}_{i} }\ge \boldsymbol{\phi}_0$. Furthermore, $\|\boldsymbol{\omega}_{\mathcal{S}_{i}}\|_{\infty} \le \|\boldsymbol{\omega}_{\mathcal{S}_{i}^c}\|_{-\infty}$.

Let $n=n^{(1)}+n^{(2)}$, $c_{\max}=c^{(1)}_1\lor c^{(2)}_1$, and $\mathbf{B} = [\boldsymbol{\beta}_1,\boldsymbol{\beta}_2,\ldots, \boldsymbol{\beta}_{p}]$. The matrix norms $\lonen{\cdot }$ and $\|\cdot\|_{\infty}$ are the maximum of column and row sums of absolute values of the matrix, respectively. For a vector, we define $\|\cdot\|_{\infty}$ and $\|\cdot\|_{-\infty}$ to be the maximum and minimum absolute values of its components. Then, we can derive the following loss bounds for the estimation and prediction at the second stage on the basis of Theorem~\ref{theorem:step1Maintheorem}.

\begin{theorem} \label{theoremAdaConsisit}
Suppose that, for node $i$, the adaptive lasso at the second stage takes the tuning parameter \makebox{\small $\lambda_{i} \asymp \|\boldsymbol{\omega}_i\|_{-\infty}^{-1} \lonen{\mathbf{B}}\lonen{\boldsymbol{\Pi}}\sqrt{ (d \lor r_{\max}\lor f_{\max})\log(p) \big/ n_{\text{min}} }$}, and $\sqrt{(d \lor r_{\max}\lor f_{\max}) \big/ n} + c_{\max}\lonen{\boldsymbol{\Pi}}\le \sqrt{c_{\max}^2\lonen{\boldsymbol{\Pi}}^2+\phi_0^2/(64C_2 |\mathcal{S}_{i}|)}$. Let $h_n = (\lonen{\mathbf{B}}^2 \land 1)$ $\times \left( (n\lonen{\boldsymbol{\Pi}}^2/d) \land (d\lor r_{\max} \lor f_{\max})\right)\log(p)$. Under Assumptions 1-4, there exist positive constants $C_3$ and $C_4$ such that, with probability at least \makebox{\small $1- 3e^{-C_3 h_n + \log(4pq)} - e^{-f^{(1)}+\log(p)} - e^{-f^{(2)} + \log(p)}$},\\
1. Estimation Loss:
\begin{eqnarray*}
\lefteqn{\lonen{\hat{\boldsymbol{\beta}}_i-\boldsymbol{\beta}_{i }}\le 8C_4 |\mathcal{S}_{i}| \times}\\
&& \frac{\|\boldsymbol{\omega}_{\mathcal{S}_{i}}\|_{\infty} \lonen{\mathbf{B}} \lonen{\boldsymbol{\Pi}}}{\boldsymbol{\phi}_0^2 \|\boldsymbol{\omega}_{i}\|_{-\infty}} \sqrt{\frac{(d \lor r_{\text{max}}\lor f_{\max} )\log(p)}{n_{\min}}};
\end{eqnarray*}
2. Prediction Loss:
\begin{eqnarray*}
\lefteqn{\frac{1}{n}\ltwon{ \mathbf{H}_i\hat{\mathbf{Z}}_{-i} (\hat{\boldsymbol{\beta}}_i - \boldsymbol{\beta}_{i } )}^2 \le C_4^2 |\mathcal{S}_{i}|\times} \\
& & \frac{ \|\boldsymbol{\omega}_{\mathcal{S}_{i}}\|_{\infty}^2 \lonen{\mathbf{B}}^2\lonen{\boldsymbol{\Pi}}^2}{\boldsymbol{\phi}_0^2 \|\boldsymbol{\omega}_{i}\|_{-\infty}^2} \frac{(d \lor r_{\text{max}}\lor f_{\max})\log(p)}{n_{\min}}.
\end{eqnarray*}
\label{theorem:step2estimate}
\end{theorem}

The main idea of the proof is to take advantage of the commonly used restricted eigenvalue condition and irrepresentable condition for lasso-type estimator. However, the design matrix in our case includes predicted values instead of the original one, which complicates the proof. We claim that the restricted eigenvalue and irrepresentable condition still hold for the predicted design matrix as long as the estimation and prediction losses are well controlled at the calibration stage. The proof is detailed in the supplemental materials.

The available anchoring regulators as required by Assumption~1 implies that both $\lonen{\mathbf{B}}>0$ and $\lonen{\boldsymbol{\Pi}}>0$, so $h_n/\log(p)\rightarrow\infty$. That is, these loss bounds hold with a sufficient large probability with properly chosen $f^{(k)}$.

The two sets of losses in Theorem~\ref{theoremAdaConsisit} can also be controlled across the whole system by the same upper bounds defined by replacing $|\mathcal{S}_i|$ with $s_{\max} = \max_{i}|\mathcal{S}_i|$, with probability at least $1- 3e^{-C_3 h_n + \log(4 q) + 2\log(p)} - e^{-f^{(1)}+2\log(p)} - e^{-f^{(2)} + 2\log(p)}$. When both $p$ and $q$ are divergent up to an exponential order, say $p\asymp q\asymp e^{n_{\min}^c}$ for some $c\in(0,1)$, we can set $f^{(1)} = f^{(2)}=n_{\min}^{(1+c)/2}$ to guarantee the bounds at a sufficient large probability. However, the bounds are determined by $(d \lor r_{\max}\lor f_{\max})\log(p)$ which is $o(n_{\min})$ only when $c<\min(1/3,\theta)$. Therefore, if $s_{\max}$ also diverges up to $n_{\min}^{\tilde{c}}$ with $\tilde{c}<\min(1/4,\theta/2,1-\theta)$, the losses can be well controlled for $c<\min((1-4\tilde{c})/3, \theta-2\tilde{c})$.

Note that, with properly chosen $f^{(1)}$ and $f^{(2)}$, these losses are well controlled at $o (n_{\min})$, revealing the fact that we need to have sufficient observations for each network for consistent differential analysis of the two networks.

Let $W_i = diag\{\bweight_{i}\}$. Denote $\mathcal{I}_{i} = \frac{1}{n} \boldsymbol{\Pi}_{-i}^T\mathbf{X}^T\mathbf{H}_i\mathbf{X}\boldsymbol{\Pi}_{-i}$ and $\hat{\mathcal{I}}_{i} = \frac{1}{n} \hat{\boldsymbol{\Pi}}_{-i}^T\mathbf{X}^T\mathbf{H}_i\mathbf{X}\hat{\boldsymbol{\Pi}}_{-i}$. Let $\soo$ be a submatrix of $\mathcal{I}_{i}$ with rows and columns both indexed by $\mathcal{S}_{i}$, and $\sto$ be a submatrix of $\mathcal{I}_{i}$ with rows and columns indexed by $\mathcal{S}_{i}^c$ and $\mathcal{S}_{i}$, respectively. $\hsoo$ and $\hsto$ are similarly defined from $\hat{\mathcal{I}}_{i}$. We further define the minimal signal strength $b_{i} =\underset{j\in \mathcal{S}_{i}}{\max} |\boldsymbol{\beta}_{ij}|$ and $\psi_i = \infn{\soo[-1]  W_{\mathcal{S}_i}}$.

The following assumption, reminiscent of the \textit{adaptive irrepresentable condition} in \citet{huang2008adaptive}, helps investigate the selection consistency of regulatory effects.

\textbf{Assumption 5.} (Weighted Irrepresentable Condition) There exists a constant $\tau\in (0,1)$ such that $\infn{W_{\mathcal{S}_{i}^c}^{-1} \sto \soo[-1] W_{\mathcal{S}_{i}}} < 1- \tau$.

\begin{theorem}\label{theorem:variableselection} (Variable Selection Consistency) Denote $\mathcal{\hat{S}}_i= \{ j: \hat{\boldsymbol{\beta}}_{ij} \ne 0 \}$. Suppose that, for each node $i$, $\hsoo$ is invertible, $b_{i} > \lambda_{i} \psi_i/(2-\tau)$, and $\scriptsize \sqrt{(d \lor r_{\max}\lor f_{\max}) \big/ n} + c_{\max}\lonen{\boldsymbol{\Pi}}\le \sqrt{c_{\max}^2\lonen{\boldsymbol{\Pi}}^2+\min(\phi_0^2\big/64, \tau (4-\tau)^{-1}\|\boldsymbol{\omega}_{i}\|_{-\infty}/\psi_i)\big/(C_2 |\mathcal{S}_{i}|)}$. Under Assumptions 1-5, there exists some constant $C_5>0$ such that $\hat{\mathcal{S}}_i = \mathcal{S}_i$ with probability at least \makebox{$1- 3e^{-C_5 h_n + \log(4pq)} - e^{-f^{(1)}+\log(p)} - e^{-f^{(2)} + \log(p)}$}.
\end{theorem}

This theorem implies that our proposed method can identify both common and differential regulatory effects between the two networks with a sufficiently large probability. On the other hand, the assumed weighted irrepresentable condition means that the true signal should not correlate too much with irrelevant predictors so as to conduct a successful differential analysis. The corresponding proof is detailed in the supplemental materials.

\section{\uppercase{Experiments}}\label{sec:simulation}

\subsection{\uppercase{Synthetic Data Evaluation}}

\begin{figure*}[ht]
\centering
\subfigure[$X_1 \ne X_2$, MCC]{\includegraphics[scale=0.28]{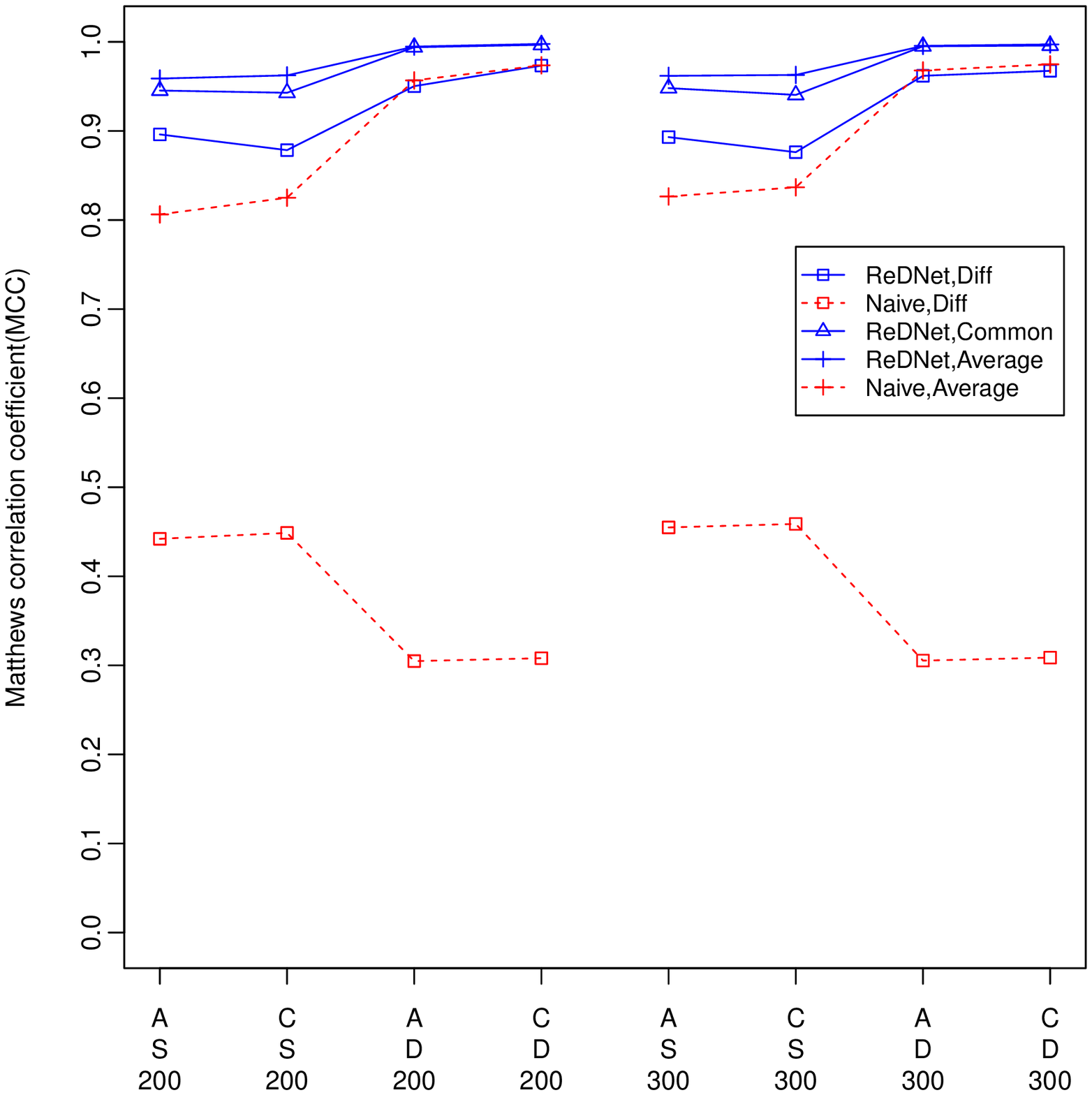}}
\subfigure[$X_1 \ne X_2$, FDR]{\includegraphics[scale=0.28]{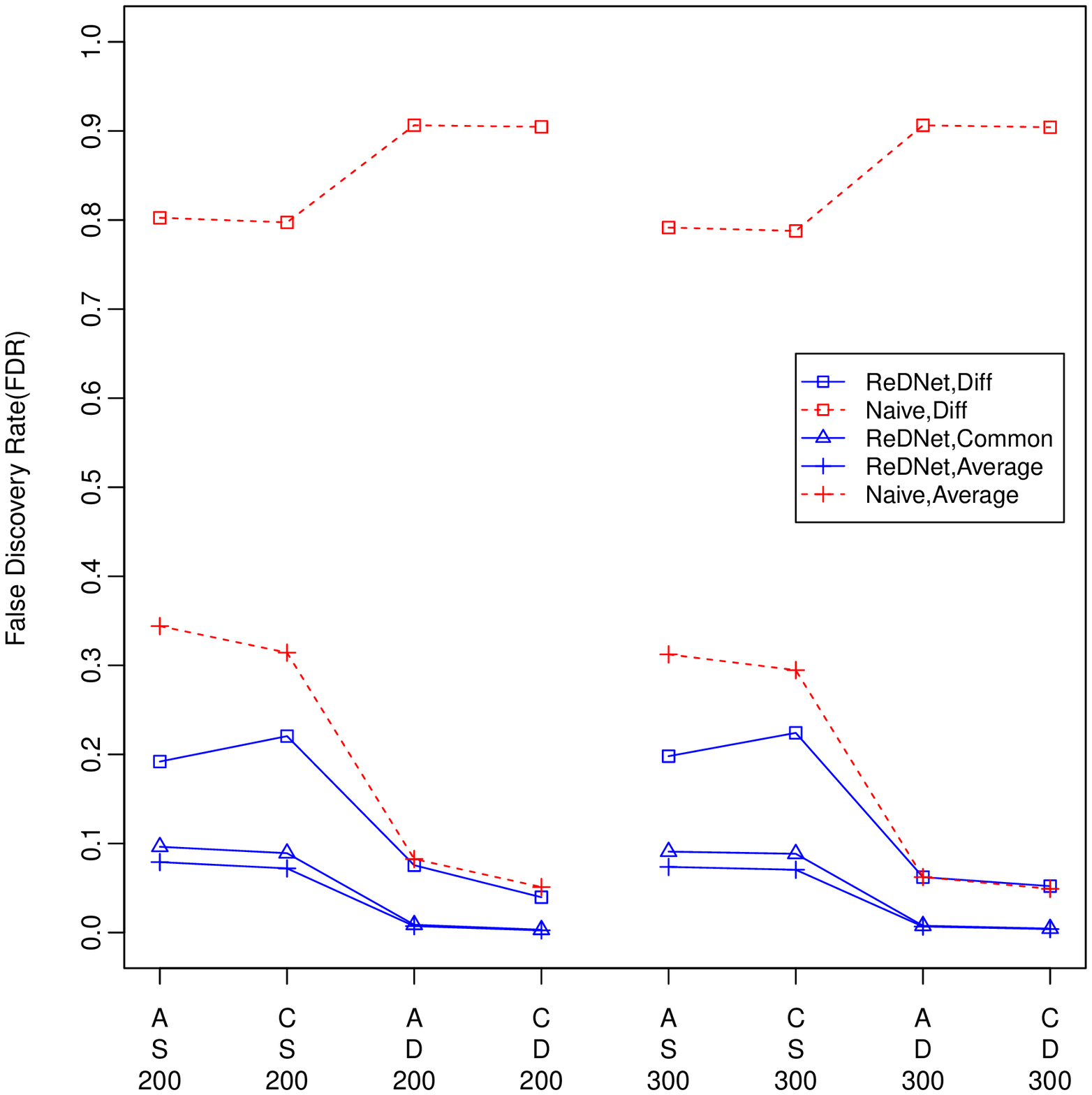}}
\subfigure[$X_1 \ne X_2$, Power]{\includegraphics[scale=0.28]{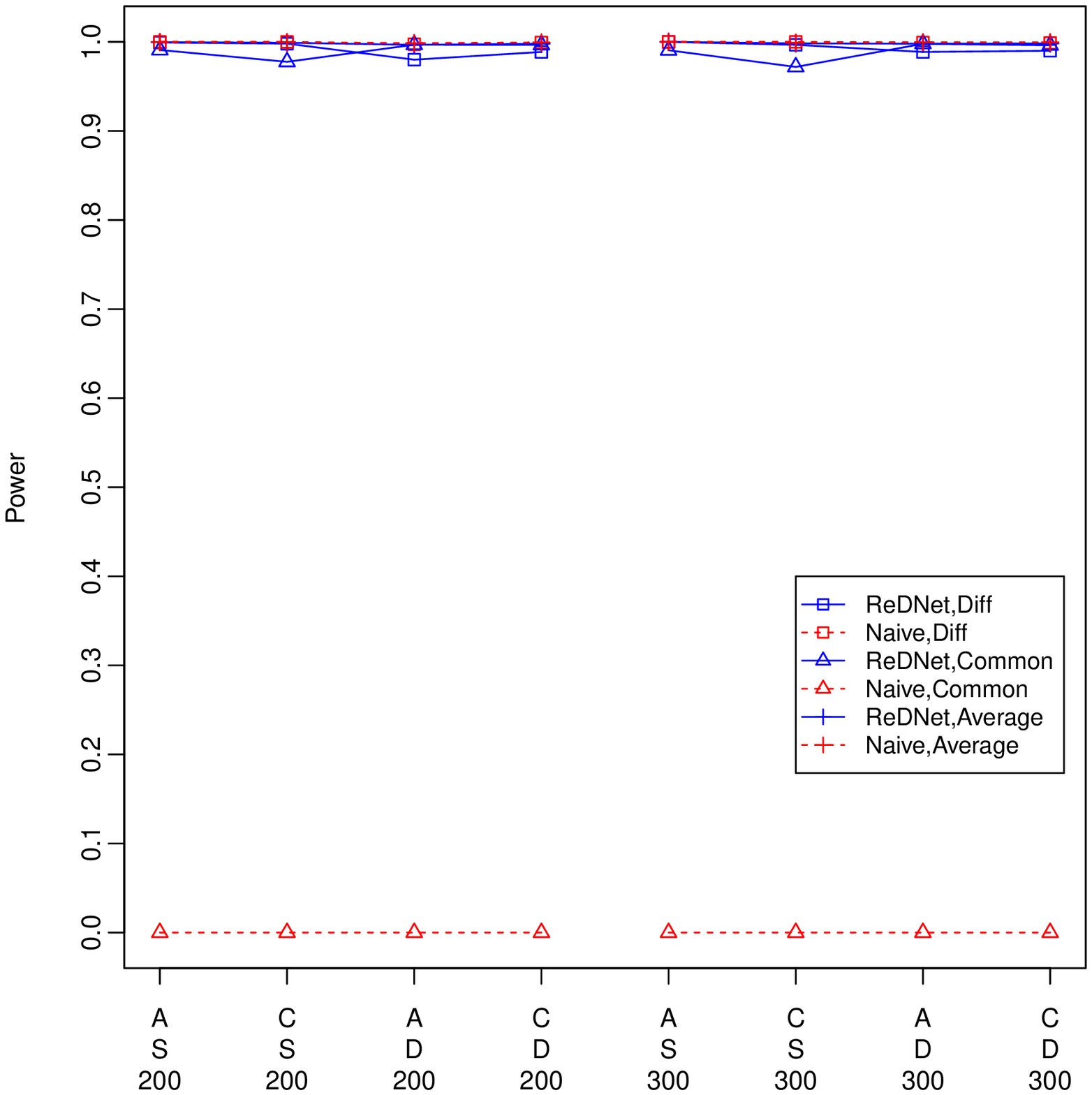}}

\subfigure[$X_1 = X_2$, MCC]{\includegraphics[scale=0.28]{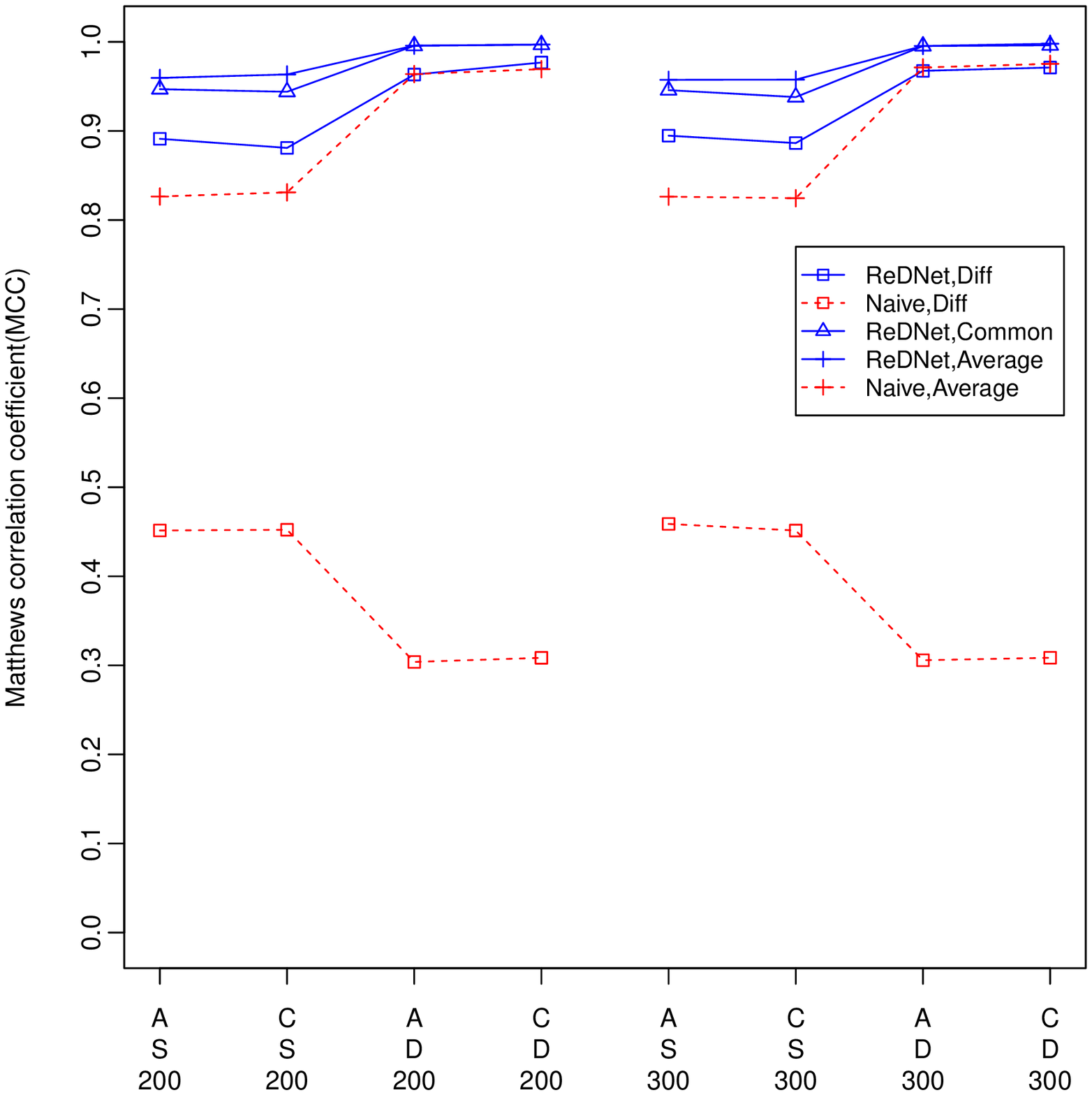}}
\subfigure[$X_1 = X_2$, FDR]{\includegraphics[scale=0.28]{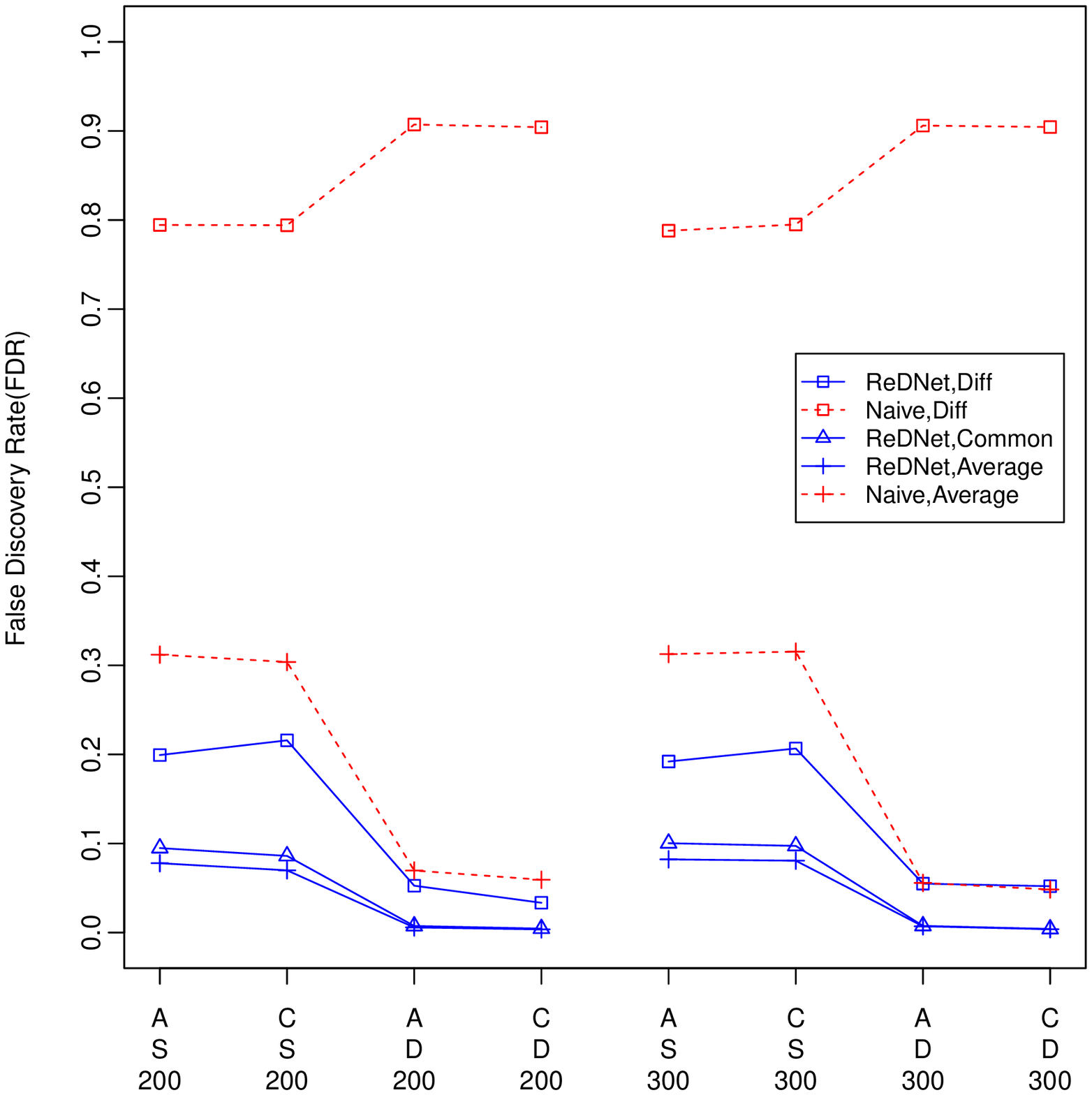}}
\subfigure[$X_1 = X_2$, Power]{\includegraphics[scale=0.28]{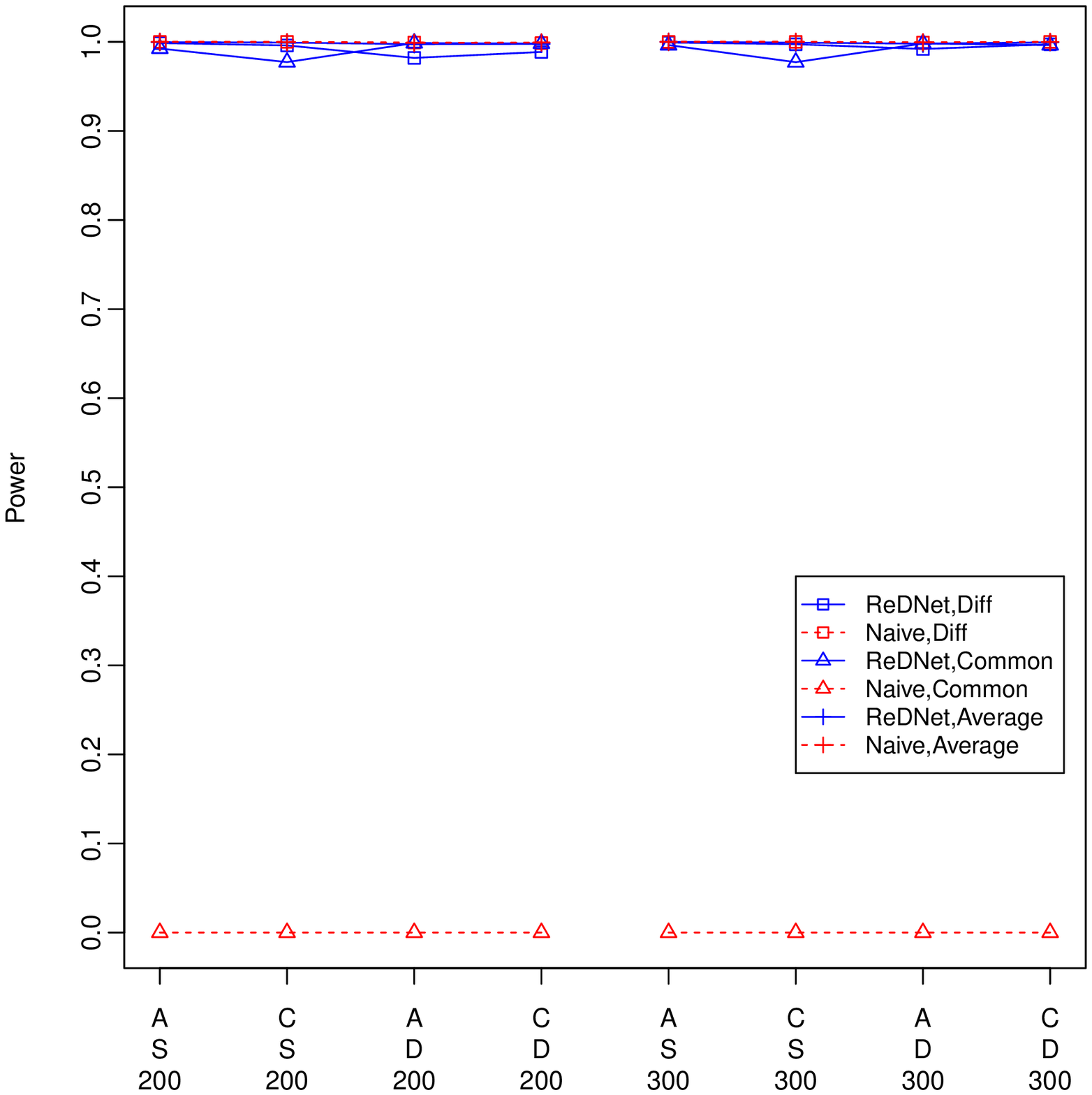}}

\caption{Performance of ReDNet Versus the Naive Approach which Independently Constructs Two Networks. The results average over $100$ synthetic data sets for different types of networks, with letters \textit{A, C, S, D} in the x-axis denoting \underline{A}cyclic, \underline{C}yclic, \underline{S}parse and \underline{D}ense networks, respectively. ``Diff", ``Common"  and ``Average'' summarize the performance on differential, common and average regulatory effects, respectively. FDR and MCC of the naive approach are undefined due to its failure to identify common effects. The sample size $n^{(2)} = n^{(2)}$ is either $200$ or $300$.}
\label{fig:simulation}
\end{figure*}

Here we report on experiments with synthetic data to show the superior performance of our method. We compare the method $\RDnet$ to a naive differential analysis which employs the 2SPLS method proposed by \citep{chen2015two} to construct each network separately. Note that the 2SPLS method is modified here by applying ISIS to screen exogenous variables before conducting ridge regression to predict endogenous variables, making the naive differential analysis comparable to $\RDnet$.

Synthetic data are generated from both acyclic and cyclic networks involving 1000 endogenous variables, with the sample size from 200 to 300. Each network includes a subnetwork of 50 endogenous variables, whose shared and differential structures will be investigated against its pair. On average, each endogenous variable has one regulatory effect in a sparse subnetwork, and three regulatory effects in a dense network. While each pair of subnetworks in comparison share many identical regulatory effects, they also share five regulatory effects but with opposite signs, and each network has five unique regulatory effects (so the total number of differential regulatory effects is 15). The nonzero regulatory effects were independently sampled from a uniform distribution over the range $[-0.8,-0.3]\cup[0.3,8]$. While assuming each node is directly regulated by one exogenous variable, each exogenous variable was sampled from discrete values 0,1 and 2 with probabilities 0.25, 0.5 and 0.25, respectively. All of the noise terms were independently sampled from the normal distribution $N(0,0.1^2)$. We also conducted differential analysis between two networks with both $\bX^{(1)} \ne \bX^{(2)}$ and $\bX^{(1)} = \bX^{(2)}$ as in practice the paired networks may or may not share identically valued exogenous variables.

We evaluate the the performance in terms of the false discovery rate (FDR), power and Matthews correlation coefficient (MCC) \citep{matthews1975comparison}. Let
TP, TN, FP and FN denote the numbers of true positives, true negatives, false positives, and false negatives, respectively.
MCC is defined as,
\begin{equation*}
\text{MCC} = \frac{\text{TP} \times \text{TN} - \text{FP} \times \text{FN}}{\sqrt{  (\text{TP} + \text{FP})(\text{TP} + \text{FN})(\text{TN} + \text{FP})(\text{TN} + \text{FN}) }}.
\end{equation*}
\vskip-10pt
Here we refer nonzero effects as positives and zero effects as negatives. The MCC varies from 0 to 1 with larger values implying better variable selection.

In each differential analysis, the ridge regression employed the generalized cross validation \citep{golub1979generalized} to select the ridge parameter, and the adaptive lasso used 10-fold cross-validation to choose its tuning parameter. Following the recommendation by \citet{fan2008sure}, $(n^{(k)})^{0.9}$ variables are screened by ISIS.

For each type of networks, $100$ synthetic data sets were generated, and the differential analysis results are summarized in Figure~\ref{fig:simulation}. Overall, both $\RDnet$ and the naive approach maintain high power in identifying differential regulatory effects. However, the naive approach fails to identify common regulatory effects and tends to report FDR over 80\% on differential regulatory effects. Such a tendency to report false positives by the naive approach results in lower MCC, with dramatic decrease in identifying differential regulatory effects.

While both methods performed stably across networks with $\bX^{(1)} \ne \bX^{(2)}$ and $\bX^{(1)} = \bX^{(2)}$, $\RDnet$ performed better in identifying differential regulatory effects from dense networks than sparse networks in terms of FDR and MCC. However, the naive approach tends to report even higher FDR and so much lower MCC when identifying differential regulatory effects from dense networks. Nonetheless, the naive approach fails to identify common regulatory effects for each type of networks so the corresponding FDR and MCC are undefined.

We also calculated the standard errors (SE) of the reported FDR, power, and MCC over 100 synthetic data sets (the results are not shown). They are all small with most at the scale of thousandth and others at the scale of hundredth. Therefore, $\RDnet$ performed robustly in differential analysis of networks, and the 2SPLS approach by \citet{chen2015two} performed also robustly in constructing single networks.

\subsection{\uppercase{The Genotype-Tissue Expression Data}}\label{sec:realdata}

\begin{figure*}[ht]
\centering
\subfigure[]{
\includegraphics[width=2.7in,height=1.4in,clip=TRUE]{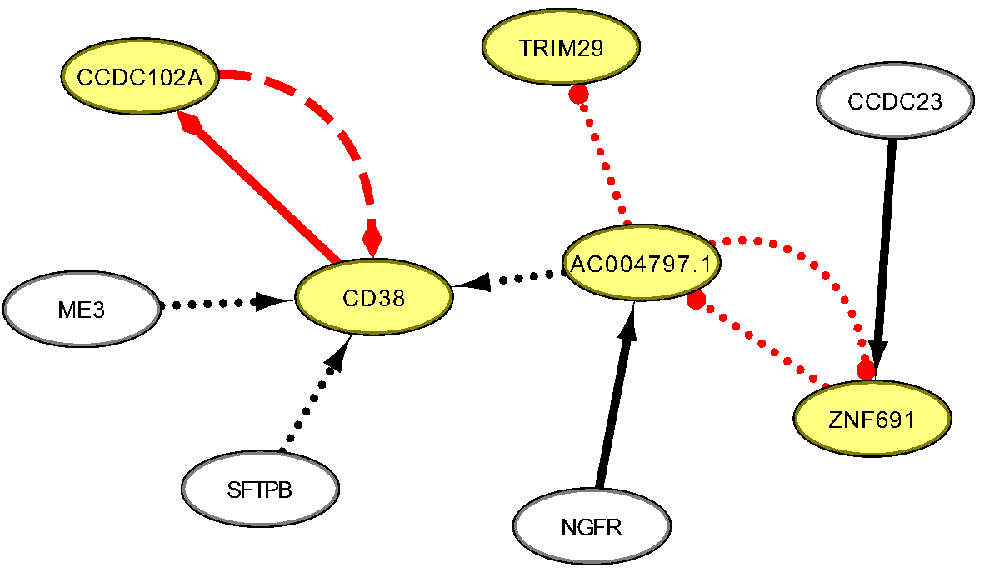}} \hspace{20pt}
\subfigure[]{
\includegraphics[width=2.5in,height=1.4in]{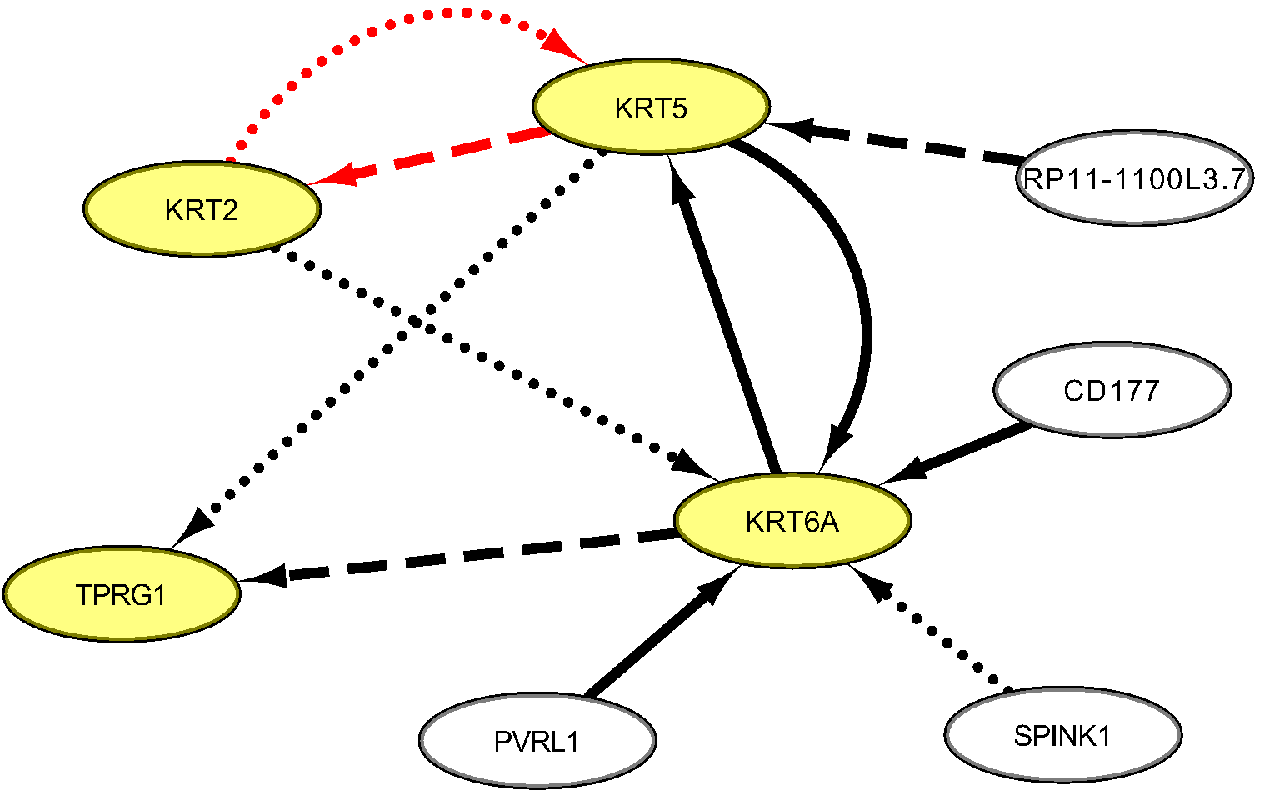}
} \hspace{20pt}
\subfigure[]{
\includegraphics[width=1.6in,height=0.9in,clip=TRUE]{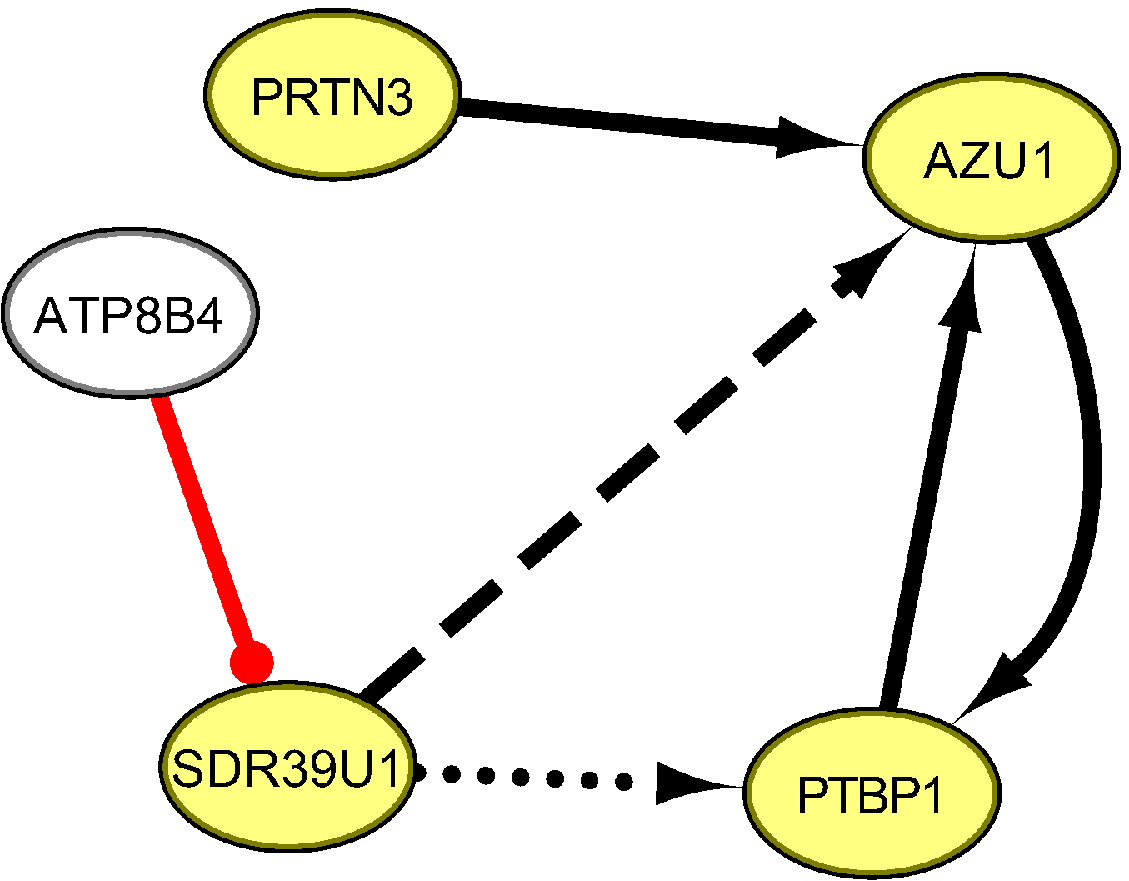}
} \hspace{20pt}
\subfigure[]{
\includegraphics[width=1.35in,height=0.9in]{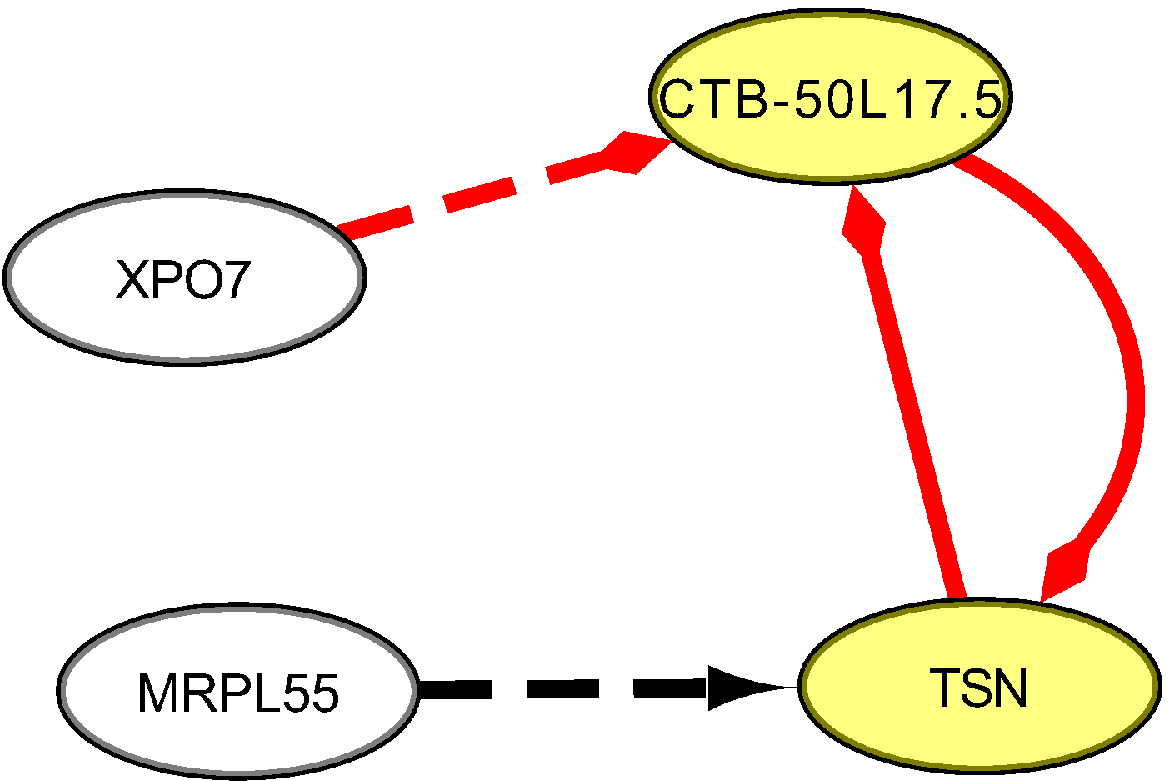}
} \hspace{30pt}
\subfigure[]{\includegraphics[width=1.35in,height=0.9in]{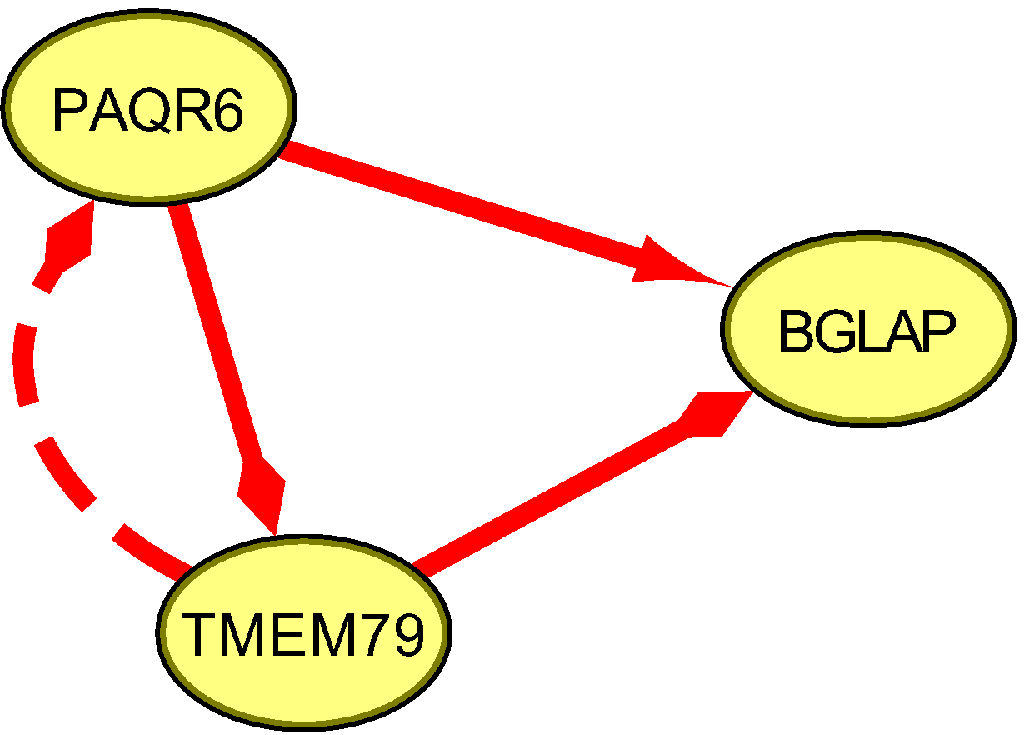}
} \hspace{20pt}
\caption{The Top Five Differential Subnetworks of Gene Regulation Identified by $\RDnet$ From GTEx Data. The dotted, dashed, and solid lines imply regulations constructed in over 70\%, 80\%, and 90\% of the bootstrap data sets, respectively. Highlighted in yellow are the target genes whose regulatory genes are focused in this study. The differential regulations are in red while common regulations are in black. The arrow head implies up regulation in both networks or no regulation in at most one network; the circle head implies down regulation in the whole blood but up regulation in muscle skeletal; and the diamond head implies up regulation in whole blood but down regulation muscle skeletal.}
\label{fig:difnet}
\end{figure*}

We performed differential analysis of gene regulatory networks on two sets of genetic genomics data from the Genotype-Tissue Expression (GTEx) project~\citep{carithers2015novel}, with one collected from human whole blood (WB) and another one from human muscle skeletal (MS). The WB and MS data included genome-wide genetic and genotypic values from 350 and 367 healthy subjects, respectively. Both data sets were preprocessed following \citet{carithers2015novel} and \citet{stegle2010bayesian}, resulting in a total of 15,899 genes and 1,083,917 single nucleotide polymorphisms (SNPs) being shared by WB and MS.

Expression quantitative trait loci (eQTL) mapping \citep{gilad2008revealing} was conducted and identified 9875 genes with at least one marginally significant cis-eQTL (with p-value$< 0.05$). For each gene, we further filtered its set of cis-eQTL by controlling the pairwise correlation under 0.9 and keeping up to three cis-eQTL which have the strongest association with the corresponding gene expression. These cis-eQTL serve as anchoring exogenous variables for the genes, and expression levels of different genes are endogenous variables. At completion of preprocessing data, we have 9,875 endogenous variables and 23,920 exogenous variables.

We applied $\RDnet$ to infer the differential gene regulation on a set of eighty target genes, which had largest changes on gene-gene correlation between the two tissues. We identified a total of 640 common and 572 differential regulations on the eighty target genes. To evaluate the significance of identified regulations, we bootstrapped 100 data sets, and conducted differential analysis on each bootstrap data set. As summarized in Table~\ref{table:realdatadiff}, 50, 43 and 34 differential regulatory effects were identified in over 70\%, 80\% and 90\% of the bootstrap data sets, respectively.

\begin{table}[ht]
\centering
\caption{Summary of Regulations Identified in Over 70\%, 80\%, 90\% of the Bootstrap Data Sets by $\RDnet$ From GTEx Data. Shown under ``Original" are for those identified from the original data.}
\begin{tabular}{ccccc}
\toprule
& Original & 70\% & 80\% & 90\%   \\ \hline
Common & 640 & 49  & 40  & 34  \\
Differential & 572 & 50 & 43 & 34\\ \bottomrule
\end{tabular}
\label{table:realdatadiff}
\end{table}

The top five subnetworks bearing differential regulations on some of the eighty target genes were shown in Figure~\ref{fig:difnet}. We also constructed the differential networks using the naive approach (the results are not shown), and reported more regulations which cover the reported ones by $\RDnet$. This concurs with our observation in the synthetic data evaluation that the naive approach tends to report higher false positives, especially for differential regulatory effects.

\section{\uppercase{Conclusion}} \label{sec:discussion}

We have developed a novel two-stage differential analysis method named $\RDnet$. The first stage, i.e., the calibration stage, aims for good prediction of the endogenous variables, and the second stage, i.e., the construction stage, identifies both common and differential network structures in a node-wise fashion. The key idea of $\RDnet$ method is to appropriately reparametrize the independent models into a joint model so as to estimate differential and common effects directly. This approach can dramatically reduce the false discovery rate. In the experiments with synthetic data, we demonstrated the effectiveness of our method, which outperformed the naive approach with a large margin. Note that $\RDnet$ allows independently conducting all $\ell_2$ regularized regressions at the same time at the first stage, and all $\ell_1$ regularized regressions at the same time at the second stage. Therefore, $\RDnet$ not only permits parallel computation but also allows for fast subnetwork construction to avoid potential huge computational demands from differential analysis of large networks.

There are some interesting directions for future research. Firstly, it is worthwhile to explore other re-parametrization approaches such as baseline reparametrizaiton in a case-control study. Secondly, while we only consider differential analysis of two networks, it is possible to generalize our method to compare multiple networks, demanding more complex reparametrization. Finally, applying the proposed method for fully differential analysis of 53 tissues in the GTEx project still provides challenging computational and methodological issues.

\section*{Acknowledgments}

The Genotype-Tissue Expression (GTEx) Project was supported by the Common Fund of the Office of the Director of the National Institutes of Health, and by NCI, NHGRI, NHLBI, NIDA, NIMH, and NINDS. The data used for the analysis described in this paper were obtained from dbGaP accession number phs000424.v7.p2 on 08/18/2017.

\bibliographystyle{plainnat}
\bibliography{diffnetwork}

\onecolumn
\resetcounters

\begin{center}\Large\bf Supplementary Materials\\
	\vskip12pt
	\large Differential Analysis of Directed Networks
\end{center}
\bigskip

There are five parts. Firstly, we collect in Section~\ref{sec:notations} all notations used in our paper and here. We then describe the four conditions which help define the positive pair $\tilde{\tau}$ and $\tilde{\kappa}$ for Theorem~1, and further prove Theorem~1 in Section~\ref{sec:conditions_theory1}. In Section~\ref{sec:thm2}, we prove Theorem 2 which provides bounds for both estimation and prediction losses at the calibration stage. In Section~\ref{sec:thm3}, we prove Theorem 3 which provides bounds for both estimation and prediction losses at the construction stage. In Section~\ref{sec:thm4}, we prove the variable selection consistency in Theorem 4.

\section{Notations} \label{sec:notations}

Unless otherwise claimed, we will follow the notations defined here throughout the paper and supplementary materials.

For a vector, $||\cdot||_2$ and $||\cdot||_1$ denote the $\ell_2$ and $\ell_1$ norms, respectively; $\|\cdot\|_{\infty}$ and $\|\cdot\|_{-\infty}$ are defined to be the maximum and minimum absolute values of its components, respectively; $|\cdot|_1$ implies taking element-wise absolute values of the vector so is itself a vector. For a matrix $A=(a_{ij})_{m\times n}$, $\|A\|_{1}=\max _{1\leq j\leq n}\sum _{i=1}^{m}|a_{ij}|$, i.e., the maximum column sum of absolute values of its components, and $\|A\|_{\infty }=\max _{1\leq i\leq m}\sum _{j=1}^{n}|a_{ij}|$, i.e., the maximum row sum of absolute values of its components.

For a vector $a$ and index set $\mathcal{S}$, $a_i$ , $a_{-i}$, and $a_{\mathcal{S}}$ denote the $i$-th entry, the subvector excluding the $i$-th entry in $a$, and the subvector of $a$ indexed by $\mathcal{S}$, respectively. For a matrix $A$, $A_{i}$ and $A_{-i}$ denote its $i$-th column and the submatrix of $A$ excluding its $i$-th column, respectively. For a vector $a_i$ and an index set $\mathcal{S}_i$ both sharing the same subscript, the subvector of $a_i$ indexed by $\mathcal{S}_i$ is denoted by $a_{\scriptscriptstyle\mathcal{S}_i}$ for simplicity. Similarly, the submatrix of a matrix $A_i$ including columns indexed by the set $\mathcal{S}_i$ is denoted by $A_{\scriptscriptstyle\mathcal{S}_i}$ for simplicity.

$a \vee b$ and $a\wedge b$ denote the maximum and minimum of $a$ and $b$, respectively. $\lambda_{\min}(\cdot)$ and $\lambda_{\max}(\cdot)$ denote the minimum and maximum eigenvalues of the corresponding matrix, respectively. $\Ex{\cdot}$ denotes the expectation, and $\mathbb{P}(\cdot)$ denotes the probability of an event. Symbol $\asymp$ denotes two terms at the same order. $\tr(\cdot)$ denotes the trace of the corresponding matrix. For a set $S$, $|S|$ denotes the number of its elements. For positive integers $j$ and $p$, $j|p$ denotes the remainder of $j$ when divided by $p$.

Throughout the paper and here, $C_1, C_2, \ldots$, $c_1, c_2, \ldots$, $\tilde{c}_1, \tilde{c}_2, \cdots$, $t_1, t_2, \ldots$ are some positive constant numbers.

\section{The Conditions and Proof of Theorem~1} \label{sec:conditions_theory1}

For each $k \in \{1,2\}$, the reduced model (3) includes $p$ regression models, i.e., for $i=1, 2, \cdots, p$,
\[
\mathbf{Y}_i^{(k)} = \mathbf{X}^{(k)} \boldsymbol{\pi}_i^{(k)} + \boldsymbol{\xi}_i^{(k)}.
\]
Here we first state the four conditions in  Fan and Lv [2008] which restrict the positive pairs $\tau^{(k)}$ and $\kappa^{(k)}$ so as to define $\tilde{\tau}=\max\{\tau^{(1)},\tau^{(2)}\}$ and $\tilde{\kappa}=\max\{\kappa^{(1)},\kappa^{(2)}\}$ for Theorem 1, and then prove that we can successfully screen variables for each of the above linear regression model.
Denote $Y_{ji}^{(k)}$, $X_{jl}^{(k)}$, $\xi_{ji}^{(k)}$, and $\pi_{ji}^{(k)}$ as the $j$-th row of $\mathbf{Y}_{i}^{(k)}$, $\mathbf{X}_{l}^{(k)}$, $\boldsymbol{\xi}_i^{(k)}$, and $\boldsymbol{\pi}_{i}^{(k)}$, respectively. Further denote $\Sigma^{(k)}$ the variance-covariance matrix of the $q$ random variables in observing $\mathbf{X}^{(k)}$. For any $\mathcal{M}\subset\{1, 2,\cdots, q\}$, denote $\Sigma^{(k)}_{\mathcal{M}}$ the variance-covariance matrix of the random variables in observing $\mathbf{X}^{(k)}_{\mathcal{M}}$.

\vskip6pt\noindent
\textbf{Condition 1.} Each $\xi_{ji}^{(k)}$ is normally distributed with mean zero. $(\Sigma^{(k)})^{-1/2} \mathbf{X}^{(k)T}$ is observed from a spherically symmetric distribution, and has the concentration property: there exist some constants $\tilde{c}_1^{(k)},\tilde{c}_2^{(k)} > 1$ and $\tilde{c}_3^{(k)} > 0$ such that, for any $\mathcal{M}\subset\{1, 2,\cdots, q\}$ with $|\mathcal{M}|\ge \tilde{c}_1^{(k)} n^{(k)}$, the eigenvalues of $|\mathcal{M}|^{-1} \mathbf{X}^{(k)}_{\mathcal{M}}  (\Sigma^{(k)}_{\mathcal{M}})^{-1/2} (\Sigma^{(k)T}_{\mathcal{M}})^{-1/2} \mathbf{X}^{(k)T}_{\mathcal{M}} $ are bounded either from above by $\tilde{c}_2^{(k)}$ or from below by $1/\tilde{c}_2^{(k)}$ with probability at least $1-\exp(-\tilde{c}_3^{(k)} n^{(k)})$.

\vskip6pt\noindent
\textbf{Condition 2.} $\text{var}(Y_{ji}^{(k)}) = O(1)$. For some $\kappa^{(k)} \geq 0 $, $\tilde{c}_4^{(k)}>0$, and $\tilde{c}_5^{(k)} > 0$,
\begin{equation*}
	\underset{l \in \mathcal{M}_{i0}^{(k)}}{\text{min}} \left|\boldsymbol{\pi}_{li}^{(k)}\right| \geq \frac{\tilde{c}_4^{(k)}}{(n^{(k)})^{\kappa^{(k)}}} \quad \text{and}\quad \underset{l \in \mathcal{M}_{i0}^{(k)}}{\text{min}} \left|\text{cov}\left( (\boldsymbol{\pi}_{li}^{(k)})^{-1} Y_{ji}^{(k)}, X_{jl}^{(k)}\right)\right| \geq \tilde{c}_5^{(k)}.
\end{equation*}

\vskip6pt\noindent
\textbf{Condition 3.} $\log(q) = O((n^{(k)})^{\tilde{c}})$ for some $\tilde{c} \in (0, 1-2\kappa^{(k)})$.

\vskip6pt\noindent
\textbf{Condition 4.} There are some $\tau^{(k)} \geq 0 $ and $\tilde{c}_6^{(k)} > 0$ such that $\lambda_{\text{max}}(\Sigma^{(k)} ) \leq \tilde{c}_6^{(k)} (n^{(k)})^{\tau^{(k)}}$.
\vskip6pt

\begin{proof}[\textbf{Proof of Theorem~1}]
	Following the {\it Sure Independence Screening Property} by Fan and Lv [2008], there exists some $\theta^{(k)}\in (0, 1-2\kappa^{(k)}-\tau^{(k)})$ such that, when $d^{(k)} =|\mathcal{M}_{i}^{(k)}|= O((n^{(k)})^{1-\theta^{(k)}})$, we have, for some constant $C>0$,
	$$\scriptsize
	\mathbb{P}( \mathcal{M}^{(k)}_{i0} \subseteq \mathcal{M}^{(k)}_{i} ) = 1 - \mathcal{O}\left(\exp\left\{-\frac{C (n^{(k)})^{1-2\kappa^{(k)}}}{\log(n^{(k)})}\right\}\right).
	$$
	Let $\theta = \min(\theta^{(1)},\theta^{(2)})$, then for $d^{(k)}=|\mathcal{M}_{i}^{(k)}|\equiv  d=O(n_{\min}^{1-\theta})$, we have
	$$\scriptsize
	\mathbb{P}( \mathcal{M}^{(k)}_{i0} \subseteq \mathcal{M}^{(k)}_{i} ) = 1 - \mathcal{O}\left(\exp\left\{-\frac{C (n^{(k)})^{1-2\tilde{\kappa}}}{\log(n^{(k)})}\right\}\right).
	$$
\end{proof}

\section{Proof of Theorem 2} \label{sec:thm2}

Note that $\boldsymbol{\xi}^{(k)} = \mathcal{E}^{(k)} (\textbf{I}-\boldsymbol{\Gamma}^{(k)})^{-1}$ for $k\in\{1, 2\}$. Suppose that the singular values of both $(\textbf{I}-\boldsymbol{\Gamma}^{(k)})$ are positively bounded from below by a constant $c$. Denote $\sigma_i^{(k)2}=\var(\epsilon_{ji}^{(k)})$ and $\tilde{\sigma}_i^{(k)2}=\var(\xi_{ji}^{(k)})$. Then $\tilde{\sigma}_i^{(k)} \le \sigma_{p\max}/c  = \underset{1\le i \le p}{\max}  (\sigma_i^{(1)} \lor \sigma_i^{(2)})$/c.

\begin{lemma}\label{proposition:byNodebyNetStep1}
	Under Assumptions 1-3, for each network $k \in \{1,2\}$ in the calibration step, there exist positive constants $C_1^{(k)}$ and $C_2^{(k)}$ such that, with probability at least $1-e^{-f^{(k)}}$,
	\begin{enumerate}
		\item (Estimation Loss) $\ltwon{\boldsymbol{\hat{\pi}}^{(k)}_{i} - \boldsymbol{\pi}^{(k)}_{i}}^2 \le C_1^{(k)} \left(r_{i}^{(k)} \lor d \lor f^{(k)}\right)\big/n^{(k)}$;
		\item (Prediction Loss) $\ltwon{\bX^{(k)}(\boldsymbol{\hat{\pi}}^{(k)}_{i } - \boldsymbol{\pi}^{(k)}_{i })}^2 \big/ n^{(k)} \le C_2^{(k)}\left(r_{i}^{(k)} \lor d \lor f^{(k)}\right)\big/n^{(k)}$.
	\end{enumerate}
\end{lemma}

\begin{proof}[\textbf{Proof of Lemma \ref{proposition:byNodebyNetStep1}}]
	We have the closed form ridge estimator $\hat{\boldsymbol{\pi}}_{ \mathcal{M}_i^{(k)}}^{(k)}$ for the linear model $\mathbf{Y}_i^{(k)} = \mathbf{X}_{ \mathcal{M}_i^{(k)}}^{(k)} \boldsymbol{\pi}_{ \mathcal{M}_i^{(k)}}^{(k)} + \boldsymbol{\xi}_i^{(k)}$.
	\begin{equation*}
		\hat{\boldsymbol{\pi}}_{ \mathcal{M}_i^{(k)}}^{(k)}= \bigl(\mathbf{X}_{ \mathcal{M}_i^{(k)}}^{(k)T}\mathbf{X}_{ \mathcal{M}_i^{(k)}}^{(k)} +\lambda^{(k)}_{i}I_d\bigr)^{-1}\mathbf{X}_{ \mathcal{M}_i^{(k)}}^{(k)T} \mathbf{Y}_{i}^{(k)},
	\end{equation*}
	where $\lambda^{(k)}_{i}$ is the ridge tuning parameter. Plugging in the equation $\mathbf{Y}^{(k)}_{i} = \mathbf{X}_{ \mathcal{M}_i^{(k)}}^{(k)}\boldsymbol{\pi}_{ \mathcal{M}_i^{(k)}}^{(k)} +\boldsymbol{\xi}_i^{(k)} $, we have
	\begin{equation*}
		\begin{aligned}
			\hat{\boldsymbol{\pi}}_{ \mathcal{M}_i^{(k)}}^{(k)}= & \left\{\bigl(\mathbf{X}_{ \mathcal{M}_i^{(k)}}^{(k)T} \mathbf{X}_{ \mathcal{M}_i^{(k)}}^{(k)} +\lambda^{(k)}_{i}I_d\bigr)^{-1} \mathbf{X}_{\mathcal{M}_i^{(k)}}^{(k)T} \mathbf{X}_{\mathcal{M}_i^{(k)}}^{(k)} \boldsymbol{\pi}_{\mathcal{M}_i^{(k)}}^{(k)} \right\}\\
			&+ \left \{\bigl(\mathbf{X}_{\mathcal{M}_i^{(k)}}^{(k)T} \mathbf{X}_{ \mathcal{M}_i^{(k)}}^{(k)} +\lambda^{(k)}_{i}I_d\bigr)^{-1} \mathbf{X}_{\mathcal{M}_i^{(k)}}^{(k)T} \boldsymbol{\xi}_i^{(k)} \right\}.
		\end{aligned}
	\end{equation*}
	The difference between the ridge estimator $\hat{\boldsymbol{\pi}}_{ \mathcal{M}_i^{(k)}}^{(k)}$ and the true $\boldsymbol{\pi}_{ \mathcal{M}_i^{(k)}}^{(k)}$ can be written as
	\begin{equation*}
		\begin{aligned}
			\hat{\boldsymbol{\pi}}_{\mathcal{M}_i^{(k)}}^{(k)}- \boldsymbol{\pi}_{ \mathcal{M}_i^{(k)}}^{(k)} &= -\lambda^{(k)}_{i}\bigl(\mathbf{X}_{ \mathcal{M}_i^{(k)}}^{(k)T} \mathbf{X}_{\mathcal{M}_i^{(k)}}^{(k)} +\lambda^{(k)}_{i}I_d\bigr)^{-1} \boldsymbol{\pi}_{\mathcal{M}_i^{(k)}}^{(k)} + \bigl(\mathbf{X}_{\mathcal{M}_i^{(k)}}^{(k)T} \mathbf{X}_{\mathcal{M}_i^{(k)}}^{(k)} +\lambda^{(k)}_{i}I_d\bigr)^{-1} \mathbf{X}_{ \mathcal{M}_i^{(k)}}^{(k)T} \boldsymbol{\xi}_i^{(k)}.
		\end{aligned}
	\end{equation*}
	
	For simplicity, we denote the composite forms of  $\boldsymbol{\pi}_{ \mathcal{M}_i^{(k)}}^{(k)}$ and $\mathbf{X}_{ \mathcal{M}_i^{(k)}}^{(k)}$ as follows,
	\begin{eqnarray*}
		&&\tilde{\boldsymbol{\pi}}_{ \mathcal{M}_i^{(k)}}^{(k)} = -\lambda^{(k)}_{i} \bigl(\mathbf{X}_{\mathcal{M}_i^{(k)}}^{(k)T} \mathbf{X}_{\mathcal{M}_i^{(k)}}^{(k)} +\lambda^{(k)}_{i}I_d\bigr)^{-1}\boldsymbol{\pi}_{\mathcal{M}_i^{(k)}}^{(k)}; \\
		&&\tilde{\mathbf{X}}_{ \mathcal{M}_i^{(k)}}^{(k)} = \mathbf{X}_{\mathcal{M}_i^{(k)}}^{(k)} \bigl(\mathbf{X}_{\mathcal{M}_i^{(k)}}^{(k)T} \mathbf{X}_{\mathcal{M}_i^{(k)}}^{(k)}+\lambda^{(k)}_{i}I_d\bigr)^{-1}.
	\end{eqnarray*}
	Then we have the following simplified form of the difference,
	\begin{equation*}
		\hat{\boldsymbol{\pi}}_{ \mathcal{M}_i^{(k)}}^{(k)}- \boldsymbol{\pi}_{ \mathcal{M}_i^{(k)}}^{(k)} = \tilde{\boldsymbol{\pi}}_{\mathcal{M}_i^{(k)}}^{(k)} + \tilde{\mathbf{X}}_{ \mathcal{M}_i^{(k)}}^{(k)T} \boldsymbol{\xi}_i^{(k)}.
	\end{equation*}
	
	To obtain the $\ell_2$ norm losses of estimation and prediction, we write
	\begin{eqnarray*}
		\lefteqn{\ltwon{\hat{\boldsymbol{\pi}}_{\mathcal{M}_i^{(k)}}^{(k)}- \boldsymbol{\pi}_{ \mathcal{M}_i^{(k)}}^{(k)}}^2} \\
		&=& \underbrace{\tilde{\boldsymbol{\pi}}_{\mathcal{M}_i^{(k)}}^{(k)T} \tilde{\boldsymbol{\pi}}_{\mathcal{M}_i^{(k)}}^{(k)}}_{T_{21}}  + \underbrace{2\tilde{\boldsymbol{\pi}}_{\mathcal{M}_i^{(k)}}^{(k)T} \tilde{\mathbf{X}}_{\mathcal{M}_i^{(k)}}^{(k)T} \boldsymbol{\xi}_i^{(k)}}_{T_{22}} + \underbrace{\boldsymbol{\xi}_i^{(k)T} \tilde{\mathbf{X}}_{\mathcal{M}_i^{(k)}}^{(k)} \tilde{\mathbf{X}}_{ \mathcal{M}_i^{(k)}}^{(k)T} \boldsymbol{\xi}_i^{(k)}}_{T_{23}},\\
		\lefteqn{\ltwon{\mathbf{X}_{\mathcal{M}_i^{(k)}}^{(k)}\bigl(\hat{\boldsymbol{\pi}}_{ \mathcal{M}_i^{(k)}}^{(k)}- \boldsymbol{\pi}_{ \mathcal{M}_i^{(k)}}^{(k)}\bigr)}^2} \\
		&=&\underbrace{\boldsymbol{\pi}_{\mathcal{M}_i^{(k)}}^{(k)} \bigl(\mathbf{X}_{\mathcal{M}_i^{(k)}}^{(k)T} \mathbf{X}_{\mathcal{M}_i^{(k)}}^{(k)}\bigr) \tilde{\boldsymbol{\pi}}_{\mathcal{M}_i^{(k)}}^{(k)}}_{T_{24}} + \underbrace{2\tilde{\boldsymbol{\pi}}_{\mathcal{M}_i^{(k)}}^{(k)T} \bigl(\mathbf{X}_{\mathcal{M}_i^{(k)}}^{(k)T} \mathbf{X}_{\mathcal{M}_i^{(k)}}^{(k)}\bigr) \tilde{\mathbf{X}}_{ \mathcal{M}_i^{(k)}}^{(k)T} \boldsymbol{\xi}_i^{(k)}}_{T_{25}} \\
		&& + \underbrace{\boldsymbol{\xi}_i^{(k)T} \tilde{\mathbf{X}}_{\mathcal{M}_i^{(k)}}^{(k)} \bigl(\mathbf{X}_{\mathcal{M}_i^{(k)}}^{(k)T} \mathbf{X}_{\mathcal{M}_i^{(k)}}^{(k)}\bigr) \tilde{\mathbf{X}}_{ \mathcal{M}_i^{(k)}}^{(k)T} \boldsymbol{\xi}_i^{(k)}}_{T_{26}}.
	\end{eqnarray*}
	
	Firstly, we will derive the bound for $T_{24}, T_{25}$ and $T_{26}$ terms, then we can obtain similar results for term $T_{21}, T_{22}$ and $T_{23}$ by simply removing the matrix $\mathbf{X}_{\mathcal{M}_i^{(k)}}^{(k)T} \mathbf{X}_{\mathcal{M}_i^{(k)}}^{(k)}$. Denote the singular value decomposition $\mathbf{X}_{\mathcal{M}_i^{(k)}}^{(k)T}\mathbf{X}_{\mathcal{M}_i^{(k)}}^{(k)} = U^{(k)T}_{i} V^{(k)}_{i} U^{(k)}_{i}$, where $U^{(k)}_{i}$ is a unitary matrix, $V^{(k)}_{i}$ is a diagonal matrix with eigenvalues $v_i$. Therefore, the shared component of $\boldsymbol{\tilde{\pi}}_{\mathcal{M}_i^{(k)}}^{(k)}$ and $\tilde{\mathbf{X}}_{\mathcal{M}_i^{(k)}}^{(k)}$ can be rewritten as
	\begin{equation*}
		\bigl(\mathbf{X}_{\mathcal{M}_i^{(k)}}^{(k)T} \mathbf{X}_{\mathcal{M}_i^{(k)}}^{(k)} +\lambda^{(k)}_{i}I_d\bigr)^{-1} = U^{(k)T}_{i} \bigl(V^{(k)}_{i} + \lambda^{(k)}_{i} I_d\bigr)^{-1}U^{(k)}_{i}.
	\end{equation*}
	
	By Assumption~3, there are some constants $c_1,c_2$ such that $\max_{\ltwon{\delta}=1} (n^{(k)})^{-1/2}\ltwon{\mathbf{X}_{\mathcal{M}_i^{(k)}}^{(k)}\delta} \le c_1$ and $\min_{\ltwon{\delta}=1} (n^{(k)})^{-1/2}\ltwon{\mathbf{X}_{ \mathcal{M}_i^{(k)}}^{(k)}\delta} \ge c_2$. Thus, $\lambda_{\max}\bigl(\mathbf{X}_{\mathcal{M}_i^{(k)}}^{(k)T} \mathbf{X}_{\mathcal{M}_i^{(k)}}^{(k)}\bigr) < c_1^2 n^{(k)}$ and $\lambda_{\min}\bigl(\mathbf{X}_{\mathcal{M}_i^{(k)}}^{(k)T} \mathbf{X}_{\mathcal{M}_i^{(k)}}^{(k)}\bigr) > c_2^2 n^{(k)}$. That is, $v_j \asymp n^{(k)}$ for each eigenvalue. Let $b = U^{(k)}_{i}\boldsymbol{\boldsymbol{\pi}}_{\mathcal{M}_i^{(k)}}^{(k)}$, then $\ltwon{b} = \ltwon{\boldsymbol{\pi}_{\mathcal{M}_i^{(k)}}^{(k)}}$. Noting that $\lambda^{(k)}_{i} = o(n^{(k)})$ in Assumption~3, we can bound the term $T_{24}$ as follows,
	\begin{equation} \label{bound:t4}
		\begin{aligned}
			T_{24} &= \tilde{\boldsymbol{\pi}}_{ \mathcal{M}_i^{(k)}}^{(k)T} \bigl(\mathbf{X}_{\mathcal{M}_i^{(k)}}^{(k)T} \mathbf{X}_{\mathcal{M}_i^{(k)}}^{(k)}\bigr) \tilde{\boldsymbol{\pi}}_{\mathcal{M}_i^{(k)}}^{(k)} = \lambda^{(k)2}_{i} b^T V^{(k)}_{i}\bigl(V^{(k)}_{i}+\lambda^{(k)}_{i}I_d\bigr)^{-1}V^{(k)}_{i}\bigl(V^{(k)}_{i}+ \lambda^{(k)}_{i}I_d\bigr)^{-1}b\\
			&= \lambda^{(k)2}_{i} \sum_{j=1}^{d} \frac{v_j b_{ij}^2}{\bigl(v_j+\lambda^{(k)}_{i}\bigr)^2} = \bigO{\lambda^{(k)2}_{i} \ltwon{\boldsymbol{\pi}^{(k)}_{\mathcal{M}_i^{(k)}}}^2\big/n^{(k)}} = \bigO{r^{(k)}_i}.
		\end{aligned}
	\end{equation}
	Similarly, removing the term $\mathbf{X}_{\mathcal{M}_i^{(k)}}^{(k)T} \mathbf{X}_{\mathcal{M}_i^{(k)}}^{(k)}$, we have
	\begin{equation} \label{bound:t1}
		T_{21} = \bigO{\lambda^{(k)2}_{i} \ltwon{\boldsymbol{\pi}_{\mathcal{M}_i^{(k)}}^{(k)}}^2\big/n^{(k)}} = \bigO{r_{i}^{(k)}\big/n^{(k)}}.
	\end{equation}
	
	Noting that $T_{25}$ follows a Gaussian distribution, we can write the probability of deviation of $T_{25}$ with the classical Gaussian tail inequality, for any positive number $t$,
	\begin{equation*}
		\mathbb{P}\left( T_{25} \le t \right) \ge 1- \exp\left(-\frac{1}{2} t^2\big/\var(T_{25})\right).
	\end{equation*}
	Furthermore,
	\begin{eqnarray*}
		\var(T_{25}) &=& 4\tilde{\sigma}^{(k)2}_{i} \boldsymbol{\tilde{\pi}^{(k)T}_{(i)}} \bigl(\mathbf{X}_{\mathcal{M}_i^{(k)}}^{(k)T}\mathbf{X}_{\mathcal{M}_i^{(k)}}^{(k)}\bigr) \tilde{\mathbf{X}}_{\mathcal{M}_i^{(k)}}^{(k)T}  \tilde{\mathbf{X}}_{\mathcal{M}_i^{(k)}}^{(k)} \bigl(\mathbf{X}_{\mathcal{M}_i^{(k)}}^{(k)T} \mathbf{X}_{\mathcal{M}_i^{(k)}}^{(k)}\bigr) \tilde{\boldsymbol{\pi}}_{\mathcal{M}_i^{(k)}}^{(k)} \\
		&=& 4\tilde{\sigma}^{(k)2}_{i} \lambda^{(k)2}_{i} b^T (V +\lambda^{(k)}_{i}I_d)^{-1}V^{(k)}_{i}(V^{(k)}_{i}+\lambda^{(k)}_{i}I_d)^{-1} \\
		&& \times V^{(k)}_{i}\bigl(V^{(k)}_{i}+\lambda^{(k)}_{i}I_d\bigr)^{-1}V^{(k)}_{i}\bigl(V^{(k)}_{i}+\lambda^{(k)}_{i}I_d\bigr)^{-1}b \\
		&=& 4\tilde{\sigma}^{(k)2}_{i} \lambda^{(k)2}_{i} \sum_{j=1}^{d}\frac{v_j^3 b_{ij}^2}{(v_j+\lambda^{(k)}_{i})^4}=\bigO{\tilde{\sigma}^{(k)2}_{i} \lambda^{(k)2}_{i}\ltwon{\boldsymbol{\pi}_{\mathcal{M}_i^{(k)}}^{(k)}}^2\big/n^{(k)}} = \bigO{\tilde{\sigma}^{(k)2}_{i} r_{i}^{(k)}}.
	\end{eqnarray*}
	Letting $t = \sqrt{2 \var(T_{25}) (f^{(k)} +\log 2)}$, we obtain that, with probability at least $1- e^{-f^{(k)}}/2$,
	\begin{equation} \label{bound:t5}
		T_{25} = \bigO{\sqrt{r_{i}^{(k)} f^{(k)}}}.
	\end{equation}	
	Similarly, removing $\mathbf{X}_{\mathcal{M}_i^{(k)}}^{(k)T} \mathbf{X}_{\mathcal{M}_i^{(k)}}^{(k)}$, we can obtain that, concurring with (\ref{bound:t5}),
	\begin{equation} \label{bound:t2}
		T_{22} = \bigO{\sqrt{r_{i}^{(k)} f^{(k)}}\big/n^{(k)}}.
	\end{equation}
	
	The term $T_{26}$ follows a non-central $\chi^2$ distribution. We can invoke the Hanson-Wright inequality [Rudelson et al., 2013] to bound the probability of its extreme deviation, for some constant $t_2>0$, 
	\begin{eqnarray}\label{ineq:Ta6_1}
		\lefteqn{\mathbb{P}(T_{26} \le \Ex{T_{26} } + t)} \nonumber \\ &\ge& 1- \exp\left\{ \frac{-t^2 t_2}{\tilde{\sigma}^{(k)4}_{i} \frobn{\tilde{\mathbf{X}}_{ \mathcal{M}_i^{(k)}}^{(k)} \mathbf{X}_{\mathcal{M}_i^{(k)}}^{(k)T} \mathbf{X}_{\mathcal{M}_i^{(k)}}^{(k)} \tilde{\mathbf{X}}_{\mathcal{M}_i^{(k)}}^{(k)T}}^2} \wedge \frac{-t t_2}{\tilde{\sigma}^{(k)2}_{i} \opn{\tilde{\mathbf{X}}_{\mathcal{M}_i^{(k)}}^{(k)} \mathbf{X}_{\mathcal{M}_i^{(k)}}^{(k)T} \mathbf{X}_{\mathcal{M}_i^{(k)}}^{(k)} \tilde{\mathbf{X}}_{\mathcal{M}_i^{(k)}}^{(k)T}}} \right\}.
	\end{eqnarray}
	To understand this probabilistic bound, we need to calculate $\Ex{T_{26}}$ and the two involved norms. Firstly,
	\begin{eqnarray} \label{ineq:Ta6_2}
		\Ex{T_{26}} &=& \tilde{\sigma}^{(k)2}_{i} \tr\left(\tilde{\mathbf{X}}_{\mathcal{M}_i^{(k)}}^{(k)} \mathbf{X}_{ \mathcal{M}_i^{(k)}}^{(k)T} \mathbf{X}_{ \mathcal{M}_i^{(k)}}^{(k)} \tilde{\mathbf{X}}_{ \mathcal{M}_i^{(k)}}^{(k)T}\right) \nonumber\\
		&=& \tilde{\sigma}^{(k)2}_{i} \tr\left(V^{(k)}_{i}\bigl(V^{(k)}_{i}+\lambda^{(k)}_{i}I_d\bigr)^{-1}V^{(k)}_{i}\bigl(V^{(k)}_{i}+\lambda^{(k)}_{i}I_d\bigr)^{-1}\right) \nonumber\\
		&=& \tilde{\sigma}^{(k)2}_{i} \sum_{j=1}^{d}\frac{v_{j}^2}{(v_{j}+\lambda^{(k)}_{i})^{2}} = \bigO{d \tilde{\sigma}^{(k)2}_{i}}.
	\end{eqnarray}
	The Frobenius norm can be simplified as follows,
	\begin{eqnarray} \label{ineq:Ta6_3}
		\lefteqn{\frobn{\tilde{\mathbf{X}}_{\mathcal{M}_i^{(k)}}^{(k)} \mathbf{X}_{\mathcal{M}_i^{(k)}}^{(k)T}\mathbf{X}_{\mathcal{M}_i^{(k)}}^{(k)} \tilde{\mathbf{X}}_{\mathcal{M}_i^{(k)}}^{(k)T}}^2} \nonumber\\
		&=& \tr\left(\tilde{\mathbf{X}}_{\mathcal{M}_i^{(k)}}^{(k)} \mathbf{X}_{\mathcal{M}_i^{(k)}}^{(k)T} \mathbf{X}_{\mathcal{M}_i^{(k)}}^{(k)} \tilde{\mathbf{X}}_{\mathcal{M}_i^{(k)}}^{(k)T} \tilde{\mathbf{X}}_{\mathcal{M}_i^{(k)}}^{(k)} \mathbf{X}_{\mathcal{M}_i^{(k)}}^{(k)T} \mathbf{X}_{\mathcal{M}_i^{(k)}}^{(k)} \tilde{\mathbf{X}}_{\mathcal{M}_i^{(k)}}^{(k)T} \right)  \nonumber\\
		&=& \tr\left(((\mathbf{X}_{ \mathcal{M}_i^{(k)}}^{(k)})^T\mathbf{X}_{ \mathcal{M}_i^{(k)}}^{(k)}) (\tilde{\mathbf{X}}_{ \mathcal{M}_i^{(k)}}^{(k)})^T \tilde{\mathbf{X}}_{ \mathcal{M}_i^{(k)}}^{(k)} ((\mathbf{X}_{ \mathcal{M}_i^{(k)}}^{(k)})^T\mathbf{X}_{ \mathcal{M}_i^{(k)}}^{(k)}) (\tilde{\mathbf{X}}_{ \mathcal{M}_i^{(k)}}^{(k)})^T\tilde{\mathbf{X}}_{ \mathcal{M}_i^{(k)}}^{(k)}\right)  \nonumber\\
		&=& \tr\left(V^{(k)}_{i}(V^{(k)}_{i}+\lambda^{(k)}_{i}I_d)^{-1} V^{(k)}_{i}(V^{(k)}_{i}+\lambda^{(k)}_{i}I_d)^{-1}V^{(k)}_{i} (V^{(k)}_{i}+\lambda^{(k)}_{i}I_d)^{-1}V^{(k)}_{i}(V^{(k)}_{i}+\lambda^{(k)}_{i}I_d)^{-1}\right) \nonumber\\
		&=& \sum_{j=1}^{d}\frac{v_{j}^4}{(v_{j}+\lambda^{(k)}_{i})^{4}} = \bigO{d}.
	\end{eqnarray}
	Note that $\lambda_{\text{max}}(\mathbf{X}_{\mathcal{M}_i^{(k)}}^{(k)} \mathbf{X}_{\mathcal{M}_i^{(k)}}^{(k)T} \mathbf{X}_{\mathcal{M}_i^{(k)}}^{(k)} \mathbf{X}_{ \mathcal{M}_i^{(k)}}^{(k)T}) \asymp n^{(k)}$, then, the operator norm can be simplified as follows,
	\begin{eqnarray} \label{ineq:Ta6_4}
		\lefteqn{\opn{\tilde{\mathbf{X}}_{\mathcal{M}_i^{(k)}}^{(k)} \mathbf{X}_{\mathcal{M}_i^{(k)}}^{(k)T} \mathbf{X}_{\mathcal{M}_i^{(k)}}^{(k)} \tilde{\mathbf{X}}_{\mathcal{M}_i^{(k)}}^{(k)T}} } \nonumber\\
		&=& \opn{\mathbf{X}_{\mathcal{M}_i^{(k)}}^{(k)}\bigl(\mathbf{X}_{\mathcal{M}_i^{(k)}}^{(k)T} \mathbf{X}_{\mathcal{M}_i^{(k)}}^{(k)}+\lambda^{(k)}_{i}I_d\bigr)^{-1} \mathbf{X}_{\mathcal{M}_i^{(k)}}^{(k)T} \mathbf{X}_{\mathcal{M}_i^{(k)}}^{(k)}\bigl(\mathbf{X}_{\mathcal{M}_i^{(k)}}^{(k)T} \mathbf{X}_{\mathcal{M}_i^{(k)}}^{(k)}+\lambda^{(k)}_{i}I_d\bigr)^{-1} \mathbf{X}_{ \mathcal{M}_i^{(k)}}^{(k)T}}  \nonumber\\
		&=& \bigO{\lambda_{\max}\bigl(\mathbf{X}_{\mathcal{M}_i^{(k)}}^{(k)} \mathbf{X}_{\mathcal{M}_i^{(k)}}^{(k)T} \mathbf{X}_{\mathcal{M}_i^{(k)}}^{(k)} \mathbf{X}_{\mathcal{M}_i^{(k)}}^{(k)T}\bigr)\big/n^{(k)2}} = \bigO{1}.
	\end{eqnarray}
	Letting
	\begin{equation*}
		\begin{aligned}
			t = & \sqrt{\tilde{\sigma}^{(k)4}_{i} \frobn{ \tilde{\mathbf{X}}_{\mathcal{M}_i^{(k)}}^{(k)} \mathbf{X}_{\mathcal{M}_i^{(k)}}^{(k)T} \mathbf{X}_{\mathcal{M}_i^{(k)}}^{(k)} \tilde{\mathbf{X}}_{\mathcal{M}_i^{(k)}}^{(k)T}}^2\times (f^{(k)}+\log 2 )/t_2}\\
			&\lor \left( \tilde{\sigma}^{(k)2}_{i} \opn{\tilde{\mathbf{X}}_{\mathcal{M}_i^{(k)}}^{(k)} \mathbf{X}_{\mathcal{M}_i^{(k)}}^{(k)T} \mathbf{X}_{\mathcal{M}_i^{(k)}}^{(k)} \tilde{\mathbf{X}}_{\mathcal{M}_i^{(k)}}^{(k)T}}\times (f^{(k)} + \log 2)/t_2 \right),
		\end{aligned}
	\end{equation*}
	and combining (\ref{ineq:Ta6_1}), (\ref{ineq:Ta6_2}), (\ref{ineq:Ta6_3}), and (\ref{ineq:Ta6_4}), we obtain that, with probability at least $1-e^{-f^{(k)} }/2$,
	\begin{equation} \label{bound:t6}
		T_{26} = \bigO{d \lor \sqrt{d f^{(k)}} \lor f^{(k)}}.
	\end{equation}
	Similarly, removing $\mathbf{X}_{\mathcal{M}_i^{(k)}}^{(k)T} \mathbf{X}_{\mathcal{M}_i^{(k)}}^{(k)}$, we can obtain that, concurring with (\ref{bound:t6}),
	\begin{equation} \label{bound:t3}
		T_{23} = \bigO{\bigl(d \lor  \sqrt{ d\, f^{(k)} } \lor f^{(k)}\bigr) \big/n^{(k)}}.
	\end{equation}
	Collecting the bounds (\ref{bound:t4}), (\ref{bound:t5}), (\ref{bound:t6}) and noting the definition of $\mathbf{X}_{ \mathcal{M}_i^{(k)}}^{(k)}$ and $\boldsymbol{\pi}_{ \mathcal{M}_i^{(k)}}^{(k)}$,  we conclude there exists some constant $C_2^{(k)}>0$ such that, with probability at least $1-e^{-f^{(k)}}$,
	\begin{equation*}
		\frac{1}{n^{(k)}}\ltwon{\mathbf{X}^{(k)}( \hat{\boldsymbol{\pi}}^{(k)}_i- \boldsymbol{\pi}^{(k)}_i  ) }^2=\frac{1}{n^{(k)}}\ltwon{\mathbf{X}_{ \mathcal{M}_i^{(k)}}^{(k)}( \hat{\boldsymbol{\pi}}_{ \mathcal{M}_i^{(k)}}^{(k)} - \boldsymbol{\pi}_{ \mathcal{M}_i^{(k)}}^{(k)} ) }^2 \le C_{2}^{(k)} \frac{r^{(k)}_{i} \lor d \lor f^{(k)}}{n^{(k)}}.
	\end{equation*}
	Similarly, collecting the bound (\ref{bound:t1}), (\ref{bound:t2}) and (\ref{bound:t3}), we conclude there exists some constant $C_1^{(k)}>0$ such that, with probability at least $1-e^{-f^{(k)}}$,
	\begin{equation*}
		\ltwon{\hat{\boldsymbol{\pi}}^{(k)}_i - \boldsymbol{\pi}^{(k)}_i }^2=
		\ltwon{\hat{\boldsymbol{\pi}}_{\mathcal{M}_i^{(k)}}^{(k)} - \boldsymbol{\pi}_{\mathcal{M}_i^{(k)}}^{(k)}}^2 \le C_1^{(k)} \frac{r_i^{(k)} \lor d \lor f^{(k)} }{n^{(k)}}.
	\end{equation*}
	This concludes the proof of Lemma~\ref{proposition:byNodebyNetStep1}.
\end{proof}

To bound the estimation loss, we write
\begin{equation*}
	\ltwon{\hat{\boldsymbol{\Pi}}_{j} - \boldsymbol{\Pi}_{j}}^2  = \ltwon{\hat{\boldsymbol{\pi}}^{(1)}_{j|p} - \boldsymbol{\pi}^{(1)}_{j|p}}^2 + \ltwon{\hat{\boldsymbol{\pi}}^{(2)}_{j|p} - \boldsymbol{\pi}^{(2)}_{j|p}}^2,
\end{equation*}
where $\boldsymbol{\pi}^{(k)}_{j|p}$ and $\hat{\boldsymbol{\pi}}^{(k)}_{j|p}$ are the ${j|p}$ columns of $\boldsymbol{\pi}^{(k)}$ and $\hat{\boldsymbol{\pi}}^{(k)}$, respectively. Following the bounds in Lemma~\ref{proposition:byNodebyNetStep1} for both networks, we obtain the overall estimation bound as, with probability at least $1 - e^{-f^{(1)}} - e^{-f^{(2)}}$,
\begin{equation*}
	\begin{aligned}
		\ltwon{\hat{\boldsymbol{\Pi}}_{j} - \boldsymbol{\Pi}_{j}}^2 &\le C_1^{(1)} \frac{r_{j|p}^{(1)} \lor d \lor f^{(1)}}{n^{(1)}} + C_1^{(2)} \frac{r_{j|p}^{(2)} \lor d \lor f^{(2)}}{n^{(2)}} \\
		& \le \bigl(C_1^{(1)} + C_1^{(2)}\bigr) \frac{\bigl(r_{j|p}^{(2)} \lor d \lor f^{(2)}\bigr) \lor \bigl(r_{j|p}^{(2)} \lor d \lor f^{(2)}\bigr) }{n^{(1)}\land n^{(2)}}\\
		& = C_1 \frac{d\lor \bigl(r_{j|p}^{(1)}\lor r_{j|p}^{(2)}\bigr) \lor \bigl(f^{(1)}\lor f^{(2)}\bigr)}{n^{(1)}\land n^{(2)}} \le C_1 \frac{d\lor r_{\max} \lor f_{\max}}{n^{(1)}\land n^{(2)}},
	\end{aligned}
\end{equation*}
where $C_1 = C_1^{(1)} + C_1^{(2)}$. Similarly, we write the prediction bound as, with probability at least $1 - e^{-f^{(1)}} - e^{-f^{(2)}}$,
\begin{equation*}
	\begin{aligned}
		\ltwon{\mathbf{X}(\hat{\boldsymbol{\Pi}}_{j} - \boldsymbol{\Pi}_{j})}^2  &= \ltwon{X^{(1)}(\hat{\boldsymbol{\pi}}^{(1)}_{j|p} - \boldsymbol{\pi}^{(1)}_{j|p}  ) }^2 + \ltwon{X^{(2)}(\hat{\boldsymbol{\pi}}^{(2)}_{j|p} - \boldsymbol{\pi}^{(2)}_{j|p} )}^2\\
		& \le  C_2^{(1)} \left\{ r_{j|p}^{(1)} \lor d \lor f^{(1)} + C_2^{(2)} r_{j|p}^{(2)} \lor d \lor f^{(2)} \right \} \\
		& \le C_2\, \left\{ d\lor \bigl(r_{j|p}^{(1)}\lor r_{j|p}^{(2)}\bigr) \lor \bigl(f^{(1)}\lor f^{(2)}\bigr) \right\} \le C_2 \left\{d \lor r_{\max}   \lor f_{\max} \right\},
	\end{aligned}
\end{equation*}
where $C_2 = C_2^{(1)} + C_2^{(2)}$ and $r_{\text{max}} = \underset{1\le i \le p}{\max} (r_{i}^{(1)} \lor r_{i}^{(2)})$. This concludes the proof of Theorem 2.

\section{Proof of Theorem 3} \label{sec:thm3}

Let $c_{\text{max}} = c_1^{(1)} \lor c_1^{(2)}$, and further denote
\begin{equation*}
	g_{n}= C_2\frac{ d\lor r_{\max} \lor f_{\max}}{n}+ 2c_{\text{max}} C_2 \lonen{\boldsymbol{\Pi}}
	\sqrt{\frac{ d\lor r_{\max} \lor f_{\max}}{n}}.
\end{equation*}

\begin{lemma} \label{lemma:restrictedEigenValue}
	Suppose that, for node $i$,
	\begin{equation} \label{ineqn:gn1}
		\sqrt{(d \lor r_{\max}\lor f_{\max}) \big/ n} + c_{\max}\lonen{\boldsymbol{\Pi}}\le \sqrt{c_{\max}^2\lonen{\boldsymbol{\Pi}}^2+\phi_0^2/(64C_2 |\mathcal{S}_{i}|)}.
	\end{equation}
	Under Assumptions 1-3, we have $\phivarmin{\mathbf{H}_i \mathbf{X}  \hat{\boldsymbol{\Pi}}_{-i}}{ \mathcal{S}_{i} } \ge \boldsymbol{\phi}_0/2$ with probability at least $1- e^{-f^{(1)}+\log p} - e^{-f^{(2)}+\log p}$.
\end{lemma}

\begin{proof}[\textbf{Proof of Lemma \ref{lemma:restrictedEigenValue}}]
	The inequality (\ref{ineqn:gn1}) implies that $g_n\le \phi_0^2/(64|\mathcal{S}_{i}|)$.
	
	For any index set $\mathcal{S}_{i} $ and vector $\delta$, note the definition of $\boldsymbol{\phi}_{\text{re}}(\cdot)$, then, we have that $\lonen{\delta}^2 \le (\lonen{\delta_{\mathcal{S}_i^c}} + \lonen{\delta_{\mathcal{S}_i}})^2 \le (3\sqrt{|\mathcal{S}_{i}|}\ltwon{\delta_{\mathcal{S}_i}} + \sqrt{|\mathcal{S}_{i}|}\ltwon{\delta_{\mathcal{S}_i}})^2 = 16 |\mathcal{S}_{i}| \ltwon{\delta_{\mathcal{S}_i}}^2$. we also have
	\begin{eqnarray} \label{inequality:Delta_decompose}
		\lefteqn{\frac{\delta^T((\mathbf{H}_i\mathbf{X}\hat{\boldsymbol{\Pi}}_{-i})^T  (\mathbf{H}_i\mathbf{X}\hat{\boldsymbol{\Pi}}_{-i}) -  (\mathbf{H}_i\mathbf{X}\boldsymbol{\Pi}_{-i})^T (\mathbf{H}_i\mathbf{X}\boldsymbol{\Pi}_{-i}) ) \delta}{n\ltwon{\delta_{\mathcal{S}_i} }^2}} \nonumber\\
		&\le& \frac{\lonen{\delta}^2}{n\ltwon{\delta_{\mathcal{S}_i }}^2  } \underset{j_{1},j_{2}}{\max}| (\mathbf{H}_{i}\mathbf{X}\hat{\boldsymbol{\Pi}}_{j_{1}})^T  (\mathbf{H}_{i}\mathbf{X}\hat{\boldsymbol{\Pi}}_{j_{2}}) -  (\mathbf{H}_{i}\mathbf{X} \boldsymbol{\Pi}_{j_{1}} )^T (\mathbf{H}_{i}\mathbf{X} \boldsymbol{\Pi}_{j_{2}} )| \nonumber\\
		&\le& \frac{16 |\mathcal{S}_i|}{n} \underset{j_{1},j_{2}}{\max}| (\mathbf{H}_{i}\mathbf{X}\hat{\boldsymbol{\Pi}}_{j_{1}})^T  (\mathbf{H}_{i}\mathbf{X}\hat{\boldsymbol{\Pi}}_{j_{2}}) -  (\mathbf{H}_{i}\mathbf{X} \boldsymbol{\Pi}_{j_{1}} )^T (\mathbf{H}_{i}\mathbf{X} \boldsymbol{\Pi}_{j_{2}})|.
	\end{eqnarray}
	Note that,
	\begin{eqnarray*}
		\lefteqn{(\mathbf{H}_{i}\mathbf{X}\hat{\boldsymbol{\Pi}}_{j_{1}})^T  (\mathbf{H}_{i}\mathbf{X}\hat{\boldsymbol{\Pi}}_{j_{2}}) -  (\mathbf{H}_{i}\mathbf{X} \boldsymbol{\Pi}_{j_{1}} )^T (\mathbf{H}_{i}\mathbf{X}  \boldsymbol{\Pi}_{j_{2}})}\\
		&=& \underbrace{(\hat{\boldsymbol{\Pi}}_{j_{1}} -  \boldsymbol{\Pi}_{j_{1}} )^T\mathbf{X}^T\mathbf{H}_{i}\mathbf{X} (\hat{\boldsymbol{\Pi}}_{j_{2}} -  \boldsymbol{\Pi}_{j_{2}})}_{T_{31}} + \underbrace{(\hat{\boldsymbol{\Pi}}_{j_{1}} -  \boldsymbol{\Pi}_{j_{1}} )^T\mathbf{X}^T\mathbf{H}_{i}\mathbf{X} \boldsymbol{\Pi}_{j_{2}}}_{T_{32}} + \underbrace{(\mathbf{X} \boldsymbol{\Pi}_{j_{1}} )^T\mathbf{H}_{i}\mathbf{X} (\hat{\boldsymbol{\Pi}}_{j_{2}} -  \boldsymbol{\Pi}_{j_{2}})}_{T_{33}}.
	\end{eqnarray*}
	We will derive the bounds for each of these three terms separately. With $\mathbf{H}_{i}$ a projection matrix, we have $\lambda_{max}(\mathbf{H}_{i}) = 1$. We can obtain that
	\begin{eqnarray*}
		|T_{31}| &\le& \ltwon{\mathbf{H}_{i}\mathbf{X} (\hat{\boldsymbol{\Pi}}_{j_{1}} - \boldsymbol{\Pi}_{j_{1}} )} \times \ltwon{\mathbf{H}_{i}\mathbf{X} (\hat{\boldsymbol{\Pi}}_{j_{2}} -  \boldsymbol{\Pi}_{j_{2}} )} \nonumber\\
		&\le& \lambda_{max}(\mathbf{H}_{i})  \ltwon{\mathbf{X} (\hat{\boldsymbol{\Pi}}_{j_{1}} - \boldsymbol{\Pi}_{j_{1}})} \times\ltwon{\mathbf{X} (\hat{\boldsymbol{\Pi}}_{j_{2}} - \boldsymbol{\Pi}_{j_{2}} )}  \nonumber\\
		&=& \ltwon{\mathbf{X} (\hat{\boldsymbol{\Pi}}_{j_{1}} - \boldsymbol{\Pi}_{j_{1}} )} \times \ltwon{\mathbf{X} (\hat{\boldsymbol{\Pi}}_{j_{2}} - \boldsymbol{\Pi}_{j_{2}} )}.
	\end{eqnarray*}
	Note that
	$|T_{32}| \le \ltwon{\mathbf{X}\  \boldsymbol{\Pi}_{j_{2}} } \ltwon{\mathbf{H}_{i}\mathbf{X} (\hat{\boldsymbol{\Pi}}_{j_{1}} -  \boldsymbol{\Pi}_{j_{1}})}$, and following Assumption 3, we have
	\begin{eqnarray*}
		\ltwon{\mathbf{X}\  \boldsymbol{\Pi}_{j_{2}} }^2 &=& \ltwon{X^{(1)} \boldsymbol{\pi}^{(1)}_{j|p} }^2 + \ltwon{X^{(2)}\boldsymbol{\pi}^{(2)}_{j|p} }^2\\
		&\le& (c_1^{(1)})^2 n^{(1)} \ltwon{\boldsymbol{\pi}^{(1)}_{j|p} }^2 + (c_1^{(2)})^2 n^{(2)} \ltwon{\boldsymbol{\pi}^{(2)}_{j|p} }^2 \\
		&\le& c_{\text{max}}^2 n (\ltwon{\boldsymbol{\pi}^{(1)}_{j|p}}^2 +\ltwon{\boldsymbol{\pi}^{(2)}_{j|p} }^2)\\
		&\le& c_{\text{max}}^2 n \left(\ltwon{\boldsymbol{\pi}^{(1)}_{j|p}} +\ltwon{\boldsymbol{\pi}^{(2)}_{j|p} } \right)^2\\
		&\le& c_{\text{max}}^2 n\lonen{\boldsymbol{\Pi}}^2.
	\end{eqnarray*}
	Therefore,
	\begin{equation} \label{Bound:T32}
		|T_{32}| \le \ltwon{\mathbf{X} \boldsymbol{\Pi}_{j_{2}} } \ltwon{\mathbf{H}_{i}\mathbf{X} (\hat{\boldsymbol{\Pi}}_{j_{1}} -  \boldsymbol{\Pi}_{j_{1}} )} \le c_{\text{max}}\sqrt{n} \lonen{\boldsymbol{\Pi}} \ltwon{\mathbf{X} (\hat{\boldsymbol{\Pi}}_{j_{1}} - \boldsymbol{\Pi}_{j_{1}} )}.
	\end{equation}
	Similarly, we can have
	\begin{equation}\label{Bound:Delta2}
		|T_{33}| \le c_{\text{max}}\sqrt{n} \lonen{\boldsymbol{\Pi}} \ltwon{\mathbf{X} (\hat{\boldsymbol{\Pi}}_{j_{2}} - \boldsymbol{\Pi}_{j_{2}})}.
	\end{equation}
	Theorem 2 leads to the following, with probability at least $1- e^{-f^{(1)}+\log(p)} - e^{-f^{(2)}+\log(p)}$,
	\begin{equation} \label{inequality:Delta123}
		\begin{cases}
			\begin{aligned}
				\frac{|T_{31}|}{n} &\le C_2 \frac{ d\lor r_{\max} \lor f_{\max}}{n},\\
				\frac{|T_{32}|}{n} &\le c_{\text{max}}C_2 \lonen{\boldsymbol{\Pi}} \sqrt{\frac{ d\lor r_{\max} \lor f_{\max}}{n}},\\
				\frac{|T_{33}|}{n} &\le c_{\text{max}}C_2 \lonen{\boldsymbol{\Pi}}\sqrt{\frac{ d\lor r_{\max} \lor f_{\max}}{n}}.
			\end{aligned}
		\end{cases}
	\end{equation}
	Putting the above three inequalities together, we have,
	\begin{eqnarray} \label{inequality:Delta_all}
		\lefteqn{\frac{\delta^T( (\mathbf{H}_i\mathbf{X}\hat{\boldsymbol{\Pi}}_{-i})^T  (\mathbf{H}_i\mathbf{X}\hat{\boldsymbol{\Pi}}_{-i}) -  (\mathbf{H}_i\mathbf{X}\boldsymbol{\Pi}_{-i})^T (\mathbf{H}_i\mathbf{X}\boldsymbol{\Pi}_{-i}) ) \delta}{n\ltwon{\delta_{\mathcal{S}_i} }^2 } } \nonumber\\
		&\le&  16 |\mathcal{S}_i| \times \frac{|T_{31}| + |T_{32}| + |T_{33}|}{n} = 16 |\mathcal{S}_i| g_{n}
		\le 16 |\mathcal{S}_i| \frac{\boldsymbol{\phi}_0^2}{64 |\mathcal{S}_i|} = \frac{\boldsymbol{\phi}_0^2}{4}.
	\end{eqnarray}
	Together with Assumption~4, we have $\phivarmin{\mathbf{H}_{i}\mathbf{X}\hat{\boldsymbol{\Pi}}_{-k}}{ \mathcal{S}_{k} } \ge \boldsymbol{\phi}_0/2$. This concludes the proof of Lemma~\ref{lemma:restrictedEigenValue}.
\end{proof}

\begin{lemma}(Basic Inequality) \label{lemma:BI}
	Let $\boldsymbol{\eta}_i =  2n^{-1} \hat{\mathbf{Z}}_{-i}^T \mathbf{H}_i \boldsymbol{\epsilon}_{i} - 2 n^{-1} \hat{\mathbf{Z}}_{-i}^T \mathbf{H}_i (\hat{\mathbf{Z}}_{-i} - \mathbf{Z}_{-i})  \boldsymbol{\beta}_{i } $ and
	\begin{equation*}
		\mathscr{E}(\lambda_{i}) = \left\{ \infn{W_i^{-1} \boldsymbol{\eta}_i} \le \lambda_{i}/2 \right\},
	\end{equation*}
	for $\lambda_{i}$ specified in Theorem 3. Under Assumptions 1-2, with $h_n$ defined in Theorem 3, there exit a positive constant $C_3 > 0$ such that
	\begin{equation*}
		\mathbb{P}(\mathscr{E}(\lambda_{i})) \ge 1-e^{-C_3 h_n +\log(4q)} - e^{-f^{(1)}+\log(p)} - e^{-f^{(2)} + \log(p)}.
	\end{equation*}
	Concurring with event $\mathscr{E}(\lambda_{i})$, we have the following basic inequality,
	\begin{equation} \label{basicinequality}
		\begin{aligned}
			&n^{-1}\ltwon{\mathbf{H}_i\hat{\mathbf{Z}}_{-i} (\hat{\boldsymbol{\beta}}_i - \boldsymbol{\beta}_{i} ) }^2 + \lambda_{i} \bweight_i^T |\hat{\boldsymbol{\beta}}_i|_1 \le \lambda_{i} \bweight_i^T |\boldsymbol{\beta}_ {i} |_1 +\boldsymbol{\eta}_{i}^T (\hat{\boldsymbol{\beta}}_i - \boldsymbol{\beta}_{i}).
		\end{aligned}
	\end{equation}
\end{lemma}

\begin{proof}[\textbf{Proof of Lemma \ref{lemma:BI}}]
	Letting
	\begin{equation} \label{notation:barxi}
		\boldsymbol{\xi}_{-i} =
		\begin{pmatrix}
			\boldsymbol{\xi}_{-i}^{(1)} & \boldsymbol{\xi}_{-i}^{(1)}\\
			\boldsymbol{\xi}_{-i}^{(2)} & -\boldsymbol{\xi}_{-i}^{(2)}
		\end{pmatrix},
	\end{equation}
	we have $\mathbf{Z}_{-i} = \mathbf{X} \boldsymbol{\Pi}_{-i} +\boldsymbol{\xi}_{-i}$. With $\hat{\mathbf{Z}}_{-i} = \mathbf{X} \hat{\boldsymbol{\Pi}}_{-i}$, we get
	\begin{eqnarray*}
		\boldsymbol{\eta}_i&=& \frac{2}{n} \hat{\boldsymbol{\Pi}}_{-i}^T \mathbf{X}^T \mathbf{H}_i\boldsymbol{\epsilon}_{i} - \frac{2}{n} \hat{\boldsymbol{\Pi}}_{-i}^T \mathbf{X}^T \mathbf{H}_i (\mathbf{X}\hat{\boldsymbol{\Pi}}_{-i} - \mathbf{X} \boldsymbol{\Pi}_{-i} -\boldsymbol{\xi}_{-i}) \boldsymbol{\beta}_{i} \nonumber\\
		&=& \underbrace{\frac{2}{n} (\hat{\boldsymbol{\Pi}}_{-i} - \boldsymbol{\Pi}_{-i})^T \mathbf{X}^T\mathbf{H}_i\boldsymbol{\epsilon}_{i}}_{T_{34}} + \underbrace{\frac{2}{n} \boldsymbol{\Pi}_{-i}^T\mathbf{X}^T\mathbf{H}_i\boldsymbol{\epsilon}_{i}}_{T_{35}} \nonumber\\
		&&+ \underbrace{\frac{2}{n} (\hat{\boldsymbol{\Pi}}_{-i} - \boldsymbol{\Pi}_{-i})^T\mathbf{X}^T \mathbf{H}_i \boldsymbol{\xi}_{-i}  \boldsymbol{\beta}_{i } }_{T_{36}}+ \underbrace{\frac{2}{n} \boldsymbol{\Pi}_{-i}^T \mathbf{X}^T \mathbf{H}_i \boldsymbol{\xi}_{-i}  \boldsymbol{\beta}_{i }}_{T_{37}} \nonumber\\
		&&- \underbrace{\frac{2}{n}(\hat{\boldsymbol{\Pi}}_{-i} - \boldsymbol{\Pi}_{-i})^T\mathbf{X}^T\mathbf{H}_i\mathbf{X}(\hat{\boldsymbol{\Pi}}_{-i} -\boldsymbol{\Pi}_{-i}) \boldsymbol{\beta}_{i }}_{T_{38}} - \underbrace{\frac{2}{n} \boldsymbol{\Pi}_{-i}^T\mathbf{X}^T\mathbf{H}_i\mathbf{X}(\hat{\boldsymbol{\Pi}}_{-i}-\boldsymbol{\Pi}_{-i}) \boldsymbol{\beta}_{i } }_{T_{39}}.
	\end{eqnarray*}
	We aim to bound each of these six terms by $\lambda_{i}/12$ either probabilistically or deterministically.
	
	Firstly, for some constant $t_{\lambda}> 0 $, we choose the adaptive lasso tuning parameter as below,
	\begin{equation} \label{eqn:lambdai}
		\lambda_{i} = t_{\lambda}\|\boldsymbol{\omega}_i\|_{-\infty}^{-1}\lonen{\mathbf{B}}\lonen{\boldsymbol{\Pi}}\sqrt{\frac{(d \lor r_{\max}\lor f_{\max})\log(p)}{n_{\min}}}.
	\end{equation}
	Denoting the $j$-th column of $\mathbf{X}$ by $X_{\cdot j}$, we have $X_{\cdot j}^T X_{\cdot j} = n^{(k)}$ for $k\in\{1,2\}$ due to standardization. Furthermore,
	\[
	\var\left(\frac{1}{n} X_{\cdot j}^T \mathbf{H}_i \boldsymbol{\epsilon}_{i} \right) \le \frac{1}{n^2} X_{\cdot j}^T \mathbf{H}_i X_{\cdot j}  \sigma_{p\max}^{2} \le \frac{n^{(k)}}{n^2} \sigma_{p\max}^2 \le \frac{1}{n} \sigma_{p\max}^2.
	\]
	
	For $T_{34}$, via the classical Gaussian tail inequality, we have
	\begin{eqnarray}\label{bound:tb1}
		\mathbb{P}\left(\infn{W_{i}^{-1} T_{34} }\ge \frac{\lambda_{i}}{12}\right) &\le& \mathbb{P}\left(  \infn{\frac{2}{n}(\hat{\boldsymbol{\Pi}}_{-i}  - \boldsymbol{\Pi}_{-i})^T  \mathbf{X}^T\mathbf{H}_i\boldsymbol{\epsilon}_{i}  }\ge \frac{\lambda_{i}\|\boldsymbol{\omega}_i\|_{-\infty}}{12}  \right) \nonumber\\
		&\le& \mathbb{P}\left( \infn{(\hat{\boldsymbol{\Pi}}_{-i}  - \boldsymbol{\Pi}_{-i})^T }  \infn{\frac{2}{n} \mathbf{X}^T\mathbf{H}_i\boldsymbol{\epsilon}_{i}} \ge  \frac{\lambda_{i}\|\boldsymbol{\omega}_i\|_{-\infty}}{12}  \right) \nonumber\\
		&\le& \mathbb{P}\left(\infn{\frac{2}{n} \mathbf{X}^T\mathbf{H}_i\boldsymbol{\epsilon}_{i}} \ge \frac{\lambda_{i}\|\boldsymbol{\omega}_i\|_{-\infty}}{12 \delta_{\Pi}} \right) \le 2q \exp \left\{- \frac{n \lambda_{i}^2\|\boldsymbol{\omega}_i\|_{-\infty}^2}{1152\sigma^2_{p\max} \delta_{\Pi}^2} \right\} \nonumber\\
		&\le& 2q\cdot p^{- \frac{n}{d} t_1 \lonen{\mathbf{B}}^2\lonen{\boldsymbol{\Pi}}^2 } \le 2q\cdot p \cdot p^{- t_1 \lonen{\mathbf{B}}^2 \frac{n}{d}\lonen{\boldsymbol{\Pi}}^2},
	\end{eqnarray}
	where $t_1 = t^2_{\lambda}/(2304C_1 \sigma^2_{p\max})$, and $\delta_{\Pi}$ is the maximum estimation loss of the first stage. The last inequality is obtained based on the following bound of $\delta_{\Pi}$. Following Theorem 2, $\delta_{\Pi}$ satisfies the following inequality with probability at least $1 - e^{-f^{(1)}+\log(p)} - e^{-f^{(2)} + \log(p)}$,
	\begin{equation} \label{notion:estimationMax}
		\begin{aligned}
			\delta_{\Pi}^2 = \underset{1\le j \le 2p}{\text{max}} \lonen{\hat{\boldsymbol{\Pi}}_{j} - \boldsymbol{\Pi}_{j}}^2 \le  \underset{1\le j \le 2p}{\max} \left(2d\ltwon{ \hat{\boldsymbol{\Pi}}_{j} - \boldsymbol{\Pi}_{j} }^2 \right)
			\le 2C_1 d \left \{\frac{d\lor r_{\max} \lor f_{\max}}{n_{\min}}\right\}.
		\end{aligned}
	\end{equation}
	Note that the first inequality of (\ref{notion:estimationMax}) holds, since $\hat{\boldsymbol{\Pi}}$ and $\boldsymbol{\Pi}$ have at most $2d$ non-zeros based on our assumptions and the screening in the calibration step.
	
	Similarly, for the second term $T_{35}$, we have that, with $t_2 = \frac{(t_{\lambda})^2}{1152\sigma_{p\max}^2}$,
	\begin{eqnarray} \label{bound:tb2}
		\mathbb{P}\left(\infn{W_{i}^{-1} T_{35} }\ge \frac{\lambda_{i}}{12}\right) &\le& \mathbb{P}\left(\infn{\frac{2}{n}\boldsymbol{\Pi}_{-i}^T  \mathbf{X}^T\mathbf{H}_i\boldsymbol{\epsilon}_{i} } \ge \frac{\lambda_{i} \|\boldsymbol{\omega}_i\|_{-\infty} }{12} \right) \nonumber\\
		&\le& \mathbb{P}\left( \infn{\boldsymbol{\Pi}_{-i}^T} \infn{\frac{2}{n} \mathbf{X}^T\mathbf{H}_i\boldsymbol{\epsilon}_{i}} \ge  \frac{\lambda_{i}\|\boldsymbol{\omega}_i\|_{-\infty}}{12} \right) \nonumber\\
		&\le& \mathbb{P}\left( \infn{\frac{2}{n} \mathbf{X}^T\mathbf{H}_i\boldsymbol{\epsilon}_{i}} \ge \frac{\lambda_{i}\|\boldsymbol{\omega}_i\|_{-\infty}}{12 \infn{\boldsymbol{\Pi}_{-i}^T }}  \right) \le 2q \exp \left\{ - \frac{n\lambda_{i}^2\|\boldsymbol{\omega}_i\|_{-\infty}^{2}}{1152\sigma^2_{p\max}\infn{\boldsymbol{\Pi}_{-i}^T}^2}  \right\} \nonumber\\
		&=& 2q \cdot p^{- t_2 \lonen{\mathbf{B}}^2 (d \lor r_{\max}\lor f_{\max}) n/n_{\min}} \le 2q \cdot p \cdot p^{- t_2 \lonen{\mathbf{B}}^2 (d \lor r_{\max}\lor f_{\max}) n/n_{\min}}.
	\end{eqnarray}
	For the third term $T_{36}$, we write
	\begin{eqnarray}\label{bound:tb3}
		\mathbb{P}\left(\infn{W_{i}^{-1} T_{36} }\ge \frac{\lambda_{i}}{12} \right) &\le& \mathbb{P}\left( \infn{(\hat{\boldsymbol{\Pi}}_{-i} -\boldsymbol{\Pi}_{-i})^T}  \lonen{\frac{2}{n} \mathbf{X}^T\mathbf{H}_i \boldsymbol{\xi}_{-i}\boldsymbol{\beta}_{i }}  \ge \frac{\lambda_{i}\|\boldsymbol{\omega}_i\|_{-\infty}}{12} \right) \nonumber\\
		&\le& \mathbb{P}\left(\delta_{\Pi} \times \underset{j_{1},j_{2}}{\text{max}} |\frac{2}{n} X_{\cdot j_{1}}^T\mathbf{H}_i\boldsymbol{\xi}_{j_{2}}| \times\lonen{\boldsymbol{\beta}_i} \ge \frac{\lambda_{i}\|\boldsymbol{\omega}_i\|_{-\infty}}{12} \right) \nonumber\\
		&\le& \mathbb{P}\left(\underset{j_{1},j_{2}}{\max} |\frac{2}{n} X_{\cdot j_{1}}^T\mathbf{H}_i\boldsymbol{\xi}_{j_{2}} | \ge  \frac{\lambda_{i}\|\boldsymbol{\omega}_i\|_{-\infty}}{12\delta_{\Pi}\lonen{\boldsymbol{\beta}_i}} \right) \le 2q \cdot 2p \exp \left\{-\frac{n\lambda_{i}^2\|\boldsymbol{\omega}_i\|_{-\infty}^{2}}{1152 \tilde{\sigma}_{p\max}^2 \delta_{\Pi}^2\lonen{\boldsymbol{\beta}_i}^2}\right\} \nonumber\\
		&=& 4q\cdot p \cdot p^{-t_3 \lonen{\boldsymbol{\Pi}}^2 n/d},
	\end{eqnarray}
	where $\tilde{\sigma}_{p\max}^2 = \underset{i}{\max}(\tilde{\sigma}_i^{(1)} \lor \tilde{\sigma}_i^{(2)})$, $\var(\frac{1}{n} X_{\cdot j_{1}}^T\mathbf{H}_i\boldsymbol{\xi}_{j_{2}} ) \le \tilde{\sigma}_{p\max}^2/n$ and  $t_3 = \frac{t_{\lambda}^2}{2304 C_1\tilde{\sigma}_{p\max}^2}$. Similarly, with $t_4 = \frac{t_{\lambda}^2}{1152\tilde{\sigma}_{p\max}^2}$, we write $T_{37}$ term as
	\begin{eqnarray}\label{bound:tb4}
		\mathbb{P}\left( \infn{ W_{i}^{-1} T_{37}} \ge \frac{\lambda_{i}}{12}  \right) &\le& 2q \cdot 2p \cdot \exp \left\{-\frac{n\lambda_{i}^2 \|\boldsymbol{\omega}_i\|_{-\infty}^{2} }{1152\tilde{\sigma}_{p\max}^2\infn{\boldsymbol{\Pi}_{-i}^T}^2\lonen{\boldsymbol{\beta}_i}^2}  \right\} \nonumber\\
		&=& 4q\cdot p \cdot p ^{- t_4 (d \lor r_{\max}\lor f_{\max}) n/n_{\min}}.
	\end{eqnarray}
	For the deterministic term $T_{38}$, choosing $t_{\lambda} \ge 12C_2\lonen{\mathbf{\Pi}}^{-1}\sqrt{(d\lor r_{\max} \lor f_{\max})/(n\log(p))}$, along with \textit{Cauchy-Schwarz Inequality}, we have
	\begin{eqnarray*}
		\infn{W_{i}^{-1} T_{38}} &\le& \frac{\lonen{\boldsymbol{\beta}_i} \|\boldsymbol{\omega}_i\|_{-\infty}^{-1} }{n} \underset{j_{1},j_{2}}{\max} | (\hat{\boldsymbol{\Pi}}_{j_{1}} - \boldsymbol{\Pi}_{j_{1}})^T\mathbf{X}^T\mathbf{H}_i\mathbf{X}(\hat{\boldsymbol{\Pi}}_{j_{2}} -\boldsymbol{\Pi}_{j_{2}})|\nonumber\\
		&\le& \frac{\lonen{\boldsymbol{\beta}_i} \|\boldsymbol{\omega}_i\|_{-\infty}^{-1}}{n}  \underset{j_{1},j_{2} }{\max} \left\{\ltwon{\mathbf{H}_i\mathbf{X}(\hat{\boldsymbol{\Pi}}_{j_{1}} -\boldsymbol{\Pi}_{j_{1}})}\ltwon{\mathbf{H}_i\mathbf{X}(\hat{\boldsymbol{\Pi}}_{j_{2}} -\boldsymbol{\Pi}_{j_{2}})} \right\}\nonumber\\
		&\le&  \frac{\lonen{\boldsymbol{\beta}_i}\|\boldsymbol{\omega}_i\|_{-\infty}^{-1}}{n}  \underset{j_{1},j_{2}}{\max}\left\{ \lambda_{\max}(\mathbf{H}_i) \ltwon{\mathbf{X}(\hat{\boldsymbol{\Pi}}_{j_{1}} -\boldsymbol{\Pi}_{i_{1}})}\ltwon{\mathbf{X}(\hat{\boldsymbol{\Pi}}_{j_{2}} -\boldsymbol{\Pi}_{j_{2}})}\right\} \nonumber\\
		&\le& \frac{\lonen{\boldsymbol{\beta}_i}\|\boldsymbol{\omega}_i\|_{-\infty}^{-1}}{n}  \underset{j_{1},j_{2}}{\max} \left\{\ltwon{\mathbf{X}(\hat{\boldsymbol{\Pi}}_{j_{1}} -\boldsymbol{\Pi}_{j_{1}})}\ltwon{\mathbf{X}(\hat{\boldsymbol{\Pi}}_{j_{2}} -\boldsymbol{\Pi}_{j_{2}})}\right\} \nonumber\\
		&\le& \lonen{\boldsymbol{\beta}_i} \|\boldsymbol{\omega}_i\|_{-\infty}^{-1} C_2 \frac{d\lor r_{\max} \lor f_{\max} }{n} \le \frac{\lambda_{i}}{12} \times \left( \frac{12C_2}{ t_{\lambda}\lonen{\boldsymbol{\Pi}}} \sqrt{\frac{ d \lor r_{\max} \lor f_{\max}}{n \log(p)}} \right)  \le \frac{\lambda_{i}}{12}.
	\end{eqnarray*}
	Similarly, we choose $t_{\lambda} \ge 24 \sqrt{C_2 n_{\min}/(n\log(p))}$, and take Theorem 2 to obtain
	\begin{eqnarray*}
		\infn{W_{i}^{-1} T_{39}} &\le& 2\frac{\lonen{\boldsymbol{\beta}_i}\infn{\boldsymbol{\Pi}_{-i}^T}\|\boldsymbol{\omega}_i\|_{-\infty}^{-1}}{n} \underset{j_{1},j_{2}}{\text{max}} | X_{\cdot j_{1}}^T\mathbf{H}_i\mathbf{X}(\hat{\boldsymbol{\Pi}}_{j_{2}} -\boldsymbol{\Pi}_{j_{2}})|  \nonumber\\
		&\le& 2\frac{\lonen{\boldsymbol{\beta}_i}\infn{\boldsymbol{\Pi}_{-i}^T}\|\boldsymbol{\omega}_i\|_{-\infty}^{-1}}{\sqrt{n}} \underset{j_{2}}{\text{max}}  \ltwon{\mathbf{H}_i\mathbf{X}(\hat{\boldsymbol{\Pi}}_{j_{2}} -\boldsymbol{\Pi}_{j_{2}})} \nonumber\\
		&\le& 2\frac{\lonen{\boldsymbol{\beta}_i}\infn{\boldsymbol{\Pi}_{-i}^T}\|\boldsymbol{\omega}_i\|_{-\infty}^{-1}}{\sqrt{n}} \underset{j_{2}}{\text{max}}  \ltwon{\mathbf{X}(\hat{\boldsymbol{\Pi}}_{j_{2}} -\boldsymbol{\Pi}_{j_{2}})}  \le \frac{\lambda_{i}}{12} \times \left(  \frac{24}{ t_{\lambda}}\sqrt{\frac{C_2 n_{\min}}{n\log(p)}} \right) \le \frac{\lambda_{i}}{12}.
	\end{eqnarray*}
	Note that $n\ge n_{\min}$. Putting together the probabilistic bounds (\ref{bound:tb1}), (\ref{notion:estimationMax}), (\ref{bound:tb2}), (\ref{bound:tb3}) and (\ref{bound:tb4}), along with union bound, there exist a constant $C_3 > 0$ such that
	\begin{equation*}
		\mathbb{P}(\mathscr{E}(\lambda_{i})) \ge 1-3e^{-C_3 h_n +\log(4pq)} - e^{-f^{(1)}+\log(p)} - e^{-f^{(2)} + \log(p)}.
	\end{equation*}
	
	Next we will establish the basic inequality, concurring with the event $\mathscr{E}(\lambda_{i})$.
	
	Since the estimator $\hat{\boldsymbol{\beta}}_i$ from the adaptive lasso minimizes the corresponding objective function, we have
	\begin{equation} \label{inequality:optimal}
		\frac{1}{n}\ltwon{  \mathbf{H}_i\mathcal{\mathbf{Y}}_i -\mathbf{H}_i \hat{\mathbf{Z}}_{-i} \hat{\boldsymbol{\beta}}_i } + \lambda_{i} \bweight_i^T |\hat{\boldsymbol{\beta}}_i|_1 \\
		\le  \frac{1}{n}\ltwon{  \mathbf{H}_i\mathbf{Y}_i -\mathbf{H}_i \hat{\mathbf{Z}}_{-i}  \boldsymbol{\beta}_{i } }   +  \lambda_{i} \bweight_i^T  |\boldsymbol{\beta}_{i }|_1.
	\end{equation}
	Because $\mathbf{H}_i\mathbf{Y}_i = \mathbf{H}_i \mathbf{Z}_{-i}  \boldsymbol{\beta}_{i }  +\mathbf{H}_i\boldsymbol{\epsilon}_{i}$, we can rewrite
	\begin{eqnarray} \label{ltwoleft}
		\lefteqn{\ltwon{  \mathbf{H}_i\mathbf{Y}_i -\mathbf{H}_i \hat{\mathbf{Z}}_{-i} \hat{\boldsymbol{\beta}}_i }^2}  \nonumber\\
		&=& \ltwon{  \mathbf{H}_i \mathbf{Z}_{-i}  \boldsymbol{\beta}_{i }  +\mathbf{H}_i\boldsymbol{\epsilon}_{i} -\mathbf{H}_i \hat{\mathbf{Z}}_{-i} \hat{\boldsymbol{\beta}}_i }^2  \nonumber\\
		&=&  \ltwon{\mathbf{H}_i\boldsymbol{\epsilon}_{i}}^2 - 2 \boldsymbol{\epsilon}_{i}^T\mathbf{H}_i(  \hat{\mathbf{Z}}_{-i} \hat{\boldsymbol{\beta}}_i -\mathbf{Z}_{-i}  \boldsymbol{\beta}_{i }  )+ \ltwon{\mathbf{H}_i \hat{\mathbf{Z}}_{-i} \hat{\boldsymbol{\beta}}_i -\mathbf{H}_i\hat{\mathbf{Z}}_{-i} \boldsymbol{\beta}_{i} + \mathbf{H}_i\hat{\mathbf{Z}}_{-i} \boldsymbol{\beta}_{i }  -\mathbf{H}_i \mathbf{Z}_{-i}  \boldsymbol{\beta}_{i } }^2  \nonumber\\
		&=&  \ltwon{\mathbf{H}_i\boldsymbol{\epsilon}_{i}}^2 - 2 \boldsymbol{\epsilon}_{i}^T\mathbf{H}_i(  \hat{\mathbf{Z}}_{-i} \hat{\boldsymbol{\beta}}_i -\mathbf{Z}_{-i}  \boldsymbol{\beta}_{i }  )+ \ltwon{\mathbf{H}_i\hat{\mathbf{Z}}_{-i} (\hat{\boldsymbol{\beta}}_i -  \boldsymbol{\beta}_{i } )}^2 +\ltwon{\mathbf{H}_i(\hat{\mathbf{Z}}_{-i} -\mathbf{Z}_{-i}) \boldsymbol{\beta}_{i } }^2  \nonumber\\
		&&+  2 \boldsymbol{\beta}_{i } ^T (\hat{\mathbf{Z}}_{-i} -\mathbf{Z}_{-i})^T\mathbf{H}_i\hat{\mathbf{Z}}_{-i}(\hat{\boldsymbol{\beta}}_i - \boldsymbol{\beta}_{i }).
	\end{eqnarray}
	Similarly we can rewrite
	\begin{eqnarray} \label{ltworight}
		\ltwon{\mathbf{H}_i\mathbf{Y}_i - \mathbf{H}_i\hat{\mathbf{Z}}_{-i}  \boldsymbol{\beta}_{i } }^2 &=& \ltwon{\mathbf{H}_i\mathbf{Z}_{-i} \boldsymbol{\beta}_{i }  +\mathbf{H}_i\boldsymbol{\epsilon}_{i} -\mathbf{H}_i\hat{\mathbf{Z}}_{-i}  \boldsymbol{\beta}_{i } }^2  \nonumber\\
		&=& \ltwon{\mathbf{H}_i\boldsymbol{\epsilon}_{i}}^2 + \ltwon{\mathbf{H}_i(\hat{\mathbf{Z}}_{-i}-\mathbf{Z}_{-i}) \boldsymbol{\beta}_{i } }^2 - 2\boldsymbol{\epsilon}_{i}^T\mathbf{H}_i(\hat{\mathbf{Z}}_{-i} -\mathbf{Z}_{-i}) \boldsymbol{\beta}_{i}.
	\end{eqnarray}
	Plugging equations (\ref{ltwoleft}) and (\ref{ltworight}) into (\ref{inequality:optimal}), we then have
	\begin{eqnarray*}
		\lefteqn{\frac{1}{n} \ltwon{  \mathbf{H}_i\hat{\mathbf{Z}}_{-i} (\hat{\boldsymbol{\beta}}_i - \boldsymbol{\beta}_{i } ) }^2 +\lambda_{i} \bweight_i^T |\hat{\boldsymbol{\beta}}_i|_1}  \nonumber\\
		&\le& \lambda_{i} \bweight_i^T |\boldsymbol{\beta}_i|_1 +\left(\frac{2}{n} \hat{\mathbf{Z}}_{-i}^T \mathbf{H}_i\boldsymbol{\epsilon}_{i} - \frac{2}{n} \hat{\mathbf{Z}}_{-i}^T \mathbf{H}_i (\hat{\mathbf{Z}}_{-i} - \mathbf{Z}_{-i})  \boldsymbol{\beta}_{i } \right)^T(\hat{\boldsymbol{\beta}}_i - \boldsymbol{\beta}_{i } )  \nonumber\\
		&=& \lambda_{i} \bweight_i^T |\boldsymbol{\beta}_i|_1 + \boldsymbol{\eta}_{i} ^T(\hat{\boldsymbol{\beta}}_i - \boldsymbol{\beta}_{i } ).
	\end{eqnarray*}
	Thus, the basic inequality is established. This concludes the proof of Lemma~\ref{lemma:BI}.
\end{proof}

Conditioning on the event $\mathscr{E}(\lambda_{i})$, we remove the random term $\boldsymbol{\eta}_{i}$ from the basic inequality as
\begin{eqnarray}\label{inequality:remove_E}
	\lefteqn{\frac{1}{n} \ltwon{\mathbf{H}_i\hat{\mathbf{Z}}_{-i} (\hat{\boldsymbol{\beta}}_i - \boldsymbol{\beta}_{i})}^2} \nonumber\\
	&\le& \lambda_{i} \bweight_i^T |\boldsymbol{\beta}_ {i}|_1 - \lambda_{i} \bweight_i^T |\hat{\boldsymbol{\beta}}_i|_1  + \boldsymbol{\eta}_{i}^T (\hat{\boldsymbol{\beta}}_i - \boldsymbol{\beta}_{i}) \nonumber \\
	&\le& \lambda_{i} \boldsymbol{\omega}^T_{\mathcal{S}_{i}} |\boldsymbol{\beta}_{\mathcal{S}_{i}}|_1 -
	\lambda_{i} \boldsymbol{\omega}^T_{\mathcal{S}_{i}} |\hat{\boldsymbol{\beta}}_{\mathcal{S}_{i}}|_1 -\lambda_{i} \boldsymbol{\omega}^T_{\mathcal{S}_{i}^{c}} |\hat{\boldsymbol{\beta}}_{\mathcal{S}_{i}^{c}}|_1 + \boldsymbol{\eta}_{\mathcal{S}_{i}^{c}}^T(\hat{\boldsymbol{\beta}}_{\mathcal{S}_{i}^{c}}) +  \boldsymbol{\eta}_{\mathcal{S}_{i}}^T(\hat{\boldsymbol{\beta}}_{\mathcal{S}_{i}} -\boldsymbol{\beta}_{\mathcal{S}_{i}}) \nonumber \\
	&\le& \lambda_{i} \boldsymbol{\omega}^T_{\mathcal{S}_{i}} |\hat{\boldsymbol{\beta}}_{\mathcal{S}_{i}} -
	\boldsymbol{\beta}_{\mathcal{S}_{i}}|_1 -\lambda_{i} \boldsymbol{\omega}^T_{\mathcal{S}_{i}^{c}} |\hat{\boldsymbol{\beta}}_{\mathcal{S}_{i}^{c}}|_1 + \frac{\lambda_{i}}{2}\boldsymbol{\omega}^T_{\mathcal{S}_{i}^{c}}|\hat{\boldsymbol{\beta}}_{\mathcal{S}_{i}^{c}}|_1 +  \frac{\lambda_{i} }{2}\boldsymbol{\omega}^T_{\mathcal{S}_{i}} |\hat{\boldsymbol{\beta}}_{\mathcal{S}_{i}} -\boldsymbol{\beta}_{\mathcal{S}_{i}}|_1  \nonumber \\
	&\le& \frac{3}{2}\lambda_{i} \boldsymbol{\omega}^T_{\mathcal{S}_{i}} |\hat{\boldsymbol{\beta}}_{\mathcal{S}_{i}} -
	\boldsymbol{\beta}_{\mathcal{S}_{i}}|_1- \frac{1}{2}\lambda_{i} \boldsymbol{\omega}^T_{\mathcal{S}_{i}^{c}} |\hat{\boldsymbol{\beta}}_{\mathcal{S}_{i}^{c}}|_1 \nonumber \\
	&\le & \frac{3}{2}\lambda_{i} \|\boldsymbol{\omega}_{\mathcal{S}_{i}}\|_{\infty} \lonen{\hat{\boldsymbol{\beta}}_{\mathcal{S}_{i}} - \boldsymbol{\beta}_{\mathcal{S}_{i}}} - \frac{1}{2}\lambda_{i}  \|\boldsymbol{\omega}_{\mathcal{S}_{i}^{c}}\|_{-\infty}
	\lonen{\hat{\boldsymbol{\beta}}_{\mathcal{S}_{i}^{c}}}.
\end{eqnarray}
The fact that $\ltwon{\mathbf{H}_i\hat{\mathbf{Z}}_{-i} (\hat{\boldsymbol{\beta}}_i - \boldsymbol{\beta}_{i})}^2$ is always positive leads to
\begin{equation}\label{inequality:pure_beta}
	\|\boldsymbol{\omega}_{\mathcal{S}_{i}^{c}}\|_{-\infty}
	\lonen{\hat{\boldsymbol{\beta}}_{\mathcal{S}_{i}^{c}}} \le 3\|\boldsymbol{\omega}_{\mathcal{S}_{i}}\|_{\infty} \lonen{\hat{\boldsymbol{\beta}}_{\mathcal{S}_{i}} -
		\boldsymbol{\beta}_{\mathcal{S}_{i}}},
\end{equation}
which, following Assumption 4, further implies that
\begin{equation}
	\lonen{\hat{\boldsymbol{\beta}}_{\mathcal{S}_{i}^{c}} -\boldsymbol{\beta}_{\mathcal{S}_{i}^c}} \le 3 \lonen{\hat{\boldsymbol{\beta}}_{\mathcal{S}_{i}} -
		\boldsymbol{\beta}_{\mathcal{S}_{i}}}.
\end{equation}

The above inequality, as well as the last inequality in (\ref{inequality:remove_E}), implies that
\begin{eqnarray}\label{bound:prediction_loss_derive}
	\lefteqn{\frac{1}{n} \ltwon{\mathbf{H}_i\hat{\mathbf{Z}}_{-i} (\hat{\boldsymbol{\beta}}_i - \boldsymbol{\beta}_{i})}^2}  \nonumber\\
	&\le & \frac{3}{2}\lambda_{i} \|\boldsymbol{\omega}_{\mathcal{S}_{i}}\|_{\infty} \lonen{\hat{\boldsymbol{\beta}}_{\mathcal{S}_{i}} -
		\boldsymbol{\beta}_{\mathcal{S}_{i}}}\le  \frac{3}{2}\lambda_{i} \|\boldsymbol{\omega}_{\mathcal{S}_{i}}\|_{\infty}\sqrt{|\mathcal{S}_{i}|} \ltwon{\hat{\boldsymbol{\beta}}_{\mathcal{S}_{i}} -
		\boldsymbol{\beta}_{\mathcal{S}_{i}}} \nonumber\\
	&\le&  \frac{3}{2}\lambda_{i} \|\boldsymbol{\omega}_{\mathcal{S}_{i}}\|_{\infty} \sqrt{|\mathcal{S}_{i}|} \frac{2 \ltwon{\mathbf{H}_i\hat{\mathbf{Z}}_{-i} (\hat{\boldsymbol{\beta}}_i - \boldsymbol{\beta}_{i})}}{\sqrt{n}\boldsymbol{\phi}_0},
\end{eqnarray}
where the last inequality follows Assumption~4 and Lemma~\ref{lemma:restrictedEigenValue}. The above inequality leads to that,
\begin{eqnarray*}
	\frac{1}{n} \ltwon{\mathbf{H}_i\hat{\mathbf{Z}}_{-i} (\hat{\boldsymbol{\beta}}_i - \boldsymbol{\beta}_{i})}^2 \le  \frac{9 ( \|\boldsymbol{\omega}_{\mathcal{S}_{i}}\|_{\infty} )^2}{\boldsymbol{\phi}_0^2}|\mathcal{S}_{i}| \lambda_{i}^2.
\end{eqnarray*}
Plugging in (\ref{eqn:lambdai}), and letting $C_4 = 3t_{\lambda}$, we obtain that
\begin{equation} \label{eqn:hzbetadiff}
	\frac{1}{n} \ltwon{\mathbf{H}_i\hat{\mathbf{Z}}_{-i} (\hat{\boldsymbol{\beta}}_i - \boldsymbol{\beta}_{i})}^2 \le \frac{ C_4^2 \|\boldsymbol{\omega}_{\mathcal{S}_{i}}\|_{\infty}^2 \lonen{\mathbf{B}}^2\lonen{\boldsymbol{\Pi}}^2}{\boldsymbol{\phi}_0^2\|\boldsymbol{\omega}_{i}\|_{-\infty}^2}   |\mathcal{S}_{i}|  \frac{(d \lor r_{\max}\lor f_{\max})\log(p)}{n_{\min}}.
\end{equation}

Taking this inequality, we can follow Assumption~4 and Lemma~\ref{lemma:restrictedEigenValue} to derive that
\begin{eqnarray} \label{inequality:derive_estimation}
	\lonen{\hat{\boldsymbol{\beta}}_{i} -\boldsymbol{\beta}_{i}} &=&\lonen{\hat{\boldsymbol{\beta}}_{\mathcal{S}_{i}^c}}+\lonen{\hat{\boldsymbol{\beta}}_{\mathcal{S}_{i}} - \boldsymbol{\beta}_{\mathcal{S}_{i}}}\le \left(3\frac{\|\boldsymbol{\omega}_{\mathcal{S}_{i}}\|_{\infty}}{\|\boldsymbol{\omega}_{\mathcal{S}_{i}^{c}}\|_{-\infty }} +1\right) \lonen{\hat{\boldsymbol{\beta}}_{\mathcal{S}_{i}} - \boldsymbol{\beta}_{\mathcal{S}_{i}}}\\
	&\le& \left(3\frac{\|\boldsymbol{\omega}_{\mathcal{S}_{i}}\|_{\infty}}{\|\boldsymbol{\omega}_{\mathcal{S}_{i}^{c}}\|_{-\infty }}  +1\right)\sqrt{|\mathcal{S}_{i}|} \frac{2 \ltwon{\mathbf{H}_i\hat{\mathbf{Z}}_{-i}(\hat{\boldsymbol{\beta}}_i - \boldsymbol{\beta}_{i})}}{\sqrt{n}\boldsymbol{\phi}_0} \nonumber\\
	&\le& \left(3\frac{\|\boldsymbol{\omega}_{\mathcal{S}_{i}}\|_{\infty}}{\|\boldsymbol{\omega}_{\mathcal{S}_{i}^{c}}\|_{-\infty }}  +1\right)\sqrt{|\mathcal{S}_{i}|}
	\frac{2C_4 \|\boldsymbol{\omega}_{\mathcal{S}_{i}}\|_{\infty} \lonen{\mathbf{B}} \lonen{\boldsymbol{\Pi}}}{\boldsymbol{\phi}_0^2 \|\boldsymbol{\omega}_{i}\|_{-\infty}} \sqrt{|\mathcal{S}_{i}|}  \sqrt{\frac{(d \lor r_{\text{max}}\lor f_{\max})\log(p)}{n_{\min}}} \nonumber\\
	&\le& 8C_4 \frac{ \|\boldsymbol{\omega}_{\mathcal{S}_{i}}\|_{\infty}  \lonen{\mathbf{B}} \lonen{\boldsymbol{\Pi}}}{\boldsymbol{\phi}_0^2 \|\boldsymbol{\omega}_{i}\|_{-\infty}} |\mathcal{S}_{i}| \sqrt{\frac{(d \lor r_{\max}\lor f_{\max})\log(p)}{n_{\min}}},
\end{eqnarray}
where the last inequality follows Assumption 4. Since the inequality~(\ref{inequality:remove_E}) concurs with the event $\mathscr{E}(\lambda_i)$, the above prediction and estimation bounds hold with probability at least $1- 3e^{-C_3 h_n + \log(4pq)} - e^{-f^{(1)}+\log(p)} - e^{-f^{(2)} + \log(p)}$. This completes the proof of Theorem 3.

\section{Proof of Theorem 4} \label{sec:thm4}

\begin{lemma} \label{lemma:predictedIrrepresentable}
	Suppose that, for node $i$,
	\begin{equation} \label{ineq:gn2}
		\sqrt{(d \lor r_{\max}\lor f_{\max}) \big/ n} + c_{\max}\lonen{\boldsymbol{\Pi}}\le \sqrt{c_{\max}^2\lonen{\boldsymbol{\Pi}}^2+\min(\phi_0^2\big/64, \tau(4-\tau)^{-1} \|\boldsymbol{\omega}_{i}\|_{-\infty}\big/\psi_i)\big/(C_2 |\mathcal{S}_{i}|)}.
	\end{equation}
	Under Assumptions 1-5, we have $\infn{W^{-1}_{\mathcal{S}^c_i}(\hsto \hsoo[-1])W_{\mathcal{S}_{i}}}\le 1- \tau/2$ with the probability at least $1- e^{-f^{(1)}+\log(p)}- e^{-f^{(2)}+\log(p)}$.
\end{lemma}

\begin{proof}[\textbf{Proof of Lemma \ref{lemma:predictedIrrepresentable}}]
	The inequality (\ref{ineq:gn2}) implies that $\psi_{i} \|\boldsymbol{\omega}_i\|_{-\infty}^{-1} |\mathcal{S}_i| g_n \le  \frac{\tau}{4 - \tau}$.
	
	By the inequalities (\ref{inequality:Delta123}) and (\ref{inequality:Delta_all}) in the proof of Lemma~\ref{lemma:restrictedEigenValue} and union bound, we have that, with probability at least $1- e^{-f^{(1)}+\log(p)}- e^{-f^{(2)}+\log(p)}$,
	\begin{equation*}
		\underset{ j_{1},j_{2}  }{\text{max}} \left \{  \frac{1}{n}|(\mathbf{H}_i\mathbf{X}\hat{\boldsymbol{\Pi}}_{j_{1}})^T  (\mathbf{H}_i\mathbf{X}\hat{\boldsymbol{\Pi}}_{j_{2}}) -  (\mathbf{H}_i\mathbf{X} \boldsymbol{\Pi}_{j_{1}} )^T (\mathbf{H}_i\mathbf{X} \boldsymbol{\Pi}_{j_{2}} )| \right\} \le g_n.
	\end{equation*}
	With the definitions of infinity norm $\infn{\cdot}$, $\hsoo$, and $\soo$, we can obtain the following inequality indexed by set $\mathcal{S}_{i}$,
	\begin{eqnarray}\label{inequality:wdiffI_S}
		\psi_{i}\infn{W_{\mathcal{S}_{i}}^{-1}(\hsoo - \soo)} \le \psi_{i}  \|\boldsymbol{\omega}_{\mathcal{S}_i}\|_{-\infty}^{-1} \infn{\hsoo - \soo}
		\le \psi_{i} \|\boldsymbol{\omega}_{\mathcal{S}_i}\|_{-\infty}^{-1} |\mathcal{S}_i| g_n \le  \frac{\tau}{4 - \tau}.
	\end{eqnarray}
	Similarly we can obtain the following bound indexed by the complement set $\mathcal{S}^{c}_{i}$,
	\begin{equation} \label{ineqality:wdiffI_S_c}
		\psi_{i}\infn{W_{\mathcal{S}^c_i}^{-1}(\hsto[] - \sto)} \le \psi_{i} \|\boldsymbol{\omega}_{\mathcal{S}_i^c}\|_{-\infty}^{-1} |\mathcal{S}_i| g_n \le  \frac{\tau}{4 - \tau}.
	\end{equation}	
	Applying the matrix inversion error bound in Horn and Johnson [2012] and the triangular inequality, we have that 
	\begin{eqnarray} \label{ineqn:hiiw}
		\infn{\hsoo[-1] W_{\mathcal{S}_{i}}} &\le& \infn{\soo[-1] W_{\mathcal{S}_{i}} } + \infn{\hsoo[-1] W_{\mathcal{S}_{i}} -\soo[-1]W_{\mathcal{S}_{i}}}   \nonumber\\
		&\le& \psi_{i}  + \frac{\psi_{i}\infn{W_{\mathcal{S}_{i}}^{-1}(\hsoo[] - \soo)}}{1- \psi_{i}\infn{W_{\mathcal{S}_{i}}^{-1}(\hsoo[] - \soo)}}\psi_{i} \le \psi_{i} + \frac{\tau}{4-2\tau}\psi_{i} \le \frac{4-\tau}{4-2\tau} \psi_{i}.
	\end{eqnarray}
	
	Also note that we can rewrite
	\begin{eqnarray*}
		\lefteqn{W^{-1}_{\mathcal{S}_{i}^c}\left(\hsto \hsoo[-1]  - \sto \soo[-1]\right)W_{\mathcal{S}_{i}}} \nonumber\\
		&=& W^{-1}_{\mathcal{S}_{i}^c}\left(\hsto - \sto \right) \hsoo[-1] W_{\mathcal{S}_{i}} + W^{-1}_{\mathcal{S}_{i}^c}\sto\soo[-1] W_{\mathcal{S}_{i}} W_{\mathcal{S}_{i}}^{-1} \left(\hsoo - \soo \right) \hsoo[-1] W_{\mathcal{S}_{i}}.
	\end{eqnarray*}
	Then, it follows from (\ref{inequality:wdiffI_S}), (\ref{ineqality:wdiffI_S_c}), (\ref{ineqn:hiiw}) and Assumption~5 that
	\begin{eqnarray*}
		\lefteqn{\infn{W^{-1}_{\mathcal{S}_{i}^c}\left(  \hsto \hsoo[-1]  - \sto \soo[-1]  \right)W_{\mathcal{S}_{i}}}} \nonumber\\
		&\le& \infn{W^{-1}_{\mathcal{S}_{i}^c}\left( \hsto - \sto \right)}\infn{\hsoo[-1] W_{\mathcal{S}_{i}}} \nonumber\\
		&& + \infn{W^{-1}_{\mathcal{S}_{i}^c}\sto\soo[-1]W_{\mathcal{S}_{i}}} \infn{W_{\mathcal{S}_{i}}^{-1}\left(\hsoo-\soo  \right)}\infn{\hsoo[-1] W_{\mathcal{S}_{i}}} \le \tau/2.
	\end{eqnarray*}
	Therefore, together with Assumption~5 again, we can conclude that $\infn{W^{-1}_{\mathcal{S}_{i}^c}(\hsto \hsoo[-1]) W_{\mathcal{S}_{i}}}\le 1- \tau/2$.
	
	This concludes the proof of Lemma~\ref{lemma:predictedIrrepresentable}.
\end{proof}

The optimality of $\hat{\boldsymbol{\beta}}_i$ in the adaptive lasso step and KKT condition lead to
\begin{equation} \label{equ-iitraw}
	-\frac{2}{n} (\mathbf{H}_{i}\hat{\mathbf{Z}}_{-i})^T(\mathbf{H}_{i}\mathbf{Y}_i - \mathbf{H}_{i}\hat{\mathbf{Z}}_{-i}\hat{\boldsymbol{\beta}}_{i}) + \lambda_{i} W_{i} \alpha_{i} =0,
\end{equation}
where $\alpha_{i} \in \mathbb{R}^{2p-2}$, satisfying that $\infn{\alpha_{i}} \le 1$ and $\alpha_{ij}I(\hat{\boldsymbol{\beta}}_{ij}\ne 0) = sign(\hat{\boldsymbol{\beta}}_{ij})$.

Plug in the equation $\mathbf{H}_{i}\bY_{i} = \mathbf{H}_{i}\mathbf{Z}_{-i} \boldsymbol{\beta}_{i} + \mathbf{H}_i\boldsymbol{\epsilon}_i$, we can have that
\begin{eqnarray} \label{equ-errorexpand}
	\mathbf{H}_{i}\mathbf{Y}_{i} - \mathbf{H}_{i} \hat{\mathbf{Z}}_{-i} \hat{\boldsymbol{\beta}}_{i} &=&\mathbf{H}\mathbf{Z}_{-i}\boldsymbol{\beta}_{i} +\mathbf{H}_{i}\boldsymbol{\epsilon}_i - \mathbf{H}_{i} \hat{\mathbf{Z}}_{-i}\hat{\boldsymbol{\beta}}_{i}  \nonumber\\
	&=&\mathbf{H}_{i}\boldsymbol{\epsilon}_i + \mathbf{H}_{i}\mathbf{Z}_{-i}\boldsymbol{\beta}_{i} - \mathbf{H}_{i}\hat{\mathbf{Z}}_{-i}\boldsymbol{\beta}_{i} + \mathbf{H}_{i}\hat{\mathbf{Z}}_{-i}\boldsymbol{\beta}_{i} - \mathbf{H}_{i} \hat{\mathbf{Z}}_{-i}\hat{\boldsymbol{\beta}}_{i}   \nonumber\\
	&=&\mathbf{H}_{i} \boldsymbol{\epsilon}_{i} - \mathbf{H}_{i}(\hat{\mathbf{Z}}_{-i} - \mathbf{Z}_{-i})\boldsymbol{\beta}_{i} -\mathbf{H}_{i}\hat{\mathbf{Z}}_{-i}(\hat{\boldsymbol{\beta}}_{i}-\boldsymbol{\beta}_{i}).
\end{eqnarray}
This, along with KKT condition~(\ref{equ-iitraw}), leads to
\begin{equation} \label{equ-KKT}
	\begin{aligned}
		2\hat{\mathcal{I}}_{i}(\hat{\boldsymbol{\beta}}_{i} - \boldsymbol{\beta}_{i}) -\boldsymbol{\eta}_i = -\lambda_{i} W_{i} \alpha_{i},
	\end{aligned}
\end{equation}
where $\boldsymbol{\eta}_i$ is defined in Lemma~\ref{lemma:BI}.

Letting $\hat{\boldsymbol{\beta}}_{\mathcal{S}_{i}^{c}} = \boldsymbol{\beta}_{\mathcal{S}_{i}^{c}} = 0$, equation~(\ref{equ-KKT}) can be decomposed as
\begin{equation} \label{equ-Aisplit}
	\begin{cases}
		\begin{aligned}
			2\hat{\mathcal{I}}_{i,11}(\hat{\boldsymbol{\beta}}_{\mathcal{S}_{i}} - \boldsymbol{\beta}_{\mathcal{S}_{i}}) -\boldsymbol{\eta}_{\mathcal{S}_{i}} &= -\lambda_{i}W_{\mathcal{S}_{i}} \alpha_{\mathcal{S}_{i}}, \\
			2\hat{\mathcal{I}}_{i,21}(\hat{\boldsymbol{\beta}}_{\mathcal{S}_{i}} - \boldsymbol{\beta}_{\mathcal{S}_{i}}) -\boldsymbol{\eta}_{\mathcal{S}_{i}^{c}} & = -\lambda_{i}W_{\mathcal{S}_{i}^{c}  } \alpha_{\mathcal{S}_{i}^{c}}.
		\end{aligned}
	\end{cases}
\end{equation}
We can solve for $\hat{\boldsymbol{\beta}}_{\mathcal{S}_{i}}$ from the first equation of (\ref{equ-Aisplit}) as
\begin{eqnarray} \label{equ-gammadiff}
	\hat{\boldsymbol{\beta}}_{\mathcal{S}_{i}} - \boldsymbol{\beta}_{\mathcal{S}_{i}} &=& 2^{-1}  \hat{\mathcal{I}}_{i,11}^{-1}(\boldsymbol{\eta}_{\mathcal{S}_{i}} -\lambda_{i}W_{\mathcal{S}_{i}}^T \alpha_{\mathcal{S}_{i}}) = 2^{-1} \hat{\mathcal{I}}_{i,11}^{-1} W_{ \mathcal{S}_{i}}   (W_{ \mathcal{S}_{i}}^{-1} \boldsymbol{\eta}_{\mathcal{S}_{i}} -\lambda_{i} \alpha_{\mathcal{S}_{i}}).
\end{eqnarray}

Following the similar strategy in the proof of Lemma~\ref{lemma:BI}, we can prove that there exists a constant $C_5 > 0$ such that $\infn{W_i^{-1}\boldsymbol{\eta}_i} \le \frac{\tau}{4-\tau}\lambda_{i}$ with probability at least $1- 3e^{-C_5\,h_n + \text{log}\,(4q) + \text{log}\,(p)}- e^{-f^{(1)}+\text{log}\,(p)} - e^{-f^{(2)} + \text{log}\,(p)}$. Thus, together with $\infn{\alpha_{\mathcal{S}_{i}}} \le 1$, we obtain the infinity norm estimation loss on the true support set $\mathcal{S}_{i}$
\begin{eqnarray*}
	\infn{ \hat{\boldsymbol{\beta}}_{\mathcal{S}_{i}} - \boldsymbol{\beta}_{\mathcal{S}_{i}}} &\le& 2^{-1} \infn{ \hat{\mathcal{I}}_{i,11}^{-1} W_{ \mathcal{S}_{i}}  } ( \infn{W_{ \mathcal{S}_{i}}^{-1} \boldsymbol{\eta}_{\mathcal{S}_{i}} } + \lambda_{i}  ) \nonumber\\
	&\le& 2^{-1} \frac{4-\tau}{4-2\tau} \psi_{i} \frac{4}{4-\tau} \lambda_{i} = \frac{\lambda_{i}\psi_{i}}{2-\tau} \le \underset{j\in \mathcal{S}_{i}}{\text{min}} |\boldsymbol{\beta}_{ij}|=b_{i},
\end{eqnarray*}
where the last inequality comes from the condition on the minimal signal strength $b_{i}$. The above inequality implies $sign( \hat{\boldsymbol{\beta}}_{\mathcal{S}_{i}}) = sign( \boldsymbol{\beta}_{\mathcal{S}_{i}})$.

Plugging (\ref{equ-gammadiff}) into the left hand side of the second equation in (\ref{equ-Aisplit}), we can verify that
\begin{eqnarray*}
	\lefteqn{\infn{ W_{\mathcal{S}_{i}^{c}  }^{-1}\hat{\mathcal{I}}_{i,21} (\hat{\mathcal{I}}_{i,11})^{-1}(\boldsymbol{\eta}_{\mathcal{S}_{i}}-\lambda_{i}W_{\mathcal{S}_{i}} \alpha_{\mathcal{S}_{i}}) -  W_{\mathcal{S}_{i}^{c}  }^{-1} \boldsymbol{\eta}_{\mathcal{S}_{i}^{c}} }} \nonumber\\
	&\le& \infn{ W_{\mathcal{S}_{i}^{c}  }^{-1} \hat{\mathcal{I}}_{i,21} \hat{\mathcal{I}}_{i,11}^{-1} W_{\mathcal{S}_{i}}} (\infn{ W^{-1}_{\mathcal{S}_{i}} \boldsymbol{\eta}_{\mathcal{S}_{i}}} +\lambda_{i}) + \infn{ W_{\mathcal{S}_{i}^{c} }^{-1} \boldsymbol{\eta}_{\mathcal{S}_{i}^{c}}}   \nonumber\\
	&\le& (1-\tau/2)(4/(4-\tau))\lambda_{i} + \tau/(4-\tau) \lambda_{i} = \lambda_{i}.
\end{eqnarray*}
Therefore, we have constructed a solution $\hat{\boldsymbol{\beta}}_i$ which satisfies the KKT condition~(\ref{equ-KKT}) and $sign(\hat{\boldsymbol{\beta}}_i) = sign(\boldsymbol{\beta}_i)$, that is, $\hat{\mathcal{S}}_i = \mathcal{S}_i$. This completes the proof of Theorem 4.

\section*{References}
\begin{itemize}[leftmargin=*]
	\item[] Jianqing Fan and Jinchi Lv. Sure independence screening for ultrahigh dimensional feature space. \textit{Journal of the Royal Statistical Society: Series B (Statistical Methodology)}, 70(5): 849–911, 2008.
	\item[] Roger A Horn and Charles R Johnson. \textit{Matrix Analysis}. Cambridge University Press, 2012.
	\item[]  Mark Rudelson, Roman Vershynin, et al. Hanson-wright inequality and sub-gaussian concentration. \textit{Electronic Communications
	in Probability}, 18, 2013.
\end{itemize}

\end{document}